%% file: ejs.tex
\documentclass[preprint,11pt]{imsart}

\RequirePackage{amsthm,amsmath,amsfonts,amssymb}
\RequirePackage[authoryear]{natbib}
\RequirePackage[overload]{textcase}
\RequirePackage{graphicx}
\RequirePackage[margin=3.5cm]{geometry}

\RequirePackage{enumerate}
\RequirePackage{grffile}
\RequirePackage{booktabs,rotating}
\RequirePackage{bm}

\startlocaldefs

\numberwithin{equation}{section}
\numberwithin{table}{section}
\numberwithin{figure}{section}

\theoremstyle{plain}
\newtheorem{prop}{Proposition}[section]

\newtheorem{thm}[prop]{Theorem}

\newtheorem{lem}[prop]{Lemma}

\theoremstyle{remark}
\newtheorem{remark}[prop]{Remark}

\newcommand{\eps}{\varepsilon}
\newcommand{\N}{\mathbb{N}}
\newcommand{\R}{\mathbb{R}}
\newcommand{\Z}{\mathbb{Z}}
\newcommand{\dd}{\mathrm{d}}
\newcommand{\Bb}{\mathbb{B}}

\newcommand{\Db}{\mathbb{D}}
\newcommand{\Fb}{\mathbb{F}}
\newcommand{\Gb}{\mathbb{G}}
\newcommand{\Hb}{\mathbb{H}}
\newcommand{\Jb}{\mathbb{J}}
\newcommand{\Kb}{\mathbb{K}}
\newcommand{\Rb}{\mathbb{R}}
\newcommand{\Sb}{\mathbb{S}}
\newcommand{\Tb}{\mathbb{T}}

\newcommand{\Cc}{\mathcal{C}}

\newcommand{\Fc}{\mathcal{F}}
\newcommand{\Lc}{\mathcal{L}}
\newcommand{\Sc}{\mathcal{S}}
\newcommand{\Xc}{\mathcal{X}}
\newcommand{\Wc}{\mathcal{W}}

\def\disp{\displaystyle}

\newcommand{\Ex}{\mathbb{E}}
\newcommand{\Var}{\mathrm{Var}}
\newcommand{\Cov}{\mathrm{Cov}}
\newcommand{\1}{\mathbf{1}}
\newcommand{\ip}[1]{\lfloor #1 \rfloor}
\renewcommand{\Pr}{\mathbb{P}}

\newcommand{\p}{\overset{\Pr}{\to}}
\newcommand{\as}{\overset{as}{\to}}

\endlocaldefs

\begin{document}

\begin{frontmatter}


\title{Nonparametric sequential change-point detection for multivariate time series based on empirical distribution functions}
\runtitle{Sequential change-point detection based on empirical d.f.s}
  
\begin{aug}

%%%%%%%%%%%%%%%%%%%%%%%%%%%%%%%%%%%%%%%%%%%%%%
%%Only one address is permitted per author. %%
%%Only division, organization and e-mail is %%
%%included in the address.                  %%
%%Additional information can be included in %%
%%the Acknowledgments section if necessary. %%
%%%%%%%%%%%%%%%%%%%%%%%%%%%%%%%%%%%%%%%%%%%%%%

\author[A]{\fnms{Ivan} \snm{Kojadinovic}\ead[label=e1,mark]{ivan.kojadinovic@univ-pau.fr}}
\and
\author[A]{\fnms{Ghislain} \snm{Verdier}\ead[label=e2]{ghislain.verdier@univ-pau.fr}}

% %%%%%%%%%%%%%%%%%%%%%%%%%%%%%%%%%%%%%%%%%%%%%%
% %% Addresses                                %%
% %%%%%%%%%%%%%%%%%%%%%%%%%%%%%%%%%%%%%%%%%%%%%%

\address[A]{CNRS / Universit\'e de Pau et des Pays de l'Adour / E2S UPPA, Laboratoire de math\'ematiques et applications, IPRA, UMR 5142, B.P. 1155, 64013 Pau Cedex, France, \printead{e1,e2}}

\runauthor{Kojadinovic and Verdier}

\end{aug}

\begin{abstract}
The aim of sequential change-point detection is to issue an alarm when it is thought that certain probabilistic properties of the monitored observations have changed. This work is concerned with nonparametric, closed-end testing procedures based on differences of empirical distribution functions that are designed to be particularly sensitive to changes in the contemporary distribution of multivariate time series. The proposed detectors are adaptations of statistics used in a posteriori (offline) change-point testing and involve a weighting allowing to give more importance to recent observations. The resulting sequential change-point detection procedures are carried out by comparing the detectors to threshold functions estimated through resampling such that the probability of false alarm remains approximately constant over the monitoring period. A generic result on the asymptotic validity of such a way of estimating a threshold function is stated. As a corollary, the asymptotic validity of the studied sequential tests based on empirical distribution functions is proven when these are carried out using a dependent multiplier bootstrap for multivariate time series. Large-scale Monte Carlo experiments demonstrate the good finite-sample properties of the resulting procedures. The application of the derived sequential tests is illustrated on financial data.
\end{abstract}



\tableofcontents

\end{frontmatter}

%%%%%%%%%%%%%%%%%%%%%%%%%%%%%%%%%%%%%%%%%%%%%%
%% Please use \tableofcontents for articles %%
%% with 50 pages and more                   %%
%%%%%%%%%%%%%%%%%%%%%%%%%%%%%%%%%%%%%%%%%%%%%%

% \tableofcontents

%%%%%%%%%%%%%%%%%%%%%%%%%%%%%%%%%%%%%%%%%%%%%%
%%%% Main text entry area:

\section{Introduction}

Let $\bm X_1,\dots,\bm X_m$, $m \geq 1$,  be a stretch from a $d$-dimensional stationary times series of continuous random vectors with unknown contemporary distribution function (d.f.) $F$ given by $F(\bm x) = \Pr(\bm X_1 \leq \bm x)$, $\bm x \in \R^d$. These available observations will be referred to as the \emph{learning sample} as we continue. The context of this work is that of \emph{sequential change-point detection}: new observations $\bm X_{m+1},\bm X_{m+2},\dots$ arrive sequentially and we wish to issue an alarm as soon as possible if the contemporary distribution of the most recent observations is not equal to $F$ anymore. If there is no evidence of a change in the contemporary distribution, the monitoring stops after the arrival of observation $\bm X_n$ for some $n > m$.

The theoretical framework of our investigations is that adopted in the seminal paper of \cite{ChuStiWhi96}. Unlike classical approaches in statistical process control (SPC) usually calibrated in terms of \emph{average run length} (ARL) \citep[see, e.g.,][for an overview]{Lai01,Mon07} and leading in general to the rejection of the underlying null hypothesis of stationarity with probability one, the approach of \cite{ChuStiWhi96} guarantees that, asymptotically (that is, as the size $m$ of the learning sample tends to infinity), stationarity, if it holds, will only be rejected with a small probability $\alpha$ to be interpreted as a type I error and thus called the \emph{probability of false alarm}. The latter paradigm is increasingly considered in the literature; see, e.g., \cite{HorHusKokSte04}, \cite{AueHorKuhSte12}, \cite{AueHorHorHusSte12}, \cite{Fre15}, \cite{KirWeb18}, \cite{DetGos19} or \cite{KirSto19}.

Among the approaches \emph{à la} \cite{ChuStiWhi96}, one can distinguish between \emph{closed-end} and \emph{open-end} procedures. The latter can in principle continue indefinitely if no evidence against the null is observed. Our approach is of the former type in the sense that at most $n - m$ new observations will be considered before the monitoring stops.

As already mentioned, the null hypothesis of the procedure that we shall investigate is that of stationarity and can be more formally stated as follows:
\begin{equation}
  \label{eq:H0}
  \begin{split}
    H_0: \, &\bm X_1,\dots,\bm X_m,\bm X_{m+1},\dots, \bm X_n \text{ is a stretch from a stationary time series} \\
    &\text{with contemporary d.f.\ } F.
  \end{split}
\end{equation}
Notice that, if one is additionally ready to assume the independence of the observations, $H_0$ simplifies to
\begin{equation}
  \label{eq:H0:ind}
  H_0^{\text{ind}}: \, \bm X_1,\dots,\bm X_m,\bm X_{m+1},\dots, \bm X_n \text{ are independent random vectors with d.f.\ } F.
\end{equation}
The aim of this work is to derive nonparametric sequential change-point detection procedures particularly sensitive to the alternative hypothesis
\begin{equation}
  \label{eq:H1}
  \begin{split}
    H_1: \, &\exists \, k^\star \in \{m,\dots,n-1\} \text{ such that }  \bm X_1,\dots,\bm X_{k^\star} \text{ is a stretch from a stationary} \\
    &\text{time series with contemporary d.f.\ } F \text { and } \bm X_{k^\star+1},\dots,\bm X_{n} \text{ is a stretch} \\
    & \text{from a stationary time series with contemporary d.f.\ } G \neq F.
  \end{split}
\end{equation}
In other words, unlike some of the approaches reported for instance in \cite{KirWeb18}, \cite{DetGos19} or \cite{KirSto19}, we are not solely interested in being sensitive to a change in a given parameter of the $d$-dimensional time series such as the mean vector or the covariance matrix. We aim at deriving nonparametric monitoring procedures that, in principle, provided $m$ and $n$ are large enough, can detect all types of changes in the contemporary d.f.

In the considered context, the two main ingredients of a sequential change-point detection procedure are a sequence of positive statistics $D_m(k)$, $k \in \{m+1,\dots,n\}$, and a sequence of suitably chosen strictly positive thresholds $w_m(k)$, $k \in \{m+1,\dots,n\}$. For $k \in \{m+1,\dots,n\}$, the statistic $D_m(k)$ (called a \emph{detector} in the literature) is used to assess a possible departure from $H_0$ in~\eqref{eq:H0} using only observations $\bm X_1, \dots, \bm X_k$ and is such that the larger $D_m(k)$, the more evidence against $H_0$. After the arrival of observation $\bm X_k$, $k \in \{m+1,\dots,n\}$, the detector $D_m(k)$ is computed and compared to the threshold $w_m(k)$. If $D_m(k) > w_m(k)$, the available evidence against stationarity is considered to be large enough and an alarm is issued resulting in the monitoring to stop. If $D_m(k) \leq w_m(k)$ and $k < n$, a new observation $\bm X_{k+1}$ is collected and the previous iteration is repeated. This monitoring process can be naturally represented by a graph illustrating the evolution of the sequence of detectors against the sequence of thresholds; see, e.g., Figures~\ref{fig:det-thresh} and~\ref{fig:biv-det-thresh}. To ensure that the monitoring can be interpreted as a testing procedure, the sequence of thresholds $w_m(k)$, $k \in \{m+1,\dots,n\}$, needs to be chosen such that, under stationarity,
\begin{equation}
  \label{eq:neg:typeI}
\Pr \{ D_m(m+1) \leq w_m(m+1), \dots, D_m(n) \leq w_m(n)  \} \geq 1 - \alpha,
\end{equation}
for some small significance level $\alpha$.

The detectors $D_m(k)$, $k \in \{m+1,\dots,n\}$, proposed in this work are defined from differences of empirical distribution functions. Specifically, as shall be explained in detail in Section~\ref{sec:detectors}, particularly powerful monitoring procedures can be obtained by defining $D_m(k)$ as a maximized, suitably normalized difference of empirical distribution functions computed from $\bm X_1,\dots, \bm X_j$ and $\bm X_{j+1},\dots, \bm X_k$, respectively, where the maximum is taken over all $j \in \{m,\dots,k-1\}$. When sequentially investigating changes in real-valued parameters of multivariate time series, such an approach (related to what was called a Page-CUSUM procedure in \cite{KirWeb18} as a consequence of the work of \cite{Fre15}) was recently suggested in \cite{DetGos19} and justified through a likelihood ratio approach. The latter principle had actually already been considered in the SPC literature (to the best of our knowledge, without any asymptotic theory, however) since at least the work of \cite{HawQiuKan03}; see also \cite{HawZam05}, \cite{RosTasAda11} and \cite{RosAda12}. As we shall see, our detectors can be cast into the framework of \cite{DetGos19} with the difference that they involve a weighting allowing to give more or less importance to recent observations. In that sense, they could also be regarded as an adaptation of the statistics considered in \cite{CsoSzy94b} for a posteriori (offline) change-point detection to the sequential setting considered in this work.

As far as the thresholds $w_m(k)$, $k \in \{m+1,\dots,n\}$, are concerned, unlike in the recent literature in which these are calibrated through parametric functions defined up to a positive multiplicative constant such that~\eqref{eq:neg:typeI} holds \citep[see, e.g.,][]{KirWeb18,DetGos19,KirSto19}, we define them through simulations or resampling with the aim that, under $H_0$ in~\eqref{eq:H0}, the probability of rejection of $H_0$ be roughly the same at every step $k \in \{m+1,\dots,n\}$ of the procedure. In the considered sequential testing framework \emph{à la} \cite{ChuStiWhi96}, the latter requirement had actually already been proposed by \cite{AnaKos18}. Interestingly enough, as we shall see, it is strongly related to an approach appearing in SPC in \cite{MarConWooDra95}; see also, for instance, \cite{HawZam05}, \cite{Ros14}, \cite{RosAda12} and the \textsf{R} package \texttt{cpm} described in \cite{Ros15}.

From the point of view of their main ingredients, the sequential change-point tests studied in this work are related to the approach put forward in SPC in \cite{RosAda12} (without any asymptotic theory) and implemented in the \textsf{R} package {\tt cpm} \citep{Ros15}. However, unlike the latter, they are not restricted to univariate observations and they can deal with time series, among several other differences. To achieve this, as soon as the learning data set $\bm X_1,\dots,\bm X_m$ is multivariate ($d > 1$) or serially dependent, resampling, under the form of a \emph{dependent multiplier bootstrap} \emph{à la} \cite{BucKoj16}, is used to carry out the monitoring procedures. By adapting the theoretical framework considered in \cite{DetGos19}, the asymptotic validity of the investigated procedures is proven under strong mixing both under $H_0$ in~\eqref{eq:H0} and under sequences of alternatives related to $H_1$ in~\eqref{eq:H1}.

This paper is organized as follows. In the second section, we propose three classes of detectors based on differences of empirical d.f.s and study their asymptotics, both under $H_0$ and sequences of alternatives related to $H_1$. A more general perspective is adopted in the third section: for an arbitrary detector in the considered closed-end setting, a procedure for estimating the threshold function such that the probability of false alarm remains approximately constant over the monitoring period is investigated and its asymptotic validity is proven under both~$H_0$ and sequences of alternatives related to~$H_1$ when the estimation is based on an asymptotically valid resampling scheme. The results of the third section are applied in the fourth section to the three proposed classes of detectors based on differences of empirical d.f.s. To do so, the consistency of a dependent multiplier bootstrap for the detectors is proven under strong mixing. The fifth section presents a summary of large-scale Monte-Carlo experiments demonstrating the good finite-sample properties of the resulting sequential testing procedures. An application on real financial data concludes the article.

Auxiliary results and proofs are deferred to a sequence of appendices. A preliminary implementation of the studied tests is available in the package {\tt npcp} \citep{npcp} for the \textsf{R} statistical system \citep{Rsystem}.

%%%%%%%%%%%%%%%%%%%%%%%%%%%%%%%%%%%%%%%%%%%%%%%%%%%%%%%%%%%%%%%%%%%%%%%%%%%%%%%%%%%%%%

\section{The detectors and their asymptotics}
\label{sec:detectors}

After defining three classes of detectors based on empirical distribution functions and noticing that they are margin-free under $H_0$ in~\eqref{eq:H0}, we study their asymptotics under the null and sequences of alternatives related to $H_1$ in~\eqref{eq:H1}.

\subsection{Detectors based on empirical distribution functions}
\label{sec:det:emp}

Let $F_{j:k}$ be the empirical d.f.\ computed from the stretch $\bm X_j,\dots,\bm X_k$ of available observations. More formally, for any integers $j,k \geq 1$ and $\bm x \in \R^d$, let
\begin{equation}
  \label{eq:Fjk}
  F_{j:k}(\bm x) =
  \left\{
    \begin{array}{ll}
      \disp \frac{1}{k-j+1}\sum_{i=j}^k \1(\bm X_i \leq \bm x), \qquad & \text{if } j \leq k, \\
      0, \qquad &\text{otherwise},
    \end{array}
  \right.
\end{equation}
where the inequalities $\bm X_i \leq \bm x$ are to be understood componentwise and $\1(\bm X_i \leq \bm x)$ is equal to 1 (resp.\ 0) if all (resp.\
some) of the $d$ underlying inequalities are true (resp.\ false).

After the $k$th observation has arrived, the available data take the form of the stretch $\bm X_1,\dots,\bm X_k$. If we were in the context of a posteriori change-point detection, a prototypical test statistic would be the maximally selected Kolmogorov--Smirnov-type statistic
\begin{equation}
  \label{eq:Rk}
R_k = \max_{1 \leq j \leq k-1}  \frac{j (k - j)}{k^{3/2}} \sup_{\bm x \in \R^d} | F_{1:j}(\bm x) - F_{j+1:k}(\bm x) |,
\end{equation}
practically considered for instance in \citet{GomHor99} or \cite{HolKojQue13}. The intuition behind $R_k$ is the following: every $j \in \{1,\dots,k-1\}$ is treated as a potential break point in the sequence and the maximum in~\eqref{eq:Rk} implies that $R_k$ will be large as soon as the difference between $F_{1:j}$ and $F_{j+1:k}$ is large for some $j$. The weighting $j (k - j) / k^{3/2}$  ensures that $R_k$ converges in distribution under stationarity as $k \to \infty$. As explained for instance in \cite{CsoSzy94b} in the case of independent observations, replacing for instance these weights simply by $\sqrt{k}$ would result in a statistic that diverges in probability to $\infty$ under stationarity. The part $j (k - j)$ in the weighting favors however the detection of potential break points in the middle of the sequence. Test statistics that are more sensitive to changes at the beginning or at the end of the sequence but still converge in distribution under stationarity can be obtained by considering weights of the form $j (k - j) / \{ k^{3/2} q(j/k) \}$ for some suitable strictly positive function $q$ on $(0,1)$; see, e.g., \citet{CsoSzy94,CsoSzy94b}, \cite{CsoHorSzy97} and \cite{CsoHor97}.

% \begin{remark}
%  \label{rem:approx:sup}
% When $d = 1$, some thought reveals that, because of~\eqref{eq:Fjk}, the supremum over~$\bm x \in \R$ in~\eqref{eq:Rk} is actually a maximum over $\{\bm X_1,\dots,\bm X_k\}$. This is not the case anymore when $d > 1$. Still, in practice, it is necessary to compute the supremum and a natural way of proceeding in the multivariate case is to approximate it by a maximum over $\{\bm X_1,\dots,\bm X_k\}$. We shall proceed accordingly when using similar statistics in the case of multivariate observations.
% \end{remark}

Going back to the setting of sequential change-point detection considered in this work, a first meaningful modification of~\eqref{eq:Rk} is to restrict the maximum over $j$ to $j \in \{m,\dots,k-1\}$ since a change cannot occur at the beginning of the sequence given that $\bm X_1, \dots, \bm X_m$ is the learning sample known to be a stretch from a stationary sequence. Another modification is a rescaling consisting of replacing $k^{3/2}$ by $m^{3/2}$ in the weighting. The latter is made solely for asymptotic reasons as shall become clear in Sections~\ref{sec:asym:det:H0} and~\ref{sec:asym:det:H1}. These modifications essentially lead to the first detectors considered in this work:
\begin{equation}
  \label{eq:Rmq}
  R_{m,q}(k) =  \max_{m \leq j \leq k-1} \frac{j (k-j)}{m^{3/2} q(j/m,k/m)}  \sup_{\bm x \in \R^d} | F_{1:j}(\bm x) - F_{j+1:k}(\bm x) |, \qquad k \in \{m+1,\dots,n\}.
\end{equation}
In the previous display, $q$ is a strictly positive function whose role is to potentially give more weight to recent observations. In the sequel, we consider the parametric form
\begin{equation}
  \label{eq:q}
q(s,t) = \max\{ s^\gamma (t-s)^\gamma, \delta \}, \qquad 0 \leq s \leq t,
\end{equation}
where $\delta \in (0,1)$ is a small constant and $\gamma$ is a parameter in $[0,1/2]$. If $\gamma = 0$, $q$ is the constant function $1$ and $R_{m,q}(k)$ is then a straightforward adaptation of $R_k$ in~\eqref{eq:Rk} to sequential change-point detection. In that case, the general form of $R_{m,q}(k)$ can also be heuristically justified through a likelihood ratio approach; see Section~2 in \cite{DetGos19}. When $\gamma = 1/2$, as shall be discussed in Remark~\ref{rem:q} using asymptotic arguments, $R_{m,q}(k)$ can be regarded, under $H_0$ in~\eqref{eq:H0}, as a maximum of random variables with, approximately, the same mean and variance. This heuristically implies that all the potential break points $j \in \{m,\dots,k-1\}$ are given roughly the same weight in the computation of $R_{m,q}(k)$ unlike in the case $\gamma = 0$ in which potential break points closest to $\ip{k/2}$ are given more weight. Hence, for certain types of alternatives to $H_0$ in~\eqref{eq:H0}, choosing $\gamma \in (0,1/2]$ might accelerate the detection of the corresponding change in the contemporary distribution of the underlying time series. Examples of such alternatives will be given in Section~\ref{sec:MC} in which the results of numerous Monte Carlo experiments are summarized.

Two similar Cramér--von Mises-like detectors are also considered in this work. They are defined, for $k \in \{m+1,\dots,n\}$, by
\begin{equation}
  \label{eq:Smq}
  \begin{split}
    S_{m,q}(k) &= \max_{m \leq j \leq k-1} \int_{\R^d} \left[ \frac{j (k-j)}{m^{3/2} q(j/m,k/m)}  \big\{ F_{1:j}(\bm x) - F_{j+1:k}(\bm x) \big\} \right]^2 \dd F_{1:k}(\bm x) \\
    &= \max_{m \leq j \leq k-1} \frac{1}{k} \sum_{i=1}^k \left[ \frac{j (k-j)}{m^{3/2} q(j/m,k/m)}  \big\{ F_{1:j}(\bm X_i) - F_{j+1:k}(\bm X_i) \big\} \right]^2,
  \end{split}
\end{equation}
and
\begin{align}
  \label{eq:Tmq}
  T_{m,q}(k)  &= \frac{1}{m} \sum_{j=m}^{k-1} \frac{1}{k} \sum_{i=1}^k \left[ \frac{j (k-j)}{m^{3/2} q(j/m,k/m)} \big \{ F_{1:j}(\bm X_i) - F_{j+1:k}(\bm X_i) \big\} \right]^2.
\end{align}

\begin{remark}
  The detectors in~\eqref{eq:Rmq} and~\eqref{eq:Smq} are related to those used in \citet[pp 104--106]{RosAda12}. The latter are also based on differences of empirical d.f.s but deal only with independent univariate observations. The analogue of~\eqref{eq:Rmq} is apparently defined as a maximum of the quantities $\sup_{\bm x \in \R} | F_{1:j}(\bm x) - F_{j+1:k}(\bm x) |$, $j \in \{m,\dots,k-1\}$, previously normalized using an empirical probability integral transformation, probably based on simulations. The analogue of~\eqref{eq:Smq} takes the form of a maximum of the quantities $\sum_{i=1}^k \{ F_{1:j}(\bm X_i) - F_{j+1:k}(\bm X_i) \}^2$, $j \in \{m,\dots,k-1\}$, after centering and scaling. The asymptotics of these detectors were not studied. Given that the detectors of \cite{RosAda12} are distribution-free (see also Section~\ref{sec:margin:free:H0} hereafter), that their approach assumes serially independent observations and is based on simulations, the absence of asymptotic theory is not problematic.
\end{remark}

\begin{remark}
The detectors in~\eqref{eq:Rmq},~\eqref{eq:Smq} and~\eqref{eq:Tmq} are almost of the Page-CUSUM type considered initially in \cite{Fre15}; see also \cite{KirWeb18}. For instance, in the case of~\eqref{eq:Rmq}, the adaption of the latter construction to the present setting would have instead involved a maximum of the quantities $\sup_{\bm x \in \R^d} | F_{1:m}(\bm x) - F_{j+1:k}(\bm x) |$, $j \in \{m,\dots,k-1\}$. As explained in \cite{DetGos19}, the use of such detectors may result in a loss of power in the case of a small learning sample and a rather late change point. In the Monte Carlo experiments carried out in \cite{DetGos19}, Page-CUSUM detectors were always outperformed by their analogues of type~\eqref{eq:Rmq}, which is why we do not consider them in this work.
\end{remark}

\begin{remark}
In addition to the detectors~\eqref{eq:Rmq},~\eqref{eq:Smq} and~\eqref{eq:Tmq}, we also considered in our Monte Carlo experiments the following natural competitors which are straightforward adaptions of the so-called CUSUM construction considered for instance in \cite{HorHusKokSte04} and \cite{AueHorHorHusSte12}. Their Kolmogorov--Smirnov versions and Cramér--von Mises versions are respectively given, for any $k \in \{m+1,\dots,n\}$, by
\begin{align}
  \label{eq:Pm}
  P_m(k) =& \frac{m (k-m)}{m^{3/2}}  \sup_{\bm x \in \R^d} | F_{1:m}(\bm x) - F_{m+1:k}(\bm x) |, \\
  \label{eq:Qm}
  Q_m(k) =&   \frac{1}{k} \sum_{i=1}^k \left[ \frac{m (k-m)}{m^{3/2}}  \big\{F_{1:m}(\bm X_i) - F_{m+1:k}(\bm X_i) \big\} \right]^2.
  % &= \int_{\R^d}  \left[ \frac{m (k-m)}{m^{3/2}}  \{F_{1:m}(\bm x) - F_{m+1:k}(\bm x) \} \right]^2 \dd F_{1:k}(\bm x)
\end{align}
The asymptotic theory for these detectors being simpler than for the detectors~\eqref{eq:Rmq},~\eqref{eq:Smq} and~\eqref{eq:Tmq},  it will not be stated in the forthcoming sections for the sake of readability.
\end{remark}

\subsection{The detectors are margin-free under the null}
\label{sec:margin:free:H0}

The detectors defined previously are actually margin-free under $H_0$ in~\eqref{eq:H0}, a property that shall be exploited in the forthcoming sections to carry out the corresponding sequential change-point detection procedures.

Recall that $\bm X_1,\dots,\bm X_n$ are assumed to be continuous random vectors. Saying that the detectors are margin-free under the null means that they do not depend on the $d$ univariate margins $F_1,\dots,F_d$ of $F$ (the unknown d.f.\ of $\bm X_1$) or, equivalently, that they can alternatively be written in terms of the unobservable random vectors $\bm U_1,\dots,\bm U_n$ defined from $\bm X_1,\dots,\bm X_n$ through marginal probability integral transformations:
\begin{equation}
  \label{eq:Ui}
  \bm U_i = (F_1(X_{i1}),\dots,F_d(X_{id})).
\end{equation}
Notice that we can recover the $\bm X_i$ from the $\bm U_i$ by marginal quantile transformations:
\begin{equation}
  \label{eq:Xi}
  \bm X_i = (F_1^{-}(U_{i1}),\dots,F_d^{-}(U_{id})),
\end{equation}
where, for any univariate d.f.\ $G$, $G^{-1}$ denotes its associated quantile function defined by
\begin{equation}
  \label{eq:quant:func}
  G^{-1}(y) = \inf \{x \in \R : G(x) \geq y\}, \qquad y \in [0,1],
\end{equation}
with the convention that the infimum of the empty set is $\infty$.

To verify that the detectors are margin-free under the null, for any integers $j,k \geq 1$ and $\bm u \in [0,1]^d$, let
\begin{equation}
  \label{eq:Cjk}
  C_{j:k}(\bm u) =
  \left\{
    \begin{array}{ll}
      \disp \frac{1}{k-j+1}\sum_{i=j}^k \1(\bm U_i \leq \bm u),  \qquad &\text{if } j \leq k, \\
      0, \qquad &\text{otherwise},
    \end{array}
  \right.
\end{equation}
be the analogue of $F_{j:k}$ in~\eqref{eq:Fjk} based on the $\bm U_i$ in~\eqref{eq:Ui}. For any $j \in \{1,\dots,d\}$, by (right) continuity of $F_j$, we have that $\1 \{ F_j^-(u) \leq x \} = \1 \{ u \leq F_j(x) \}$ for all $u \in [0,1]$ and $x \in \R$; see, e.g., Proposition~1~(5) in \cite{EmbHof13}. The latter property combined with~\eqref{eq:Xi} implies that, under $H_0$ in~\eqref{eq:H0}, for any $i \in \{1,\dots,n\}$,
\begin{equation}
  \label{eq:1Xi:1Ui}
  \1 ( \bm X_i \leq \bm x ) = \1 \{ \bm U_i \leq \bm F(\bm x) \}, \qquad \bm x \in \R^d,
\end{equation}
where $\bm F(\bm x) = (F_1(x_1),\dots,F_d(x_d))$. Hence, for any $k \in \{m+1,\dots,n\}$,
\begin{align*}
  R_{m,q}(k) &=  \max_{m \leq j \leq k-1} \frac{j (k-j)}{m^{3/2} q(j/m,k/m)}  \sup_{\bm x \in \R^d} | C_{1:j}\{\bm F(\bm x)\} - C_{j+1:k}\{\bm F(\bm x)\} |, \\
             &=  \max_{m \leq j \leq k-1} \frac{j (k-j)}{m^{3/2} q(j/m,k/m)}  \sup_{\bm u \in [0,1]^d} | C_{1:j}(\bm u) - C_{j+1:k}(\bm u) |
\end{align*}
using again the continuity of $F_1,\dots,F_d$. Similarly,
\begin{align*}
  S_{m,q}(k) &= \max_{m \leq j \leq k-1} \frac{1}{k} \sum_{i=1}^k \left[ \frac{j (k-j)}{m^{3/2} q(j/m,k/m)}  \{ C_{1:j}(\bm U_i) - C_{j+1:k}(\bm U_i) \} \right]^2, \\
  T_{m,q}(k) &= \frac{1}{m} \sum_{j=m}^{k-1} \frac{1}{k} \sum_{i=1}^k \left[ \frac{j (k-j)}{m^{3/2} q(j/m,k/m)}  \{ C_{1:j}(\bm U_i) - C_{j+1:k}(\bm U_i) \} \right]^2.
\end{align*}

In the case of univariate independent observations, the margin-free property under the null implies that the detectors are distribution-free under the null. When $d>1$, this is not true anymore as the null distribution of the detectors depends on the \emph{copula} $C$ associated with~$F$. The latter is merely the d.f.\ of the random vector $\bm U_1$ obtained through~\eqref{eq:Ui}. Equivalently, $C$ is a $d$-dimensional d.f.\ with standard uniform margins further uniquely defined \citep[see][]{Skl59} through the relationships
\begin{equation*}
  \label{eq:F:C}
F(\bm x) = C\{F_1(x_1),\dots,F_d(x_d) \}, \qquad \bm x \in \R^d,
\end{equation*}
and
\begin{equation*}
  \label{eq:C:F}
C(\bm u) = F\{F_1^-(u_1),\dots,F_d^-(u_d) \}, \qquad \bm u \in [0,1]^d.
\end{equation*}

\begin{remark}
 % \label{rem:copula}
To be able to handle both the univariate and the multivariate situations, in the rest of the paper, we adopt the convention that $C$ is the copula associated with $F$ when $d > 1$ and merely the identity function when $d=1$.
\end{remark}

\subsection{Asymptotics of the detectors under the null}
\label{sec:asym:det:H0}

As shall become clear in the forthcoming sections, the knowledge of the asymptotic behavior of the detectors under $H_0$ in~\eqref{eq:H0} is instrumental in showing the asymptotic validity of the corresponding sequential change-point detection procedures. To study these asymptotics, we follow \cite{DetGos19}, among others, and set $n=\ip{m(T+1)}$ for some fixed real number $T > 0$. This will imply that, in the asymptotics, as the size $m$ of the learning sample goes to infinity, the maximum number of new observations considered in the monitoring increases proportionally.

Let $\Delta = \{(s,t) \in [0,T+1]^2 : s \leq t \}$ and let
\begin{equation}
  \label{eq:lambda}
  \lambda_m(s,t) = (\ip{mt} - \ip{ms})/m, \qquad (s, t) \in \Delta.
\end{equation}
Then, for any $(s, t) \in \Delta$ and $\bm u \in [0,1]^d$, let
\begin{equation}
  \label{eq:Gbm}
\Gb_m(s,t,\bm u) = \sqrt{m} \lambda_m(0,s) \lambda_m(s,t) \{ C_{1:\ip{ms}}(\bm u) - C_{\ip{ms}+1:\ip{mt}}(\bm u)\}
\end{equation}
where $C_{1:\ip{ms}}$ and $C_{\ip{ms}+1:\ip{mt}}$ are generically defined by~\eqref{eq:Cjk}, and let
\begin{equation}
  \label{eq:Gbmq}
\Gb_{m,q}(s,t,\bm u) = \frac{\Gb_m(s,t,\bm u)}{q \{ \lambda_m(0,s),\lambda_m(0,t) \}},
\end{equation}
where $q$ is defined in~\eqref{eq:q}. Notice that, with the definitions adopted thus far, $\Gb_m(s,s,\cdot) = \Gb_{m,q}(s,s,\cdot) = 0$ for all $s \in [0,T+1]$.

For any $k \in \{m+1,\dots,n\}$ with $n = \ip{m(T+1)}$ and any $j \in \{m,\dots,k-1\}$, there exists $(s,t) \in  \Delta \cap [1,T+1]^2$ such that $k = \ip{mt}$ and $j = \ip{ms}$. We can thus write $R_{m,q}(k)$ as
\begin{equation}
  \label{eq:Rmq:Gmq}
  \begin{split}
    R_{m,q}(\ip{mt}) &= \max_{m \leq j \leq k-1} \frac{j (k-j)}{m^{3/2} q(j/m,k/m)}   \sup_{\bm u \in [0,1]^d} | C_{1:j}(\bm u) - C_{j+1:k}(\bm u) | \\
    &= \sup_{s \in [1,t]} \sup_{\bm u \in [0,1]^d} \frac{\sqrt{m} \lambda_m(0,s) \lambda_m(s,t) | C_{1:\ip{ms}}(\bm u) - C_{\ip{ms}+1:\ip{mt}}(\bm u) |}{q\{ \lambda_m(0,s),\lambda_m(0,t) \}}  \\
    &= \sup_{s \in [1,t]} \sup_{\bm u \in [0,1]^d} |\Gb_{m,q}(s,t,\bm u)|.
  \end{split}
\end{equation}
Similarly, it can be verified that
\begin{align}
  \label{eq:Smq:Gmq}
  S_{m,q}(k) = S_{m,q}(\ip{mt}) &= \sup_{s \in [1,t]} \int_{[0,1]^d} \{  \Gb_{m,q}(s,t,\bm u) \}^2 \dd C_{1:\ip{mt}}(\bm u), \\
  \label{eq:Tmq:Gmq}
  T_{m,q}(k) = T_{m,q}(\ip{mt}) &= \int_1^t \int_{[0,1]^d} \{  \Gb_{m,q}(s,t,\bm u) \}^2 \dd C_{1:\ip{mt}}(\bm u) \dd s.
\end{align}

As we continue, we adopt the convention that $R_{m,q}(m) =  S_{m,q}(m) =  T_{m,q}(m) = 0$. Furthermore, given a set $\Sc$, the space of all bounded real-valued functions on $\Sc$ equipped with the uniform metric is denoted by $\ell^\infty(\Sc)$. The main purpose of this section is to study the asymptotics under the null of the elements $\Rb_{m,q}$, $\Sb_{m,q}$ and $\Tb_{m,q}$ of $\ell^\infty([1,T+1])$ defined respectively, for any $t \in [1,T+1]$, by
\begin{equation}
  \label{eq:det:as:func}
  \Rb_{m,q}(t) = R_{m,q}(\ip{mt}), \qquad \Sb_{m,q}(t) = S_{m,q}(\ip{mt}), \qquad \Tb_{m,q}(t) = T_{m,q}(\ip{mt}).
\end{equation}
Specifically, we will provide conditions under which they converge weakly in the sense of Definition~1.3.3 in \cite{vanWel96} under $H_0$ in~\eqref{eq:H0}. Throughout the paper, this mode of convergence will be denoted by the arrow~`$\leadsto$' and all convergences will be for $m \to \infty$ unless mentioned otherwise.

From the expressions given in~\eqref{eq:Rmq:Gmq},~\eqref{eq:Smq:Gmq} and~\eqref{eq:Tmq:Gmq}, we see that, under the null, the detectors studied in this work are functionals of $\Gb_{m,q}$ in~\eqref{eq:Gbmq}, and thus of $\Gb_m$ in~\eqref{eq:Gbm}. Under stationarity, the latter is in turn a functional of the \emph{sequential empirical process} defined, for any $s \in [0,T+1]$ and $\bm u \in [0,1]^d$, by
\begin{equation}
  \label{eq:Bbm}
  \begin{split}
    \Bb_m(s,\bm u) &= \frac{1}{\sqrt{m}} \sum_{i=1}^{\ip{ms}} \{ \1(\bm U_i \leq \bm u) - C(\bm u) \} = \sqrt{m} \lambda_m(0,s) \{ C_{1:\ip{ms}}(\bm u) - C(\bm u)\}.
  \end{split}
  \end{equation}
Indeed, under $H_0$ in~\eqref{eq:H0}, for any $(s, t) \in \Delta$ and $\bm u \in [0,1]^d$,
\begin{equation}
  \label{eq:GbmH0}
  \begin{split}
    \Gb_m(s,t,\bm u) &= \sqrt{m} \lambda_m(0,s) \lambda_m(s,t) \{ C_{1:\ip{ms}}(\bm u) - C(\bm u) - C_{\ip{ms}+1:\ip{mt}}(\bm u) + C(\bm u)\} \\
    &= \lambda_m(s,t) \Bb_m(s, \bm u) - \lambda_m(0,s) \times \sqrt{m} \lambda_m(s,t) \{  C_{\ip{ms}+1:\ip{mt}}(\bm u) - C(\bm u) \} \\
    &= \{ \lambda_m(0,t) - \lambda_m(0,s) \} \Bb_m(s, \bm u) - \lambda_m(0,s) \{ \Bb_m(t, \bm u) - \Bb_m(s, \bm u) \} \\
    &= \lambda_m(0,t) \Bb_m(s, \bm u) - \lambda_m(0,s) \Bb_m(t, \bm u).
  \end{split}
\end{equation}

In the forthcoming asymptotic results, we shall assume that the underlying stationary sequence $(\bm X_i)_{i \in \Z}$ (or, equivalently, the corresponding unobservable stationary sequence $(\bm U_i)_{i \in \Z}$ defined through~\eqref{eq:Ui}) is \emph{strongly mixing}. Denote by $\Fc_j^k$ the $\sigma$-field generated by $(\bm X_i)_{j \leq i \leq k}$, $j, k \in \Z \cup \{-\infty,+\infty \}$, and recall that the strong mixing coefficients corresponding to the stationary sequence $(\bm X_i)_{i \in \Z}$ are then defined by $\alpha_0^{\bm X} = 1/2$,
\begin{equation*}
%\label{eq:alpha}
\alpha_r^{\bm X} = \sup_{A \in \Fc_{-\infty}^0,B\in \Fc_{r}^{+\infty}} \big| \Pr(A \cap B) - \Pr(A) \Pr(B) \big|, \qquad r \in \N, \, r > 0,
\end{equation*}
and that the sequence $(\bm X_i)_{i \in \Z}$ is said to be \emph{strongly mixing} if $\alpha_r^{\bm X} \to 0$ as $r \to \infty$. The following result is proven in Appendix~\ref{sec:proof:asym:det}.

\begin{prop}
  \label{prop:H0}
  Assume that $H_0$ in~\eqref{eq:H0} holds and that, additionally, $\bm X_1,\dots,\bm X_n$ is a stretch from a stationary sequence $(\bm X_i)_{i \in \Z}$ of continuous $d$-dimensional random vectors whose strong mixing coefficients satisfy $\alpha_r^{\bm X} = O(r^{-a})$ for some $a > 1$ as $r \to \infty$. Then, $\Bb_m \leadsto \Bb_C$ in $\ell^\infty([0,T+1] \times [0,1]^d)$, where $\Bb_C$ is a tight centered Gaussian process with covariance function
$$
\Cov\{\Bb_C(s,\bm u), \Bb_C(t,\bm v)\} = (s \wedge t) \Gamma(\bm u, \bm v), \qquad s,t \in [0,T+1], \bm u, \bm v \in [0,1]^d,
$$
with $\wedge$ the minimum operator and
\begin{equation}
  \label{eq:Gamma}
\Gamma(\bm u, \bm v) = \sum_{i \in \Z} \Cov\{ \1(\bm U_0 \leq \bm u), \1(\bm U_i \leq \bm v) \}.
\end{equation}
As a consequence, $\Gb_m \leadsto \Gb_C$ and $\Gb_{m,q} \leadsto \Gb_{C,q}$ in $\ell^\infty(\Delta \times [0,1]^d)$, where $\Gb_m$ and $\Gb_{m,q}$ are defined in~\eqref{eq:Gbm} and~\eqref{eq:Gbmq}, respectively, and, for any $(s, t) \in \Delta$ and $\bm u \in [0,1]^d$,
\begin{equation}
  \label{eq:GbC}
  \Gb_C(s,t, \bm u) = t \Bb_C(s, \bm u) - s \Bb_C(t, \bm u) \qquad \text{and} \qquad \Gb_{C,q}(s,t, \bm u) = \frac{\Gb_C(s,t,\bm u)}{q(s,t)}.
\end{equation}
It follows that $\Rb_{m,q} \leadsto \Rb_{C,q}$,  $\Sb_{m,q} \leadsto \Sb_{C,q}$ and  $\Tb_{m,q} \leadsto \Tb_{C,q}$ in $\ell^\infty([1,T+1])$, where
\begin{equation}
  \label{eq:weak:limits}
  \begin{split}
    \Rb_{C,q}(t) &= \sup_{s \in [1,t]} \sup_{\bm u \in [0,1]^d} |\Gb_{C,q}(s,t,\bm u)|, \\
    \Sb_{C,q}(t) &= \sup_{s \in [1,t]} \int_{[0,1]^d} \{  \Gb_{C,q}(s,t,\bm u) \}^2 \dd C(\bm u), \\
    \Tb_{C,q}(t) &= \int_1^t \int_{[0,1]^d} \{  \Gb_{C,q}(s,t,\bm u) \}^2 \dd C(\bm u) \dd s,  \qquad t \in [1,T+1].
\end{split}
\end{equation}
Furthermore, for any interval $[t_1,t_2] \subset [1,T+1]$ such that $t_2 > 1$, the distributions of $\sup_{t \in [t_1,t_2]}\Rb_{C,q}(t)$, $\sup_{t \in [t_1,t_2]}\Sb_{C,q}(t)$ and $\sup_{t \in [t_1,t_2]}\Tb_{C,q}(t)$ are absolutely continuous with respect to the Lebesgue measure.
\end{prop}

Combined with a generic result on the threshold estimation procedure to be stated in Section~\ref{sec:thresh:generic} and additional bootstrap consistency results to be stated in Section~\ref{sec:thresh:estim}, the last claims of Proposition~\ref{prop:H0} constitute a first step in proving that the derived change-point detection procedures hold their level asymptotically.

% \begin{remark}
%   %\label{rem:Gamma}
% When $(\bm X_i)_{i \in \Z}$ is an independent and identically distributed sequence, it can be verified that $\Gamma(\bm u, \bm v) =  C(\bm u \wedge \bm v) - C(\bm u) C(\bm v)$, $\bm u, \bm v \in [0,1]^d$, and $\Bb_C$ is called {\em $C$-Kiefer-M\"uller} process on $[0,T+1] \times [0,1]^d$.
% \end{remark}

\begin{remark}
  \label{rem:q}
  Proposition~\ref{prop:H0} can be used to heuristically justify the form of the weight function $q$ in~\eqref{eq:q} appearing in the expression of the detectors~\eqref{eq:Rmq},~\eqref{eq:Smq} and~\eqref{eq:Tmq}. From~\eqref{eq:GbC}, for any $(s,t) \in \Delta$ and $\bm u \in [0,1]^d$, we obtain that
\begin{align*}
\Var\{\Gb_C(s,t,\bm u) \} &= t^2 \Var\{ \Bb_C(s, \bm u) \} + s^2 \Var\{ \Bb_C(t, \bm u) \} - 2 s t \Cov\{ \Bb_C(s, \bm u), \Bb_C(t, \bm u) \} \\
  &= (s t^2  + s^2 t - 2 s^2 t) \Gamma(\bm u, \bm u) = st (t-s) \Gamma(\bm u, \bm u).
\end{align*}
As a consequence, for any $1 \leq s < t \leq T+1$ and $\bm u \in [0,1]^d$,
$$
\Var\{s^{-1/2}  (t-s)^{-1/2} \Gb_C(s,t,\bm u) \} =  t \Gamma(\bm u, \bm u).
$$
Under the conditions of the proposition, when $\gamma = 1/2$ and $m$ is large, we could then expect that, very roughly, $\Var \{ \Gb_{m,q}(s, t,\bm u) \} \approx \Var\{s^{-1/2}  (t-s)^{-1/2} \Gb_C(s,t,\bm u) \}$ does not depend on $s$ and thus regard the quantities $\sup_{\bm u \in [0,1]^d} \Gb_{m,q}(j/m, t,\bm u)$, $j \in \{m,\dots,\ip{mt}-1\}$, appearing in the expression of $R_{m,q}(\ip{mt})$ in~\eqref{eq:Rmq:Gmq} as random variables with, approximately, the same mean and variance. The latter conveys the intuition that, when $\gamma = 1/2$, all the potential break points $j \in \{m,\dots,\ip{mt}-1\}$ are given roughly the same weight in the computation of $R_{m,q}(\ip{mt})$. This is, of course, only approximately true because of the presence of the constant $\delta$ in the expression of the weight function $q$ in~\eqref{eq:q}. In practice, the setting $\gamma = 1/2$  might accelerate the detection of certain types of changes.
\end{remark}

\subsection{Asymptotics of the detectors under alternatives related to $H_1$}
\label{sec:asym:det:H1}

Under alternatives related to $H_1$ in~\eqref{eq:H1}, the detectors are not margin-free anymore. As we shall see in the forthcoming proposition, their asymptotic behavior is then a consequence of that of the process
\begin{equation*}
%  \label{eq:Hbmq}
\Hb_{m,q}(s,t,\bm x) = \frac{\Hb_m(s,t,\bm x)}{q \{ \lambda_m(0,s),\lambda_m(0,t) \}}, \qquad (s,t) \in \Delta, \bm x \in \R^d,
\end{equation*}
where
\begin{equation}
  \label{eq:Hbm}
\Hb_m(s,t,\bm x) = \sqrt{m} \lambda_m(0,s) \lambda_m(s,t) \{ F_{1:\ip{ms}}(\bm x) - F_{\ip{ms}+1:\ip{mt}}(\bm x)\},
\end{equation}
the empirical d.f.s $F_{1:\ip{ms}}$ and $F_{\ip{ms}+1:\ip{mt}}$ are generically defined by~\eqref{eq:Fjk}, and $q$ is defined in~\eqref{eq:q}.% Notice that, with the definitions adopted thus far, $\Hb_m(s,s,\cdot) = \Hb_{m,q}(s,s,\cdot) = 0$ for all $s \in [0,T+1]$.

The following result is proven in Appendix~\ref{sec:proof:asym:det}. The arrow~`$\p$' in its statement denotes convergence in probability.

\begin{prop}
\label{prop:H1}
Assume that, for some $c \in [1,T+1)$, $\bm X_1,\dots, \bm X_{\ip{mc}}$, denoted equivalently by $\bm Y_1,\dots, \bm Y_{\ip{mc}}$, is a stretch from a stationary sequence $(\bm Y_i)_{i \in \Z}$ of continuous $d$-dimensional random vectors with contemporary d.f.\ $F$ whose strong mixing coefficients satisfy $\alpha_r^{\bm Y} = O(r^{-a})$ for some $a > 1$ as $r \to \infty$, and that $\bm X_{\ip{mc}+1},\dots, \bm X_{\ip{m(T+1)}}$, denoted equivalently by $\bm Z_{\ip{mc}+1},\dots, \bm Z_{\ip{m(T+1)}}$, is a stretch from a stationary sequence $(\bm Z_i)_{i \in \Z}$ of continuous $d$-dimensional random vectors with contemporary d.f.\ $G \neq F$ whose strong mixing coefficients satisfy $\alpha_r^{\bm Z} = O(r^{-b})$ for some $b > 1$ as $r \to \infty$. Then, $m^{-1/2} \Hb_m \p K_c$ in $\ell^\infty(\Delta \times \R^d)$, where 
\begin{equation}
  \label{eq:Kc}
  K_c(s,t,\bm x) = (s \wedge c) \{(t \vee c) - (s \vee c) \} \{F(\bm x) - G(\bm x)\}, \qquad (s,t) \in \Delta, \bm x \in \R^d.
\end{equation}
Consequently, $m^{-1/2} \Hb_{m,q} \p K_{c,q}$ in $\ell^\infty(\Delta \times \R^d)$, where $K_{c,q}(s,t,\bm x) = K_c(s,t,\bm x) / q(s,t)$, $(s,t) \in \Delta$, $\bm x \in \R^d$, and
$$
\sup_{t \in [1,T+1]} \Rb_{m,q}(t) \p \infty, \quad \sup_{t \in [1,T+1]} \Sb_{m,q}(t) \p \infty \quad \text{ and } \quad \sup_{t \in [1,T+1]} \Tb_{m,q}(t) \p \infty.
$$
\end{prop}

Combined with a generic result on the threshold estimation procedure to be stated in the forthcoming section and additional results on the asymptotic validity of adequate resampling methods to be stated in Section~\ref{sec:thresh:estim}, the three last claims of the previous result will be instrumental in showing that the derived change-point detection procedures have asymptotic power one under sequences of alternatives related to $H_1$.

%%%%%%%%%%%%%%%%%%%%%%%%%%%%%%%%%%%%%%%%%%%%%%%%%%%%%%%%%%%%%%%%%%%%%%%%%%%%%%%%%%%%%

\section{A generic threshold estimation procedure}
\label{sec:thresh:generic}

In the studied context, the second ingredient of a sequential change-point detection procedure is a set of strictly positive thresholds to which detectors will be compared. In this section, we consider a generic threshold estimation procedure that can be employed with any type of detector, and provide conditions under which it is asymptotically valid. The derived results will be applied in the next section to establish the asymptotic validity of sequential change-point detection procedures based on the detectors studied in Section~\ref{sec:detectors}.

\subsection{A constant probability of false alarm at each step} Within the context of closed-end monitoring from time $m+1$ to time $n$, let $D_m(k)$, $k \in \{m+1,\dots,n\}$, be arbitrary detectors. As discussed in the introduction, it seems natural to choose the corresponding thresholds $w_m(k)$, $k \in \{m+1,\dots,n\}$, so that, under $H_0$ in~\eqref{eq:H0}, the probability of rejection of $H_0$ is the same at every step $k \in \{m+1,\dots,n\}$ of the procedure. Adapting the reasoning of \citet{AnaKos18} to our context, the latter requirement consists of choosing the $w_m(k)$, $k \in \{m+1,\dots,n\}$, such that, under stationarity, for all $k \in \{m+1,\dots,n\}$,
\begin{equation}
  \label{eq:unif:FA}
\Pr \{ D_m(m+1) \leq w_m(m+1), \dots, D_m(k) \leq w_m(k)  \} = 1 - \frac{k-m}{n-m} \alpha,
\end{equation}
where $\alpha \in (0,1/2)$ is the desired significance level of the sequential testing procedure, or, equivalently, such that, under $H_0$, for all $k \in \{m+1,\dots,n\}$,
$$
\Pr \{ \exists \, i \in \{m+1,\dots,k\} \text{ s.t. } D_m(i) > w_m(i) \}  = \frac{k-m}{n-m} \alpha.
$$
Notice that the previous display translates mathematically our requirement that the probability of false alarm be proportional to the number of monitoring steps.

Interestingly enough, the previous way of choosing the thresholds $w_m(k)$, $k \in \{m+1,\dots,n\}$, is strongly related to an approach used in SPC and possibly first appearing in \cite{MarConWooDra95} \citep[see also, e.g.,][]{HawZam05,Ros14}. It consists of choosing the $w_m(k)$, $k \in \{m+1,\dots,n\}$, such that, under stationarity, for some small $\xi_m > 0$,
\begin{equation}
  \label{eq:thresh:proc}
  \begin{split}
\left\{
  \begin{array}{l}
    \Pr\{ D_m(m+1) > w_m(m+1) \} = \xi_m, \\
    \\
    \text{and, for all } k \in \{m+2,\dots,n\} , \\
    \\
    \Pr\{ D_m(k) > w_m(k) \mid D_m(m+1) \leq w_m(m+1), \dots, D_m(k-1) \leq w_m(k-1)
    \} = \xi_m.
  \end{array}
\right.
\end{split}
\end{equation}
We then obtain that, under $H_0$, for all $k \in \{m+2,\dots,n\}$,
\begin{equation}
  \label{eq:decomp:cond}
  \begin{split}
    \Pr \{ &D_m(m+1) \leq w_m(m+1), \dots, D_m(k) \leq w_m(k)  \} \\
    =& \Pr \{ D_m(k) \leq w_m(k) \mid D_m(m+1) \leq w_m(m+1), \dots, D_m(k-1) \leq w_m(k-1) \} \\ &\times \Pr \{ D_m(m+1) \leq w_m(m+1), \dots \dots, D_m(k-1) \leq w_m(k-1) \} \\
    =& (1-\xi_m) \times \Pr \{ D_m(k-1) \leq w_m(k-1) \mid D_m(m+1) \leq w_m(m+1), \dots \\
    &\dots, D_m(k-2) \leq w_m(k-2) \} \\
    &\times \Pr \{ D_m(m+1) \leq w_m(m+1), \dots, D_m(k-2) \leq w_m(k-2) \} \\
    =& \dots = (1-\xi_m)^{k-m}
  \end{split}
\end{equation}
with the convention that $D_m(m) = w_m(m) = 0$.

Given a desired significance level $\alpha \in (0,1/2)$, a simple way to ensure that~\eqref{eq:neg:typeI} holds under $H_0$ is to choose $\xi_m$ such that $1-\alpha = (1-\xi_m)^{n-m}$, that is, $\xi_m = 1 - (1-\alpha)^{1/(n-m)}$. As one can see from~\eqref{eq:thresh:proc}, $w_m(m+1)$ is then a quantile of order $ (1-\alpha)^{1/(n-m)}$ of $D_m(m+1)$ under stationarity and, for any $k \in \{m+2,\dots,n\}$, $w_m(k)$ is a quantile of order $ (1-\alpha)^{1/(n-m)}$ of $D_m(k)$ conditionally on $D_m(m+1) \leq w_m(m+1), \dots, D_m(k-1) \leq w_m(k-1)$ under stationarity. When $\alpha \in \{0.01, 0.05, 0.1\}$ as is typically the case, the first-order approximation $\xi_m \simeq 1 - (1- \frac{\alpha}{n-m}) =  \frac{\alpha}{n-m}$ turns out to be precise up to at least two decimals. Similarly, for all $k \in \{m+2,\dots,n\}$,
\begin{align*}
  \Pr \{ D_m(m+1) \leq w_m(m+1), \dots, D_m(k) \leq w_m(k)  \} &= (1-\xi_m)^{k-m} \\ &
                                                                                      \simeq 1 - (k - m) \xi_m \simeq 1 - \frac{k-m}{n-m} \alpha,
\end{align*}
where it can be verified that the resulting approximation is precise up to at least two decimals. In other words, for typical values of $\alpha$, choosing the thresholds $w_m(k)$, $k \in \{m+1,\dots,n\}$ such that~\eqref{eq:thresh:proc} holds with $\xi_m = 1 - (1-\alpha)^{1/(n-m)}$ is almost equivalent to choosing the thresholds such that~\eqref{eq:unif:FA} holds (some thought reveals that the latter equivalence could be made to hold exactly by allowing $\xi_m$ in~\eqref{eq:thresh:proc} to change with $k$).

Given the precision of the aforementioned first-order approximations for typical values of $\alpha$, for the sake of a simplicity, we shall base our threshold estimation procedure on~\eqref{eq:thresh:proc}. Before we discuss the estimation of the thresholds and its validity, let us give an alternative view of~\eqref{eq:thresh:proc}. In Sections~\ref{sec:asym:det:H0} and~\ref{sec:asym:det:H1} in which $n$ was taken equal to $\ip{m(T+1)}$, we saw that the asymptotic results for the detectors are given in terms of elements of $\ell^\infty([1,T+1])$. With the convention that $D_m(m) = w_m(m) = 0$, another equivalent way of looking at sequential change-point detection procedures of the considered type is then to consider that the piecewise constant \emph{detector function} $\Db_m$ defined by $\Db_m(t) = D_m(\ip{mt})$, $t \in [1,T+1]$, is compared to the piecewise constant \emph{threshold function} $\tau_m$ defined by $\tau_m(t) = w_m(\ip{mt})$, $t \in [1,T+1]$. Let $s_k = (m+k)/m$, $k \in \{0, \dots, n-m\}$ and define the intervals $J_k = [s_k,s_{k+1})$, $k \in \{0, \dots, n-m-1\}$, and $J_{n-m} = [s_{n-m},T+1]$. Some thought reveals that~\eqref{eq:thresh:proc} is then equivalent to choosing the threshold function $\tau_m$ such that, under $H_0$ in~\eqref{eq:H0}, for any $k \in \{1, \dots, n-m\}$,
\begin{multline}
\label{eq:thresh:proc1}
  \Pr\{ \exists \, t \in J_k \text{ s.t. } \Db_m(t) > \tau_m(t) \mid \Db_m(t) \leq \tau_m(t), \forall \, t \in J_0 \cup \dots \cup J_{k-1} \} \\ = 1 - (1-\alpha)^{1/(n-m)}.
\end{multline}

\subsection{A formulation compatible with asymptotic validity results}
\label{sec:thresh:asm:val:formulation}

With $n = \ip{m(T+1)}$, the threshold setting procedure as given in~\eqref{eq:thresh:proc} or~\eqref{eq:thresh:proc1} makes no sense asymptotically since the number of (conditional) probabilities tends to infinity as $m \to \infty$. A natural solution consists of keeping the number of probabilities fixed, or, equivalently, of considering a time grid that does not depend on $m$. Let $p \geq 1$ and let $t_0 = 1 < t_1 < \dots < t_p = T+1$ be a fixed uniformly spaced time grid such that $T/p > 1/m$ (a condition that will always be satisfied for $m$ large enough). Let $\tau_m$ be a piecewise constant threshold function taking the value $g_{i,m}$ on the interval $I_i = [t_{i-1}, t_i)$, $i \in \{1,\dots,p-1\}$, and $g_{p,m}$ on the interval $I_p = [t_{p-1}, t_p]$. Mimicking~\eqref{eq:thresh:proc1}, the aim is then to choose $\tau_m$ such that, under $H_0$ in~\eqref{eq:H0}, for any $i \in \{1, \dots, p\}$,
\begin{equation}
  \label{eq:thresh:proc2}
  \Pr\{ \exists \, t \in I_i \text{ s.t. }\Db_m(t) > \tau_m(t) \mid \Db_m(t) \leq \tau_m(t), \forall \, t \in I_0 \cup \dots \cup I_{i-1}  \} = 1 - (1-\alpha)^{1/p},
\end{equation}
with the convention that $I_0 = \emptyset$. Some thought reveals that the formulation in~\eqref{eq:thresh:proc2} is equivalent to choosing $\tau_m$ such that, under $H_0$ in~\eqref{eq:H0}, 
\begin{equation}
  \label{eq:thresh:proc3}
  \left\{
    \begin{array}{l}
      \disp \Pr\left\{ \sup_{t \in I_1} \Db_m(t) > g_{1,m} \right\} = 1 - (1-\alpha)^{1/p},\\ \\
      \text{and, for all } i \in \{2, \dots, p\},\\ \\
      \disp \Pr\left\{ \sup_{t \in I_i} \Db_m(t) > g_{i,m} \, \Big| \, \sup_{t \in I_1} \Db_m(t) \leq g_{1,m}, \dots, \sup_{t \in I_{i-1}} \Db_m(t) \leq g_{i-1,m}  \right\} = 1 - (1-\alpha)^{1/p}. 
    \end{array}
  \right.
\end{equation}
In other words, $g_{1,m}$ is a quantile of order $(1-\alpha)^{1/p}$ of $\sup_{t \in I_1} \Db_m(t)$ under stationarity and $g_{i,m}$, $i \in \{2,\dots,p\}$, is a quantile of order $(1-\alpha)^{1/p}$ of $\sup_{t \in I_i} \Db_m(t)$ given that $\sup_{t \in I_1} \Db_m(t) \leq g_{1,m}, \dots, \sup_{t \in I_{i-1}} \Db_m(t) \leq g_{i-1,m}$ under stationarity. Notice that the suprema in~\eqref{eq:thresh:proc3} are actually maxima since $\Db_m$ is a piecewise constant function.

\begin{remark}
A further generalization of~\eqref{eq:thresh:proc2} or, equivalently~\eqref{eq:thresh:proc3}, would be to consider that $\tau_m$ is not necessarily piecewise constant but only defined up to a multiplicative constant on each of the intervals $I_i$, $i \in \{1,\dots,p\}$. For instance, it could have one of the parametric forms considered in \citet[Section~5]{DetGos19}, among others. For the sake of simplicity, we shall not however consider such an extension in this work.
\end{remark}

\subsection{Estimation of the threshold function}
\label{sec:estim:thresh}

As we continue, we shall focus on the threshold setting procedure as formulated in~\eqref{eq:thresh:proc2} or, equivalently,~\eqref{eq:thresh:proc3}, mostly because its asymptotic validity can be studied. To estimate the threshold function $\tau_m$ in~\eqref{eq:thresh:proc2}, or, equivalently, the $g_{i,m}$, $i \in \{1,\dots,p\}$, in~\eqref{eq:thresh:proc3}, it is thus necessary to be able to compute, at least approximately, the distribution of the $p$-dimensional random vector
\begin{equation}
  \label{eq:rvp}
\left( \sup_{t \in I_1} \Db_m(t), \dots, \sup_{t \in I_p} \Db_m(t) \right).
\end{equation}

\subsubsection{Monte Carlo estimation and asymptotic validity}
\label{sec:MC:estim}

Assume that the observations to be monitored are univariate and independent, and that $\Db_m$ is distribution-free under $H_0^\text{ind}$ in~\eqref{eq:H0:ind}. Notice that the latter implies that so is the random vector~\eqref{eq:rvp}. To obtain a Monte Carlo estimate of the distribution of~\eqref{eq:rvp}, it then suffices to consider a large integer $M$, generate $M$ independent samples $U_1^{[s]},\dots,U_n^{[s]}$, $s \in \{1,\dots,M\}$, of size $n$ from the standard uniform distribution and compute the corresponding realizations $\Db_m^{[s]}$, $s \in \{1,\dots,M\}$, of~$\Db_m$. The latter can be used to obtain a Monte Carlo estimate $\tau_m^M$ of the threshold function~$\tau_m$. More formally, let
$$
g_{i,m}^M = F_{\Db_m,i}^{M,-1} \{ (1-\alpha)^{1/p} \}, \qquad i \in \{1,\dots,p\},
$$
where $F_{\Db_m,1}^M$ is the empirical d.f.\ of the sample $\sup_{t \in I_1} \Db_m^{[1]}(t),\dots,\sup_{t \in I_1} \Db_m^{[M]}(t)$, for any $i \in \{2,\dots,p\}$ and $x \in \R$,
$$
F_{\Db_m,i}^M(x) = \frac{\disp \sum_{s=1}^M \1\left\{\sup_{t \in I_i} \Db_m^{[s]}(t) \leq x, \sup_{t \in I_1} \Db_m^{[s]}(t) \leq g_{1,m}^M, \dots, \sup_{t \in I_{i-1}} \Db_m^{[s]}(t) \leq g_{i-1,m}^M \right\}}{\disp \sum_{s=1}^M \1\left\{ \sup_{t \in I_1} \Db_m^{[s]}(t) \leq g_{1,m}^M, \dots, \sup_{t \in I_{i-1}} \Db_m^{[s]}(t) \leq g_{i-1,m}^M  \right\}},
$$
and $F_{\Db_m,i}^{M,-1}$, $i \in \{1,\dots,p\}$, are the associated quantile functions generically defined by~\eqref{eq:quant:func}. Notice that, in this particular case, the resulting estimate $\tau_m^M$ of the threshold function $\tau_m$ does not at all depend on the learning sample. 

By taking a sufficiently large $M$, the Monte Carlo estimates $g_{i,m}^M$, $i \in \{1,\dots,p\}$, can be made arbitrarily close to the quantiles $g_{i,m} = F_{\Db_m,i}^{-1}\{ (1-\alpha)^{1/p} \}$, $i \in \{1,\dots,p\}$,  where $F_{\Db_m,1}$ is the d.f.\ of $\sup_{t \in I_1} \Db_m(t)$, and $F_{\Db_m,i}$, $i \in \{2,\dots,p\}$, is the d.f.\ of $\sup_{t \in I_i} \Db_m(t)$ given that $\sup_{t \in I_j} \Db_m(t) \leq g_{j,m}$ for all $j \in \{1,\dots,i-1\}$. Interestingly enough more can be said as a consequence of the fact that Monte Carlo simulation can be regarded as a particular resampling scheme. As shall become clear in the next section, the general result stated in Theorem~\ref{thm:thresh:boot} hereafter can actually be used to show the asymptotic validity of the Monte Carlo based threshold estimation procedure when both $m$ and $M$ tend to infinity, under both $H_0^\text{ind}$ in~\eqref{eq:H0:ind} and sequences of alternatives related to $H_1$ in~\eqref{eq:H1}. This is discussed in more detail in Remark~\ref{rem:thresh:MC} below.

\subsubsection{Bootstrap-based estimation and asymptotic validity}
\label{sec:boot:estim}

In settings in which $\Db_m$ is not distribution-free anymore, a natural alternative is to rely on a resampling scheme making use of the available learning sample $\bm X_1,\dots,\bm X_m$ known to be under $H_0$ in~\eqref{eq:H0}. Specifically, let $B$ be a large integer and suppose that we have available \emph{bootstrap replicates} $\Db_m^{[b]}$, $b \in \{1,\dots,B\}$, of $\Db_m$ computed from $\bm X_1,\dots,\bm X_m$ and depending on additional sources of randomness involved in the resampling scheme. Mimicking the previous situation in which $\Db_m$ was distribution-free, let
$$
g_{j,m}^B = F_{\Db_m,j}^{B,-1} \{ (1-\alpha)^{1/p} \}, \qquad j \in \{1,\dots,p\},
$$
where $F_{\Db_m,1}^B$ is the empirical d.f.\ of the sample $\sup_{t \in I_1} \Db_m^{[1]}(t),\dots,\sup_{t \in I_1} \Db_m^{[B]}(t)$, and, for any $j \in \{2,\dots,p\}$ and $x \in \R$,
$$
F_{\Db_m,j}^B(x) = \frac{\disp \sum_{b=1}^B \1 \left\{ \sup_{t \in I_j} \Db_m^{[b]}(t) \leq x, \sup_{t \in I_1} \Db_m^{[b]}(t) \leq g_{1,m}^B, \dots, \sup_{t \in I_{j-1}} \Db_m^{[b]}(t) \leq g_{j-1,m}^B \right\}}{\disp \sum_{b=1}^B \1 \left\{\sup_{t \in I_1} \Db_m^{[b]}(t) \leq g_{1,m}^B, \dots, \sup_{t \in I_{j-1}} \Db_m^{[b]}(t) \leq g_{j-1,m}^B  \right\}}.
$$

As we shall see below, the main result of this section is that, essentially, as soon as the underlying resampling scheme for $\Db_m$ is consistent, the above bootstrap-based version of the threshold setting procedure~\eqref{eq:thresh:proc3} is asymptotically valid in the sense that, under~$H_0$, $\Pr(\Db_m \leq \tau_m^B) \to 1 - \alpha$ as $m,B \to \infty$, where $\tau_m^B$ is the estimated bootstrap-based piecewise constant threshold function taking the value $g_{i,m}^B$ on the interval $I_i = [t_{i-1}, t_i)$, $i \in \{1,\dots,p-1\}$, and $g_{p,m}^B$ on the interval $I_p = [t_{p-1}, t_p]$.

\begin{remark}
  \label{rem:boot:val}
Following for instance \citet[Section~3.6]{vanWel96} or \citet[Section~2.2.3]{Kos08}, a resampling scheme for $\Db_m$ is typically considered consistent, if, informally, ``$\Db_m^{[1]}$ converges weakly to the weak limit of $\Db_m$ in $\ell^\infty([0,T+1])$ conditionally on $\bm X_1, \bm X_2, \dots$ in probability''. A rigorous definition of the underlying mode of convergence is more subtle than that of weak convergence. From Lemma 3.1 of \cite{BucKoj19}, the aforementioned validity statement is actually equivalent to the joint unconditional weak convergence of $\Db_m$ and two bootstrap replicates to independent copies of the same limit. Throughout the paper, all our bootstrap asymptotic validity results will take that form.
\end{remark}

The following general result is proved in Appendix~\ref{sec:proof:thm:thres:boot}.

\begin{thm}
  \label{thm:thresh:boot}
Assume that, under $H_0$ in~\eqref{eq:H0},
\begin{equation}
  \label{eq:boot:val}
  (\Db_m, \Db_m^{[1]}, \Db_m^{[2]}) \leadsto (\Db_F, \Db_F^{[1]}, \Db_F^{[2]})
\end{equation}
in $\{ \ell^\infty([1,T+1]) \}^3$, where $\Db_F$ is the weak limit of $\Db_m$ and $\Db_F^{[1]}$ and $\Db_F^{[2]}$ are independent copies of $\Db_F$. Assume furthermore that $\big( \sup_{t \in I_1} \Db_F(t), \dots, \sup_{t \in I_p} \Db_F(t) \big)$ has a continuous d.f. Then, under $H_0$ in~\eqref{eq:H0}, as $m,B \to \infty$,
\begin{equation}
  \label{eq:thresh:val:1}
\Pr \left\{ \sup_{t \in I_1} \Db_m(t) \leq g_{1,m}^B \right\} \to (1-\alpha)^{1/p},
\end{equation}
and, for any $i \in \{2,\dots,p\}$,
\begin{equation}
  \label{eq:thresh:val:i}
 \Pr \left\{ \sup_{t \in I_i} \Db_m(t)  \leq g_{i,m}^B \, \Big| \, \sup_{t \in I_1} \Db_m(t) \leq g_{1,m}^B, \ldots, \sup_{t \in I_{i-1}} \Db_m(t) \leq g_{i-1,m}^B \right\} \to (1-\alpha)^{1/p}.
\end{equation}
%The statements with `$\leq$' replaced by `$<$' hold as well. \\
As a consequence, on one hand, under $H_0$, $\Pr(\Db_m \leq \tau_m^B) \to 1 - \alpha$ as $m,B \to \infty$ and, on the other hand, when $\sup_{t \in [1,T+1]} \Db_m(t) \p \infty$,
$$
\Pr\{ \exists \, t \in [1,T+1] \text{ s.t. } \Db_m(t) > \tau_m^B(t) \} \to 1 \qquad \text{ as } m, B \to \infty.
$$
\end{thm}

\begin{remark}
  \label{rem:thresh:MC}
  Consider the Monte Carlo estimation setting of Section~\ref{sec:MC:estim} in which  the observations to be monitored are univariate independent and $\Db_m$ is distribution-free. Then, the weak convergence $\Db_m \leadsto \Db_F$ in $\ell^\infty([1,T+1])$ under $H_0$ immediately implies~\eqref{eq:boot:val}, where $\Db_m^{[1]}$ and $\Db_m^{[2]}$ are (independent) Monte Carlo replicates of $\Db_m$. Hence, as a consequence of Theorem~\ref{thm:thresh:boot}, the asymptotic validity under $H_0^\text{ind}$ in~\eqref{eq:H0:ind} of the sequential change-point detection procedure based on $\Db_m$ and the Monte Carlo estimated threshold function $\tau_m^M$ defined in Section~\ref{sec:MC:estim} is an immediate corollary of the weak convergence under the null of $\Db_m$ if $\big( \sup_{t \in I_1} \Db_F(t), \dots, \sup_{t \in I_p} \Db_F(t) \big)$ has a continuous d.f.
\end{remark}

\section{Threshold function estimation for the detectors based on empirical d.f.s}
\label{sec:thresh:estim}

The aim of this section is to apply the generic results of the previous section to estimate the threshold functions for the empirical d.f.-based detector functions $\Rb_{m,q}$, $\Sb_{m,q}$ and $\Tb_{m,q}$ defined in~\eqref{eq:det:as:func}. We distinguish two situations for the observations to be monitored: the independent univariate case and the possibly multivariate, time series case.

\subsection{Monte Carlo estimation in the independent univariate case}
\label{sec:thresh:ind:univ}

As verified in Section~\ref{sec:margin:free:H0}, the detector functions $\Rb_{m,q}$, $\Sb_{m,q}$ and $\Tb_{m,q}$ defined in~\eqref{eq:det:as:func} are margin-free under $H_0$ in~\eqref{eq:H0}. In the univariate case, they are thus distribution-free. When dealing with independent univariate observations, one can therefore proceed exactly as explained in Section~\ref{sec:MC:estim} to estimate the corresponding threshold functions. Furthermore, from Proposition~\ref{prop:H0}, Remark~\ref{rem:thresh:MC} and Proposition~\ref{prop:H1}, we know that the assumptions of Theorem~\ref{thm:thresh:boot} are satisfied. The latter then implies that the corresponding sequential change-point detection procedures are asymptotically valid both under $H_0^\text{ind}$ in~\eqref{eq:H0:ind} and sequences of alternatives related to $H_1$ in~\eqref{eq:H1}.

\subsection{A dependent multiplier bootstrap in the time series case}
\label{sec:thresh:time:series}

When the monitored observations are multivariate or exhibit serial dependence, the approach considered in Section~\ref{sec:thresh:ind:univ} is not meaningful anymore. Having the asymptotic results of Sections~\ref{sec:asym:det:H0} and~\ref{sec:boot:estim} in mind, our aim in the considered time series context is to define suitable bootstrap replicates of $\Bb_m$ in~\eqref{eq:Bbm} such that, following Remark~\ref{rem:boot:val}, $\Bb_m$ and two of its replicates jointly weakly converge to independent copies of the process $\Bb_C$ defined in Proposition~\ref{prop:H0}. Subsequently defining corresponding bootstrap replicates of the detectors functions $\Rb_{m,q}$, $\Sb_{m,q}$ and $\Tb_{m,q}$ defined in~\eqref{eq:det:as:func} will lead to asymptotically valid corresponding sequential change-point detection procedures.

Following \citet[Section 3.3]{Buh93} and \cite{BucKoj16}, we opted for a \emph{dependent multiplier bootstrap} in the considered time series context. In the rest of the paper, we say that a sequence of random variables $(\xi_{i,m})_{i \in \Z}$ is a {\em dependent multiplier sequence} if:
\begin{enumerate}[({M}1)]
\item \label{item:moments} The sequence $(\xi_{i,m})_{i \in \Z}$ is stationary, independent of the available learning sample $\bm X_1,\dots,\bm X_m$ and satisfies $\Ex(\xi_{0,m}) = 0$, $\Ex(\xi_{0,m}^2) = 1$ and $\sup_{m \geq 1} \Ex(|\xi_{0,m}|^\nu) < \infty$ for all $\nu \geq 1$.
\item \label{item:lm} There exists a sequence $\ell_m \to \infty$ of strictly positive constants such that $\ell_m = o(m)$ and the sequence $(\xi_{i,m})_{i \in \Z}$ is $\ell_m$-dependent, i.e., $\xi_{i,m}$ is independent of $\xi_{i+h,m}$ for all $h > \ell_m$ and $i \in \N$.
\item \label{item:varphi} There exists a function $\varphi:\R \to [0,1]$, symmetric around 0, continuous at $0$, satisfying $\varphi(0)=1$ and $\varphi(x)=0$ for all $|x| > 1$ such that $\Ex(\xi_{0,m} \xi_{h,m}) = \varphi(h/\ell_m)$ for all $h \in \Z$.
\end{enumerate}

Let $(\xi_{i,m}^{[b]})_{i \in \Z}$, $b \in \N$, be independent copies of the same dependent multiplier sequence. If we had a learning sample of size $n = \ip{m(T+1)}$, following \cite{BucKoj16}, a natural definition of a \emph{dependent multiplier replicate} of $\Bb_m$ in~\eqref{eq:Bbm} would be
\begin{equation}
  \label{eq:check:Bbmb}
\check \Bb_m^{[b]}(s, \bm u) = \frac{1}{\sqrt{m}} \sum_{i=1}^{\ip{ms}}  \xi_{i,m}^{[b]} \{ \1(\bm U_i \leq \bm u) - C_{1:n}(\bm u) \}, \qquad s \in [0,T+1], \bm u \in [0,1]^d,b \in \N,
\end{equation}
where $C_{1:n}$ is generically defined by~\eqref{eq:Cjk}. Since threshold functions need to be estimated prior to the beginning of the monitoring and the learning sample is only of size~$m$, we consider a time-rescaled version of $\check \Bb_m^{[b]}$ in which, roughly, $m'= \ip{ (m/n) m} \simeq m/(T+1)$ and $m$ play the role of $m$ and $n$, respectively. Hence, in the considered context, our definition of a dependent multiplier replicate of $\Bb_m$ is
\begin{equation}
  \label{eq:hat:Bbmb}
  \hat \Bb_m^{[b]}(s, \bm u) = \frac{1}{\sqrt{m'}} \sum_{i=1}^{\ip{m's}}  \xi_{i,m}^{[b]} \{ \1(\bm U_i \leq \bm u) - C_{1:m}(\bm u) \}, \qquad s \in [0,T+1], \bm u \in [0,1]^d,b \in \N,
\end{equation}
thereby translating the fact that we can only rely on functionals computed from the learning sample to approximate the variability of the detector functions under the null. 

From the two previous displays, we see that the multipliers act as random weights and that the bandwidth $\ell_m$ defined in Assumption~(M\ref{item:lm}) plays a role somehow similar to that of the {\em block length} in the block bootstrap of \cite{Kun89}. Note that, in our Monte Carlo experiments to be presented in Section~\ref{sec:MC}, $\ell_m$ was estimated from the learning sample $\bm X_1,\dots,\bm X_m$ as explained in detail in Section~5.1 of \cite{BucKoj16} while corresponding dependent multiplier sequences were generated using the so-called \emph{moving average approach} based on an initial standard normal random sample and Parzen's kernel as precisely described in Section~5.2 of the same reference.

The latter construction based on a time-rescaling suggests to form a dependent multiplier replicate of $\Gb_m$ in~\eqref{eq:GbmH0} as
\begin{equation*}
  \label{eq:hatGbmb}
\hat \Gb_m^{[b]}(s,t, \bm u) = \lambda_{m'}(0,t)  \hat \Bb_m^{[b]}(s,\bm u) - \lambda_{m'}(0,s)  \hat \Bb_m^{[b]}(t,\bm u), \qquad (s,t) \in \Delta, \bm u \in [0,1]^d, b \in \N,
\end{equation*}
with its weighted version being
$$
\hat \Gb_{m,q}^{[b]}(s,t, \bm u) = \frac{\hat \Gb_m^{[b]}(s,t, \bm u)}{q\{\lambda_{m'}(0,s),\lambda_{m'}(0,t) \}}, \qquad (s,t) \in \Delta, \bm u \in [0,1]^d, b \in \N,
$$
where $\lambda_{m'}$ is defined as in~\eqref{eq:lambda}. Finally, for any $b \in \N$ and $t \in [1,T+1]$, let
\begin{equation}
  \label{eq:det:mult}
  \begin{split}
    \hat \Rb_{m,q}^{[b]}(t) &= \sup_{s \in [1,t]} \sup_{\bm u \in [0,1]^d} | \hat \Gb_{m,q}^{[b]}(s,t,\bm u)|, \\
    \hat \Sb_{m,q}^{[b]}(t) &= \sup_{s \in [1,t]} \int_{[0,1]^d} \{ \hat \Gb_{m,q}^{[b]}(s,t,\bm u) \}^2 \dd C_{1:\ip{m't}}(\bm u), \\
    \hat \Tb_{m,q}^{[b]}(t) &= \int_1^t \int_{[0,1]^d} \{ \hat \Gb_{m,q}^{[b]}(s,t,\bm u) \}^2 \dd C_{1:\ip{m't}}(\bm u) \dd s
  \end{split}
\end{equation}
be dependent multiplier replicates of $\Rb_{m,q}$, $\Sb_{m,q}$ and $\Tb_{m,q}$, respectively, defined in~\eqref{eq:det:as:func}, where $C_{1:\ip{m't}}$ is defined generically by~\eqref{eq:Cjk}.

The definitions given in~\eqref{eq:det:mult} hide the fact that the proposed dependent multipliers replicates of the detector functions $\Rb_{m,q}$, $\Sb_{m,q}$ and $\Tb_{m,q}$ actually depend on the learning sample $\bm X_1,\dots,\bm X_m$. To verify that this is the case, for any $b \in \N$, $(s,t) \in \Delta$ and $\bm x \in \R^d$, let $\hat \Fb_m^{[b]}(s, \bm x) = \hat \Bb_m^{[b]}\{s, \bm F(\bm x)\}$, where $\bm F(\bm x) = (F_1(x_1),\dots,F_d(x_d))$,
$$
\hat \Hb_m^{[b]}(s,t, \bm x) = \hat \Gb_m^{[b]}\{s,t, \bm F(\bm x) \} = \lambda_{m'}(0,t)  \hat \Fb_m^{[b]}(s,\bm x) - \lambda_{m'}(0,s)  \hat \Fb_m^{[b]}(t,\bm x)
$$
and
$$
\hat \Hb_{m,q}^{[b]}(s,t, \bm x) = \hat \Gb_{m,q}^{[b]}\{s,t, \bm F(\bm x) \} = \frac{\hat \Hb_m^{[b]}(s,t, \bm x)}{q\{\lambda_{m'}(0,s),\lambda_{m'}(0,t) \}}.
$$
Since~\eqref{eq:1Xi:1Ui} always holds for all $i \in \{1,\dots,m\}$, we immediately obtain that, for any $b \in \N$,
\begin{equation*}
  \label{eq:hat:Fbmb}
  \hat \Fb_m^{[b]}(s, \bm x) = \frac{1}{\sqrt{m'}} \sum_{i=1}^{\ip{m's}}  \xi_{i,m}^{[b]} \{ \1(\bm X_i \leq \bm x) - F_{1:m}(\bm x) \}, \qquad s \in [0,T+1], \bm x \in \R^d,
\end{equation*}
where $F_{1:m}$ is generically defined by~\eqref{eq:Fjk}, and furthermore that, for any $t \in [1,T+1]$,
\begin{align*}
  \nonumber %\label{eq:Rmq:beta}
  \hat \Rb_{m,q}^{[b]}(t) &= \sup_{s \in [1,t]} \sup_{\bm x \in \R^d} | \hat \Hb_{m,q}^{[b]}(s,t,\bm x)|, \\
  \nonumber %\label{eq:Smq:beta}
  \hat \Sb_{m,q}^{[b]}(t) &= \sup_{s \in [1,t]} \int_{\R^d} \{ \hat \Hb_{m,q}^{[b]}(s,t,\bm x) \}^2 \dd F_{1:\ip{m't}}(\bm x), \\
  %\label{eq:det:mult}
  \hat \Tb_{m,q}^{[b]}(t) &= \int_1^t \int_{\R^d} \{ \hat \Hb_{m,q}^{[b]}(s,t,\bm x) \}^2 \dd F_{1:\ip{m't}}(\bm x) \dd s.
\end{align*}

The following result is proven in Appendix~\ref{sec:proof:thresh:estim}.

\begin{prop}
  \label{prop:mult}
  Assume that $H_0$ in~\eqref{eq:H0} holds and that, additionally, $\bm X_1,\dots,\bm X_n$ is a stretch from a stationary sequence $(\bm X_i)_{i \in \Z}$ of continuous $d$-dimensional random vectors whose strong mixing coefficients satisfy $\alpha_r^{\bm X} = O(r^{-a})$ for some $a > 3+3d/2$ as $r \to \infty$. If $\ell_m = O(m^{1/2-\eps})$ for some $0 < \eps < 1/2$, then
$$
(\Bb_m, \hat \Bb_m^{[1]}, \hat \Bb_m^{[2]}) \leadsto ( \Bb_C, \Bb_C^{[1]}, \Bb_C^{[2]} )
$$
in $\{ \ell^\infty([0,T+1] \times [0,1]^d) \}^3$, where $\Bb_m$ is defined in~\eqref{eq:Bbm}, $\hat \Bb_m^{[1]}$ and $\hat \Bb_m^{[2]}$ are defined in~\eqref{eq:hat:Bbmb}, $\Bb_C$ is the weak limit of $\Bb_m$ defined in Proposition~\ref{prop:H0}, and $\Bb_C^{[1]}$ and $\Bb_C^{[2]}$ are independent copies of $\Bb_C$.

As a consequence,
\begin{align*}
  &(\Rb_{m,q},\hat \Rb_{m,q}^{[1]}, \hat \Rb_{m,q}^{[2]}) \leadsto (\Rb_{C,q},\Rb_{C,q}^{[1]}, \Rb_{C,q}^{[2]}), &(\Sb_{m,q},\hat \Sb_{m,q}^{[1]}, \hat \Sb_{m,q}^{[2]}) \leadsto (\Sb_{C,q},\Sb_{C,q}^{[1]}, \Sb_{C,q}^{[2]}), \\
  &(\Tb_{m,q},\hat \Tb_{m,q}^{[1]}, \hat \Tb_{m,q}^{[2]}) \leadsto (\Tb_{C,q},\Tb_{C,q}^{[1]}, \Tb_{C,q}^{[2]}), \qquad &\text{in } \{\ell^\infty([1,T+1])\}^3,
\end{align*}                                                                         where $\hat \Rb_{m,q}^{[b]}$, $\hat \Sb_{m,q}^{[b]}$ and $\hat \Tb_{m,q}^{[b]}$, $b \in \{1,2\}$, are defined in~\eqref{eq:det:mult}, $\Rb_{C,q}$, $\Sb_{C,q}$ and $\Tb_{C,q}$ are defined in~\eqref{eq:weak:limits}, and $\Rb_{C,q}^{[b]}$, $\Sb_{C,q}^{[b]}$ and $\Tb_{C,q}^{[b]}$, $b \in \{1,2\}$, are independent copies of $\Rb_{C,q}$, $\Sb_{C,q}$ and $\Tb_{C,q}$, respectively.
\end{prop}

The last claims of the previous proposition along with the last claim of Proposition~\ref{prop:H0} and Proposition~\ref{prop:H1} are the assumptions of Theorem~\ref{thm:thresh:boot} for $\Db_m \in \{\Rb_{m,q}, \Sb_{m,q}, \Tb_{m,q}\}$. It follows that the sequential change-point detection procedures based on these detector functions carried out as explained in Section~\ref{sec:boot:estim} using the above dependent multiplier replicates are asymptotically valid under $H_0$ in~\eqref{eq:H0} and sequences of alternatives related to $H_1$ in~\eqref{eq:H1}. Note that, in practice, since in the considered approach $m'= \ip{ (m/n) m}$ and $m$ play the role of $m$ and $n$, respectively, the largest possible value for $p$, the number of steps of the estimated threshold function $\tau_m^B$, is $m - m'$ and, in this case, each of the $p$ estimated thresholds covers approximately $\ip{n/m}$ time steps in the monitoring.

\section{Monte Carlo experiments}
\label{sec:MC}

Large-scale Monte Carlo experiments were carried out to investigate the finite-sample properties of the studied sequential change-point detection procedures. The aim was in particular to try to answer the following questions:
\begin{itemize}

\item How well do the procedures hold their level, in particular, when the threshold functions are estimated using the dependent multiplier bootstrap of Section~\ref{sec:thresh:time:series}?

\item What is the influence of the number of steps $p$ of the estimated threshold function  (see Sections~\ref{sec:thresh:asm:val:formulation} and~\ref{sec:estim:thresh}) on the distribution of the false alarms?

\item What is the effect of $p$ on the power and the \emph{mean detection delay} (the latter is the expectation under $H_1$ of the difference between the time at which the change was detected and the time $k^\star$ at which the change really occurred)?

\item What is the effect of the parameter $\gamma$ appearing in the expression of the weight function $q$ defined in~\eqref{eq:q} on the power and mean detection delay?

\item How do the detectors $R_{m,q}$, $S_{m,q}$, $T_{m,q}$, $P_m$ and $Q_m$ defined in~\eqref{eq:Rmq},~\eqref{eq:Smq},~\eqref{eq:Tmq},~\eqref{eq:Pm} and~\eqref{eq:Qm}, respectively, compare in terms of power and mean detection delay?

\item How do the derived procedures compare with similar, more specialized procedures in terms of power and mean detection delay?

\end{itemize}

We tried to answer these questions in detail in the univariate independent case when the estimation of the threshold functions of the sequential change-point detection procedures can be rightfully so based on the Monte Carlo approach described in Sections~\ref{sec:MC:estim} and~\ref{sec:thresh:ind:univ}. When the observations to be monitored are not univariate or independent so that resampling as described in Section~\ref{sec:thresh:time:series} is needed for the estimation of the threshold functions, we essentially investigated how well the procedures hold their level depending on the underlying data generating mechanism. Although many other questions could be formulated given the complexity of the problem, the following experiments should allow the reader to grasp the main finite-sample properties of the studied procedures.

\subsection{Monte Carlo estimation in the independent univariate case}

As already discussed in Section~\ref{sec:MC:estim}, the estimation of the threshold functions when monitoring independent univariate observations can be made arbitrarily precise by increasing the number $M$ of Monte Carlo samples. We used the setting $M=10^5$ in our experiments and estimated all the rejection percentages from $10^4$ samples. The change-point detection procedures were always carried out at the $\alpha = 5\%$ nominal level. %Note that, under the assumption of serial independence for univariate data, no learning sample is required to carry out the considered sequential tests.

%\input{R/sim/H0/H0sim}
\input{H0sim}

\begin{figure}
\begin{center}
  \includegraphics*[width=0.75\linewidth]{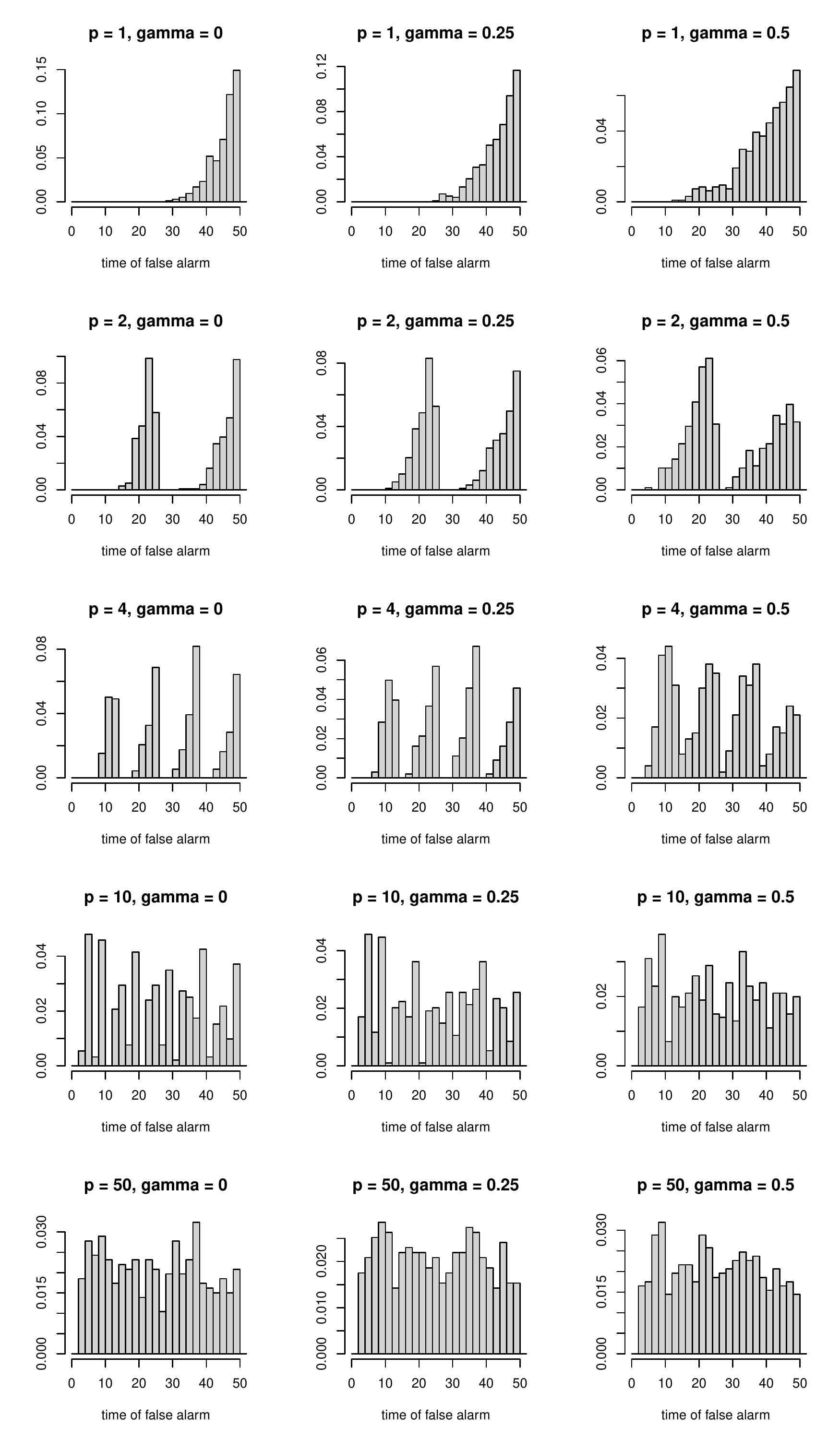}
  \caption{\label{fig:DistFalseAlarms_mmk_m_50_n_100} For the procedure based on $R_{m,q}$ in~\eqref{eq:Rmq} and Monte Carlo estimation with $M=10^5$, distribution of the false alarms (that is, of the time of rejection of $H_0$ in~\eqref{eq:H0}) for $m=50$ and $n=100$ obtained from $10^4$ samples of size $n$ from the standard uniform distribution. The histograms for $S_{m,q}$ in~\eqref{eq:Smq} and $T_{m,q}$ in~\eqref{eq:Tmq} are very similar.}
\end{center}
\end{figure}

\subsubsection{Under the null} Unsurprisingly, all the studied tests were found to hold their level very well as can for instance be seen by inspecting the rejection percentages reported in Table~\ref{tab:H0:sim}. Furthermore, as could have been expected given the fact that the studied detector functions have a tendency to be increasing on average, it was observed that the setting $p=1$ resulted in a concentration of false alarms at the end of the monitoring period, while the larger~$p$, the more uniform the distribution of the false alarms over the monitoring period. These unsurprising findings are for instance illustrated in Figure~\ref{fig:DistFalseAlarms_mmk_m_50_n_100} for the procedure based on $R_{m,q}$ in~\eqref{eq:Rmq}.

\begin{figure}[t!]
\begin{center}
  \includegraphics*[width=0.7\linewidth]{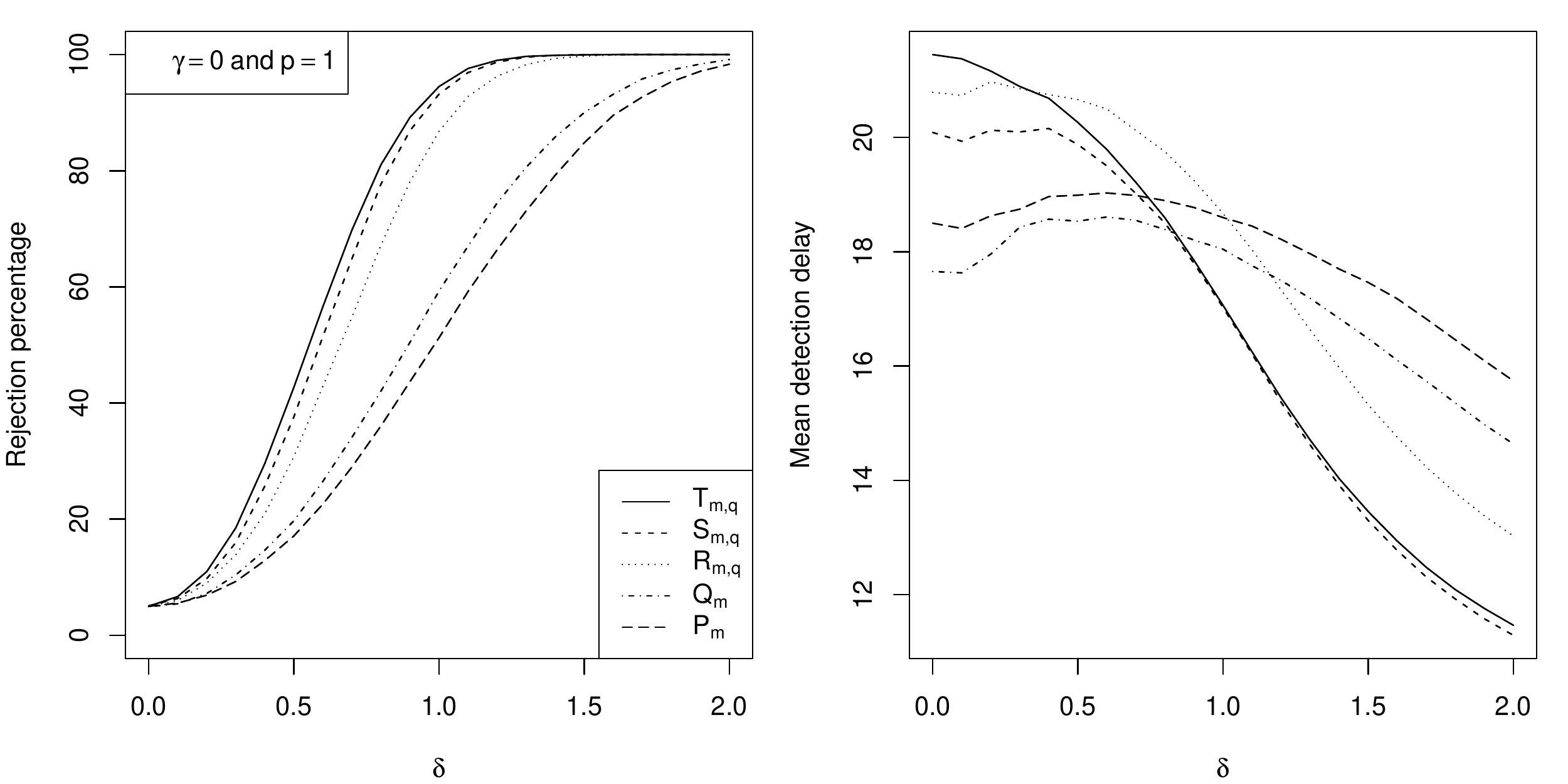} 
  \caption{\label{fig:CompDet_m_50_p_1} Left: estimated rejection probabilities of $H_0^\text{ind}$ in~\eqref{eq:H0:ind} under $H_1$ in~\eqref{eq:H1} with $m=50$, $k^\star = 75$, $n = 100$, $F$ the d.f.\ of the standard normal and $G$ the d.f.\ of the $N(\delta,1)$. Right: corresponding mean detection delays. The value of $\gamma$ in~\eqref{eq:q} is 0. The number of steps in the threshold functions is $p=1$.}
\end{center}
\end{figure}

\subsubsection{Change in mean} To answer the aforementioned questions related to the behavior of the procedures under $H_1$ in~\eqref{eq:H1}, we first considered a simple experiment consisting of a change in the expectation of a normal distribution. Specifically, $m$, $k^\star$, $n$, $F$ and $G$ in $H_1$ were taken equal to 50, 75, 100, the d.f.\ of the standard normal and the d.f.\ of the $N(\delta,1)$, respectively. The left graph in Figure~\ref{fig:CompDet_m_50_p_1} displays the estimated rejection percentages for the five detectors $R_{m,q}$, $S_{m,q}$, $T_{m,q}$, $P_m$ and $Q_m$ with $\gamma$ in~\eqref{eq:q} set to zero and $p=1$. The right graph represents the corresponding mean detection delays which were estimated only from the samples for which neither a false alarm was obtained (which occurs when the detector function becomes larger than the threshold function before the time of change $k^\star = 75$) nor the change was undetected (which occurs when the detector function remains below the threshold function during the entire monitoring period). Because the number of steps in the threshold function was set to $p=1$, the left graph of Figure~\ref{fig:CompDet_m_50_p_1} is directly comparable with the top left graph given in Figure~1 of \cite{DetGos19}. An inspection of the latter seems to indicate that the powers of the procedures based on $R_{m,q}$, $S_{m,q}$ and $T_{m,q}$ are not substantially different from those of the mean-specialized procedures considered in Section~5.1 of \cite{DetGos19}, even though the detectors $R_{m,q}$, $S_{m,q}$ and $T_{m,q}$ are not specifically designed to be sensitive to changes in the expectation. The graphs  are not substantially different for other values of $\gamma$ and $p$. Overall, the procedures based on $R_{m,q}$, $S_{m,q}$ and $T_{m,q}$ were always observed to be more powerful and superior in terms of mean detection delay than those based on $P_m$ and $Q_m$. The latter is in full accordance with the empirical observations of \cite{DetGos19} for more specialized procedures. Note that the procedure based on $T_{m,q}$ seems the most powerful for the alternative under consideration.

\begin{figure}[t!]
\begin{center}
  \includegraphics*[width=0.7\linewidth]{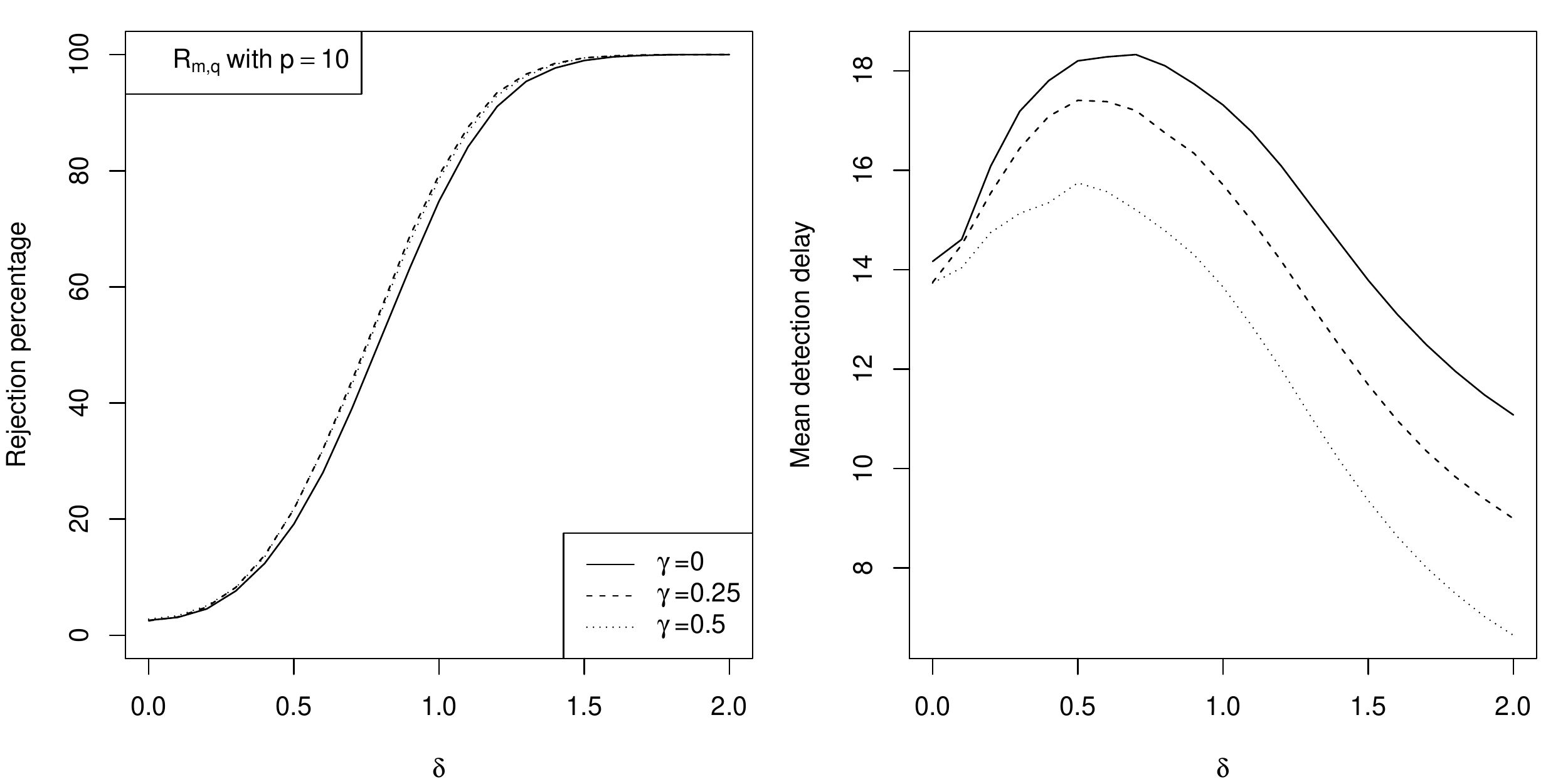} 
  \caption{\label{fig:GammaEffect_mmk_m_50} Left: estimated rejection probabilities of $H_0^\text{ind}$ in~\eqref{eq:H0:ind} under $H_1$ in~\eqref{eq:H1} with $m=50$, $k^\star = 75$, $n = 100$, $F$ the d.f.\ of the standard normal and $G$ the d.f.\ of the $N(\delta,1)$ for the procedure based on $R_{m,q}$ with $p=10$. Right: corresponding mean detection delays.}
\end{center}
\end{figure}

Figure~\ref{fig:GammaEffect_mmk_m_50} highlights the influence of the parameter $\gamma$ in~\eqref{eq:q} on the rejection percentages and the mean detection delays of the procedure based on $R_{m,q}$ with $p=10$. The graphs are very similar for other values of $p$ or for the procedures based on $S_{m,q}$ and $T_{m,q}$. The conclusion is the same for all three procedures. For a change in expectation, while the parameter $\gamma$ does not seem to affect the power of the procedures much, it has a clear influence on the mean detection delay: the greater $\gamma$, the shorter the mean detection delay. As we shall see, this behavior is not true for all types of alternatives.

\begin{figure}[t!]
\begin{center}
  \includegraphics*[width=0.7\linewidth]{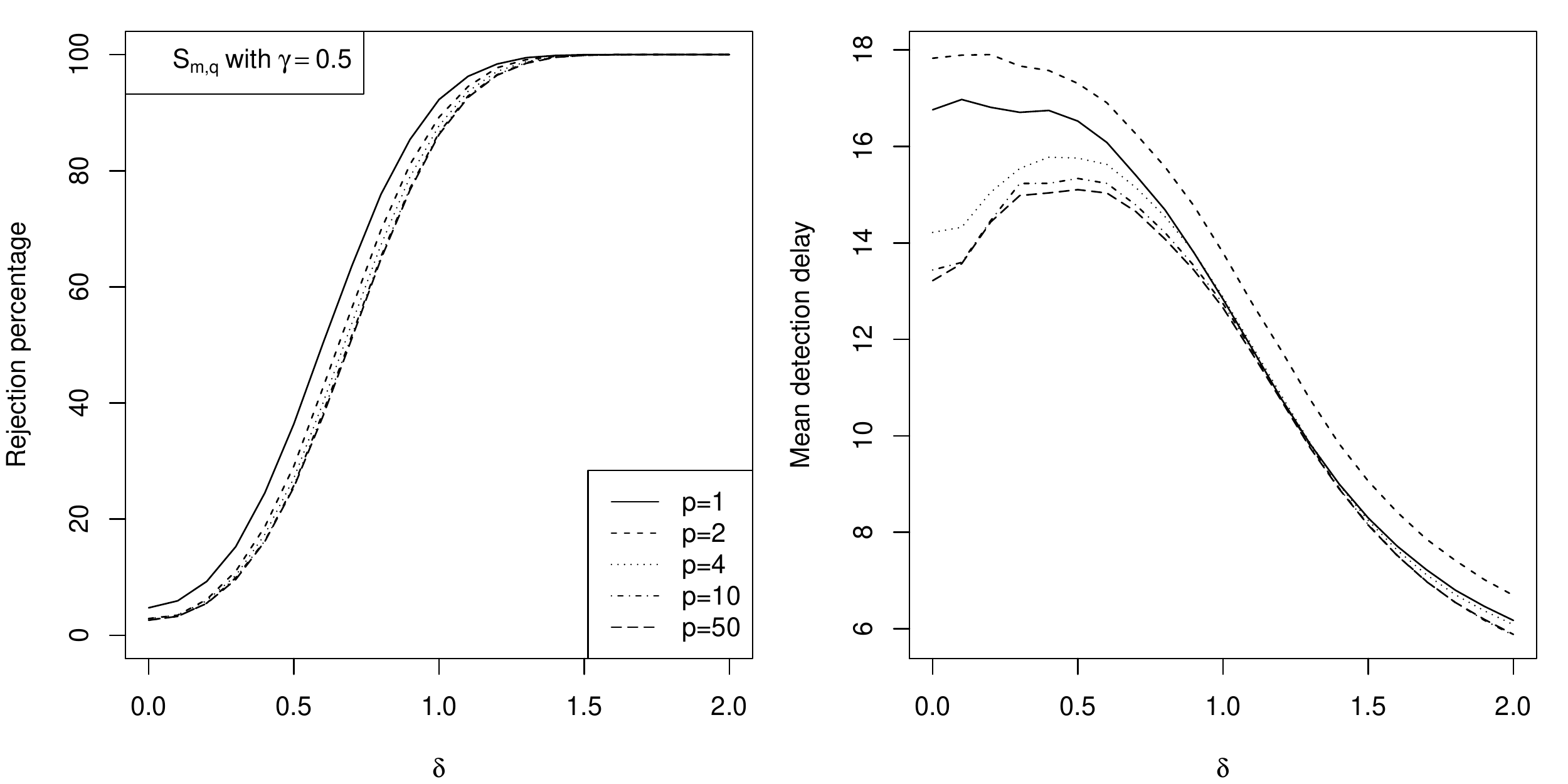} 
  \caption{\label{fig:pEffect_mmc_m_50_gamma_0.5} Left: estimated rejection probabilities of $H_0^\text{ind}$ in~\eqref{eq:H0:ind} under $H_1$ in~\eqref{eq:H1} with $m=50$, $k^\star = 75$, $n = 100$, $F$ the d.f.\ of the standard normal and $G$ the d.f.\ of the $N(\delta,1)$ for the procedure based on $S_{m,q}$ with $\gamma = 0.5$. Right: corresponding mean detection delays.}
\end{center}
\end{figure}

Figure~\ref{fig:pEffect_mmc_m_50_gamma_0.5} displays the influence of the number of steps $p$ of the threshold function on the rejection percentages and the mean detection delays of the procedure based on $S_{m,q}$ with $\gamma = 0.5$. The graphs are not qualitatively different for other values of $\gamma$ or for the procedures based on $R_{m,q}$ and $T_{m,q}$. Overall, the procedures with $p=1$ have the highest rejection percentages. The latter is due to the fact that, because $k^\star=75$, detections occur mostly at the end of the monitoring period, and, at the end of the monitoring interval, by construction, the threshold functions for $p=1$ are below the corresponding threshold functions obtained for larger values of $p$.

\subsubsection{Change in variance} The setting is similar to that of the previous experiment except that, this time, it is the variance of the normal distribution that changes from 1 to $\delta^2$. The left graph in Figure~\ref{fig:CompDet_m_50_p_50:var} displays the estimated rejection percentages for the procedures based on $R_{m,q}$, $S_{m,q}$, $T_{m,q}$, $P_m$ and $Q_m$ with $\gamma = 0.5$ and $p=50$. The graph on the right represents the corresponding mean detection delays. Again, the procedures based on $R_{m,q}$, $S_{m,q}$ and $T_{m,q}$ appear to be substantially more powerful and superior in terms of mean detection delay than those based on $P_m$ and $Q_m$. The conclusion remains true for all values of $\gamma$ and $p$. The influence of $\gamma$ and $p$ is of the same nature as in the case of a change in mean: the greater $\gamma$, the shorter the mean detection delay and the lower~$p$, the higher the power.

\begin{figure}[t!]
\begin{center}
  \includegraphics*[width=0.7\linewidth]{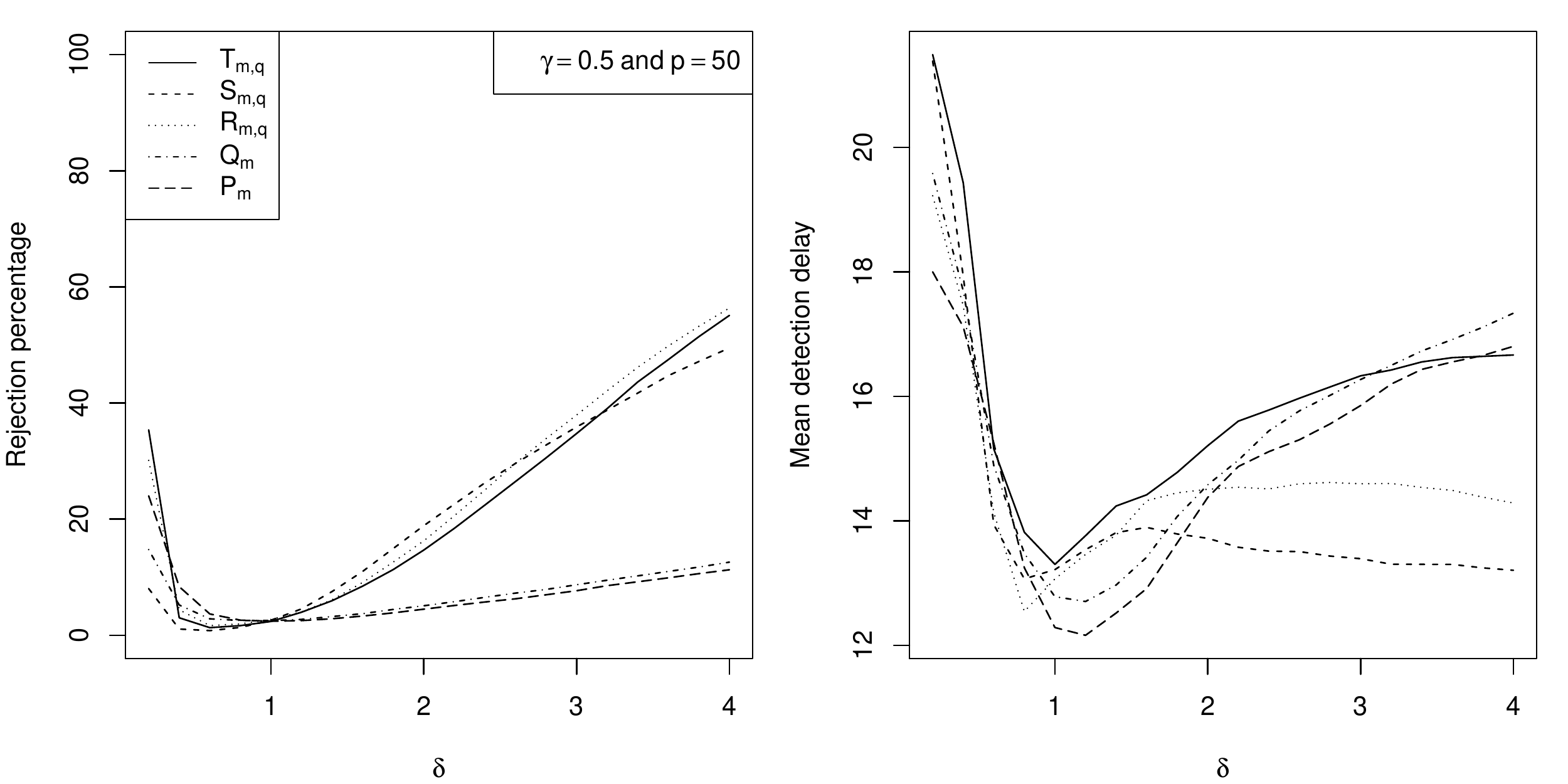} 
  \caption{\label{fig:CompDet_m_50_p_50:var} Left: estimated rejection probabilities of $H_0^\text{ind}$ in~\eqref{eq:H0:ind} under $H_1$ in~\eqref{eq:H1} with $m=50$, $k^\star = 75$, $n = 100$, $F$ the d.f.\ of the standard normal and $G$ the d.f.\ of the $N(0,\delta^2)$. Right: corresponding mean detection delays. The value of $\gamma$ in~\eqref{eq:q} is 0.5. The number of steps in the threshold functions is $p=50$.}
\end{center}
\end{figure}

\subsubsection{Change in distribution}  As a final experiment for independent univariate observations, we considered a change in distribution that keeps the expectation and the variance constant. Specifically, $F$ and $G$ in $H_1$ were taken equal to the d.f.\ of the $N(1,2)$ and the d.f.\ of the Gamma distribution whose shape and rate parameters are both equal to 1/2, respectively. The parameter $m$ was taken to be in $\{50,100\}$, $n$ was set to $2m$ and the parameter $k^\star$ to $\ip{nc}$ with $c \in \{0.51,0.56, \dots, 0.96\}$.

\begin{figure}[t!]
\begin{center}
  \includegraphics*[width=0.7\linewidth]{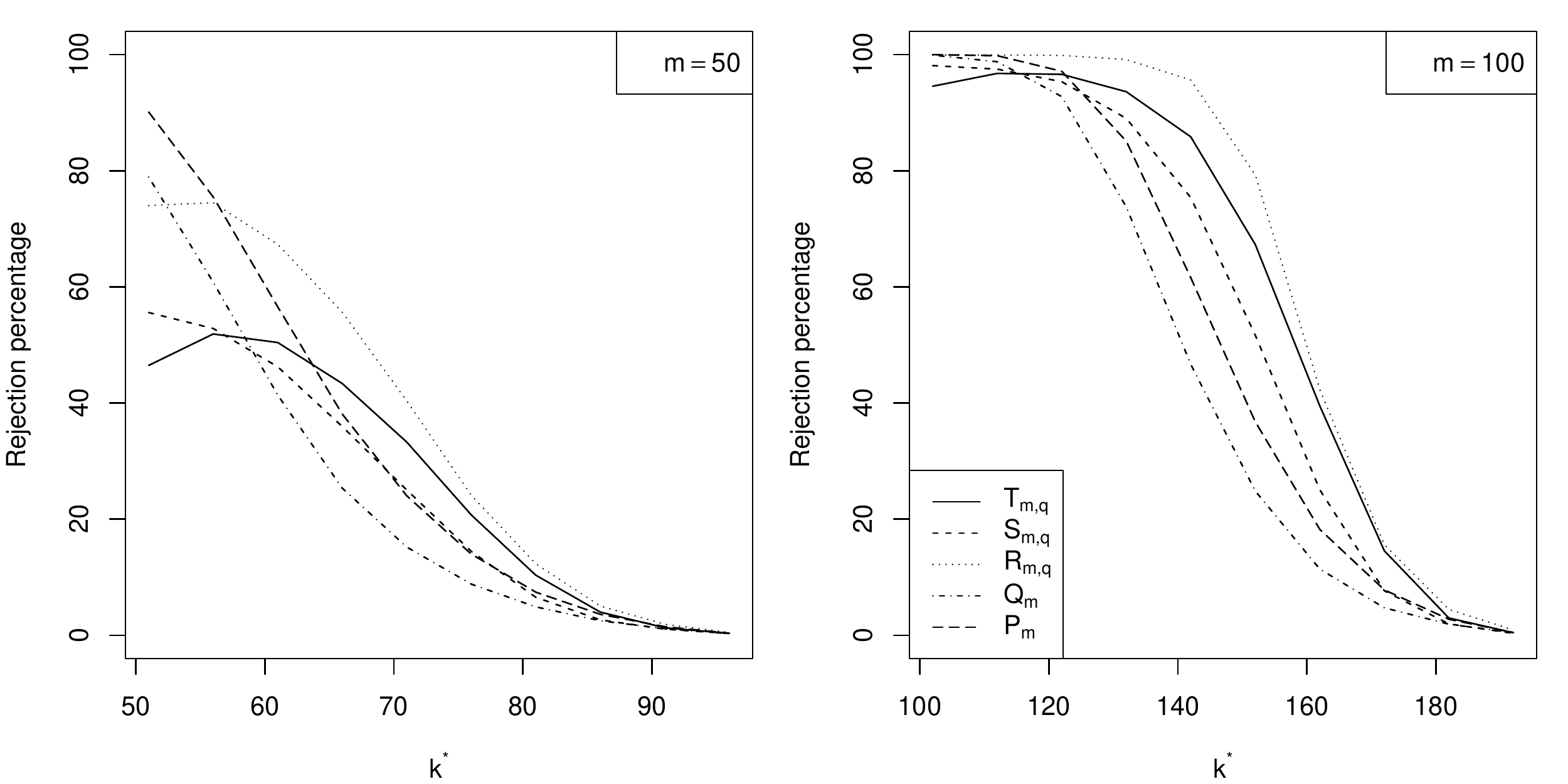} 
  \caption{\label{fig:CompDet_p_50:dist} Estimated rejection probabilities of $H_0^\text{ind}$ in~\eqref{eq:H0:ind} under $H_1$ in~\eqref{eq:H1} with $F$ the d.f.\ of  the $N(1,2)$ and $G$ the d.f.\ of the Gamma distribution whose shape and rate parameters are both equal to 1/2. Left: $m=50$ and $n=100$. Right: $m=100$ and $n=200$. The value of $\gamma$ in~\eqref{eq:q} is 0 and the number of steps in the threshold functions is $p=50$.}
\end{center}
\end{figure}

Figure \ref{fig:CompDet_p_50:dist} shows the rejection percentages of $H_0^\text{ind}$ in~\eqref{eq:H0:ind} against $k^\star$ for the procedures based on $R_{m,q}$, $S_{m,q}$, $T_{m,q}$, $P_m$ and $Q_m$ with $\gamma = 0$ and $p=50$. The graphs for other values of $\gamma$ and $p$ are very similar. As one can see, the procedures based on $P_m$ and $Q_m$ appear to be the most powerful when the change occurs at the beginning of the monitoring period. The latter could have been expected from the definition of the underlying detectors and the simulation results of \cite{DetGos19}.  As $k^\star$ increases, the procedures based on $R_{m,q}$ and $T_{m,q}$ become more powerful.

\begin{figure}[t!]
\begin{center}
  \includegraphics*[width=0.7\linewidth]{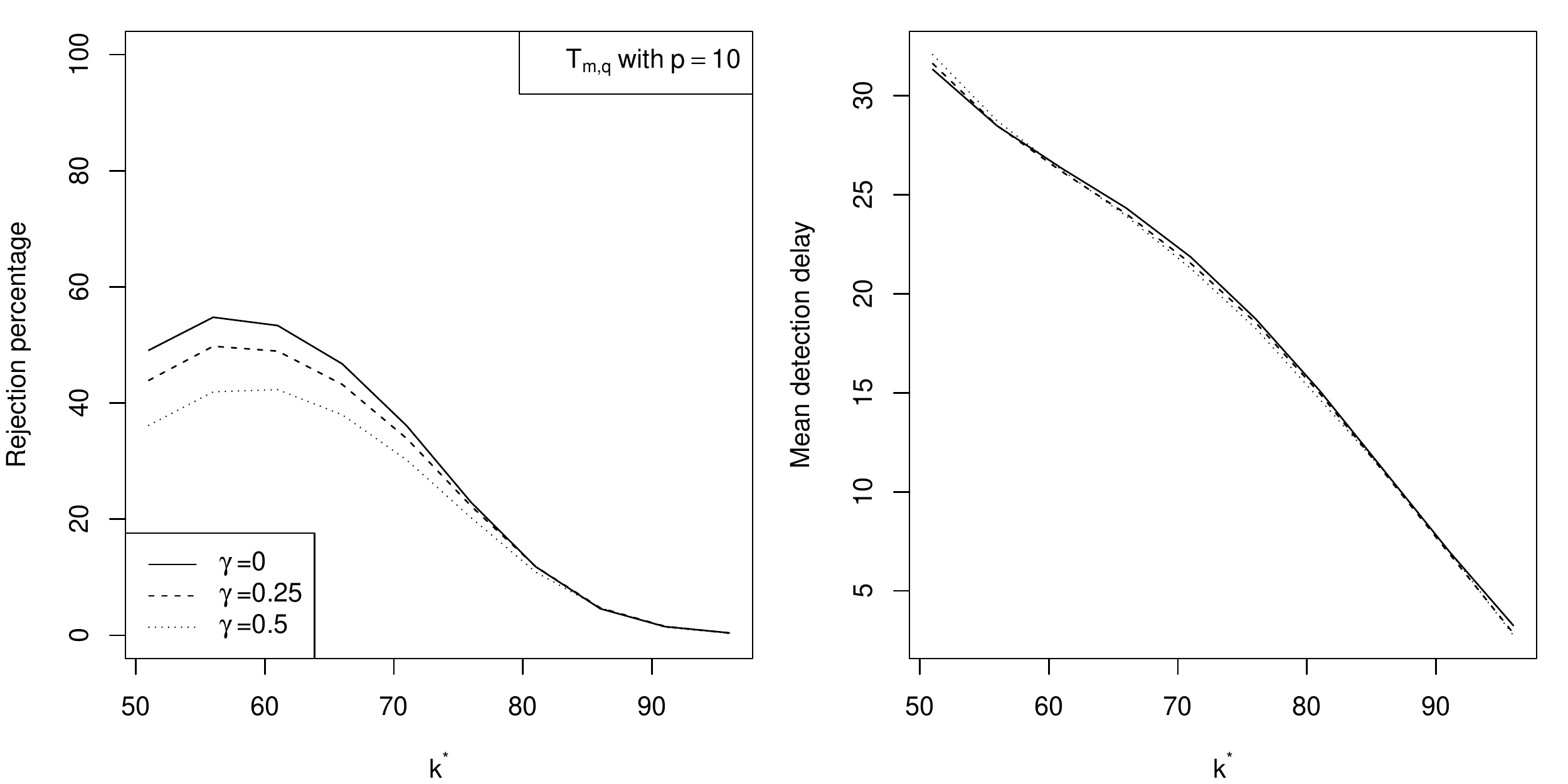} 
  \caption{\label{fig:GammaEffect_mac_m_50:dist} Left: estimated rejection probabilities of $H_0^\text{ind}$ in~\eqref{eq:H0:ind} under $H_1$ in~\eqref{eq:H1} with $m=50$, $n = 100$, $F$ the d.f.\ of  the $N(1,2)$ and $G$ the d.f.\ of the Gamma distribution whose shape and rate parameters are both equal to 1/2 for the procedure based on $T_{m,q}$ with $p=10$. Right: corresponding mean detection delays.}
\end{center}
\end{figure}

Figure \ref{fig:GammaEffect_mac_m_50:dist} displays the influence of the parameter $\gamma$ on the rejection percentages of the procedure based on $T_{m,q}$ with $p=10$. The graphs are not qualitatively different for other values of $p$ and $m$ or for the procedures based on $R_{m,q}$ and $S_{m,q}$. As one can see, unlike for a change in expectation, $\gamma$ has hardly no influence on the mean detection delay and it is the setting $\gamma = 0$ that leads to the highest rejection percentages.

\begin{figure}[t!]
\begin{center}
  \includegraphics*[width=0.7\linewidth]{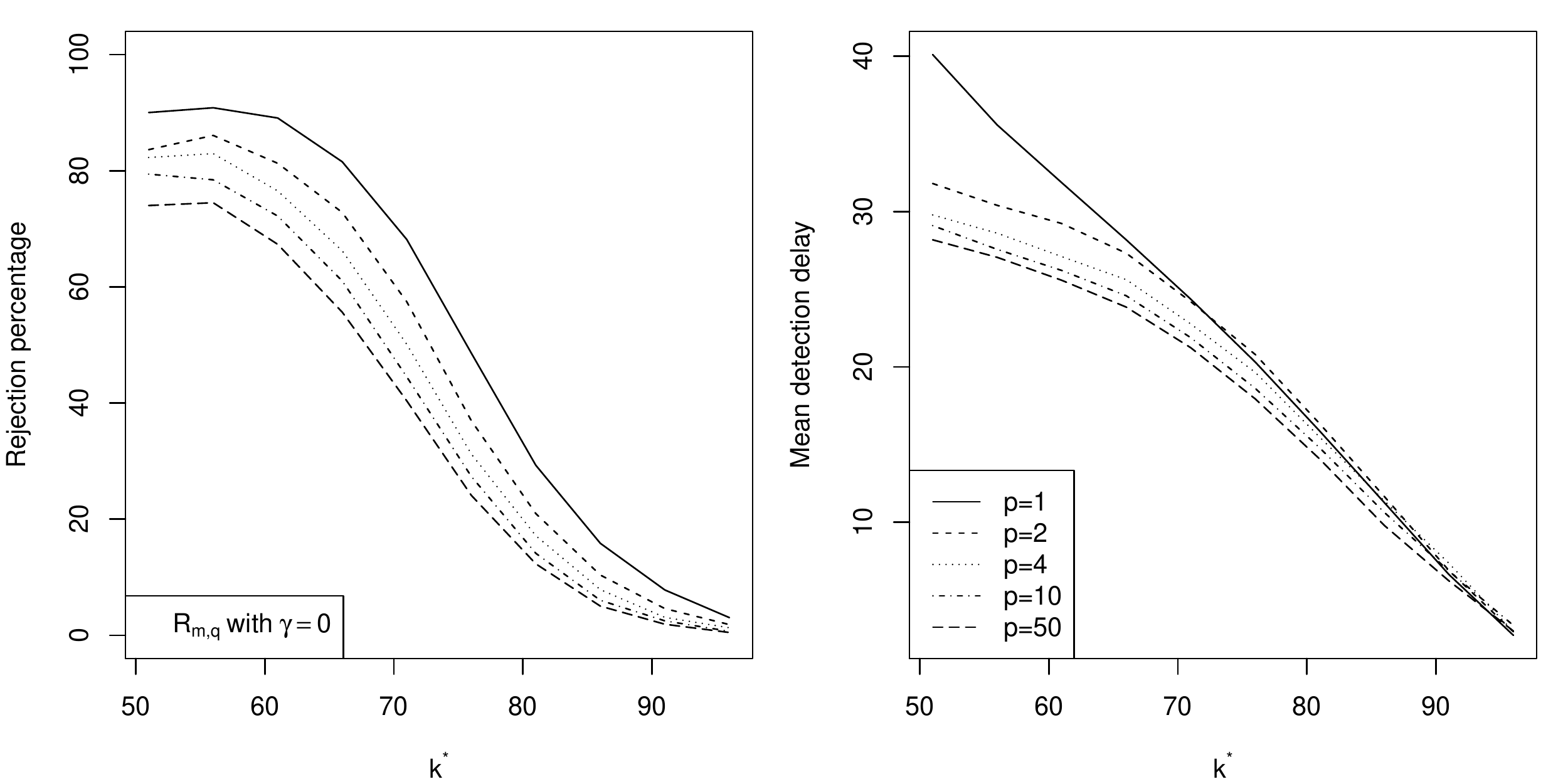} 
  \caption{\label{fig:pEffect_m_50_gamma_0:dist}  Left: estimated rejection probabilities of $H_0^\text{ind}$ in~\eqref{eq:H0:ind} under $H_1$ in~\eqref{eq:H1} with $m=50$, $n = 100$, $F$ the d.f.\ of  the $N(1,2)$ and $G$ the d.f.\ of the Gamma distribution whose shape and rate parameters are both equal to 1/2 for the procedure based on $R_{m,q}$ with $\gamma=0$. Right: corresponding mean detection delays.}
\end{center}
\end{figure}

Finally, Figure \ref{fig:pEffect_m_50_gamma_0:dist} shows the influence of $p$ for the procedure based on $R_{m,q}$ with $\gamma = 0$. The graphs are not qualitatively different for other values of $\gamma$ or for the procedures based on $S_{m,q}$ and $T_{m,q}$. As in previous experiments, the setting $p=1$ leads in the highest rejection percentages. However,  when the change occurs in the first third of the monitoring period, the mean detection delay for $p = 1$ is clearly substantially larger than for $p > 1$. Additional simulations show that the larger the monitoring period, the more pronounced this phenomenon.

\subsection{Dependent multiplier bootstrap-based estimation in the time series case}

The threshold function estimation approach based on the dependent multiplier bootstrap described in Section~\ref{sec:thresh:time:series} can in principle be used as soon as the observations to be monitored are either multivariate or serially dependent. We used the setting $B=2000$ in our experiments and estimated all the rejection percentages from 1000 samples at the $\alpha = 5\%$ nominal level.

%\input{R/mult/H0/H0multARTmq}
\input{H0multARTmq}

\subsubsection{Under the null} One of the most important practical aspects is to assess how well the procedures hold their level when based on the dependent multiplier bootstrap. To attempt to answer this question, we conducted extensive simulations in the univariate case. For $m \in \{50,100,200,400\}$ and $T \in \{0.5,1,2,3\}$, we generated samples of size $n = \ip{m(T+1)}$ from an AR(1) model with normal innovations and autoregressive parameter $\beta \in \{0,0.1,0.3,0.5\}$, and estimated the levels of the procedures based on $R_{m,q}$, $S_{m,q}$ and $T_{m,q}$ with $\gamma \in \{0,0.25,0.5\}$ and $p \in \{1,2,4\}$. The rejection percentages of $H_0$ in~\eqref{eq:H0} for $T_{m,q}$ are given in Table~\ref{tab:H0:multARTmq} (the missing entries in the table correspond to parameter settings for which computations took too long given our computer cluster resources). As one can see, unsurprisingly, the larger $\beta$, the more liberal the procedure based on $T_{m,q}$ tends to be. This phenomenon is particularly visible for $\beta \in \{0.3, 0.5\}$. Reassuringly however, for a given $\beta$, $T$, $p$ and $\gamma$, the estimated levels seem to get closer to the 5\% nominal level as $m$ increases. Hence, unsurprisingly, the stronger the serial dependence, the larger $m$ needs to be so that the procedure can be expected to hold its nominal level. For $\beta  \in \{0.3, 0.5\}$ in particular and keeping $T$, $\gamma$ and $m$ fixed, we also see that the larger $p$, the more liberal the procedure based on $T_{m,q}$ tends to be. The latter could be explained by the fact that, as $p$ increases, more thresholds need to be estimated, and that, except for the first, all the estimated thresholds are conditional empirical quantiles: the precision of the estimation of a threshold thus critically depends on the precision of the previously estimated thresholds with respect to which conditioning is performed. In other words, the fact that the empirical levels tend to become higher when $p$ increases could be explained by an error propagation effect. Finally, for $\beta$, $T$, $p$ and $m$ fixed, we also see that the larger $\gamma$, the lower the rejection percentages tend to be. Overall, for $\beta \geq 0.1$, the procedure based on $T_{m,q}$ holds its level best for $\gamma = 0.5$.

The conclusions for the procedures based on $R_{m,q}$ and $S_{m,q}$ are very similar, with the exception that the effect of $\gamma$ seems ``stronger'': while the empirical levels are better for $\gamma = 0.25$ than for $\gamma = 0$, the procedures become way too conservative for $\gamma = 0.5$. The latter effect might be due to the fact the detectors $R_{m,q}$ and $S_{m,q}$ involve maxima (unlike $T_{m,q}$ which involves means) and to our too low setting of the constant $\delta$ in the definition of the weight function~\eqref{eq:q}. The latter was arbitrarily set to $10^{-4}$ (and had \emph{de facto} no effect in our Monte Carlo experiments given the values of $m$ that we considered).

%\input{R/mult/H0/H0multGARCHTmq}
\input{H0multGARCHTmq}

A similar experiment was conducted by generating samples from a GARCH(1,1) model with parameters $\omega = 0.012$, $\beta = 0.919$ and $\alpha = 0.072$ to mimic SP500 daily log-returns following \cite{JonPooRoc07}. The empirical levels for the procedure based on $T_{m,q}$ are reported in Table~\ref{tab:H0:multGARCHTmq} and appear to be closest to the 5\% nominal overall when $\gamma = 0.5$.

%\input{R/mult/H0/H0multcop}
\input{H0multcop}

Finally, a bivariate experiment with independent observations consisting of generating samples of size $n = 2m$ for $m \in \{50, 100, 200\}$ from a normal copula with a Kendall's tau of $\tau \in \{-0.6, -0.3, 0, 0.3, 0.6 \}$  was carried out. The empirical levels for the procedure based on $T_{m,q}$ are reported in Table~\ref{tab:H0:multcop}. The effect of $\gamma$ appears as in the previous experiments. For fixed $\gamma$ and $p$, it can also be observed that the procedure has a tendency of being too conservative in the case of strong negative dependence but, reassuringly, the agreement with the 5\% nominal level seems to improve as $m$ increases.

\begin{figure}[t!]
\begin{center}
  \includegraphics*[width=0.7\linewidth]{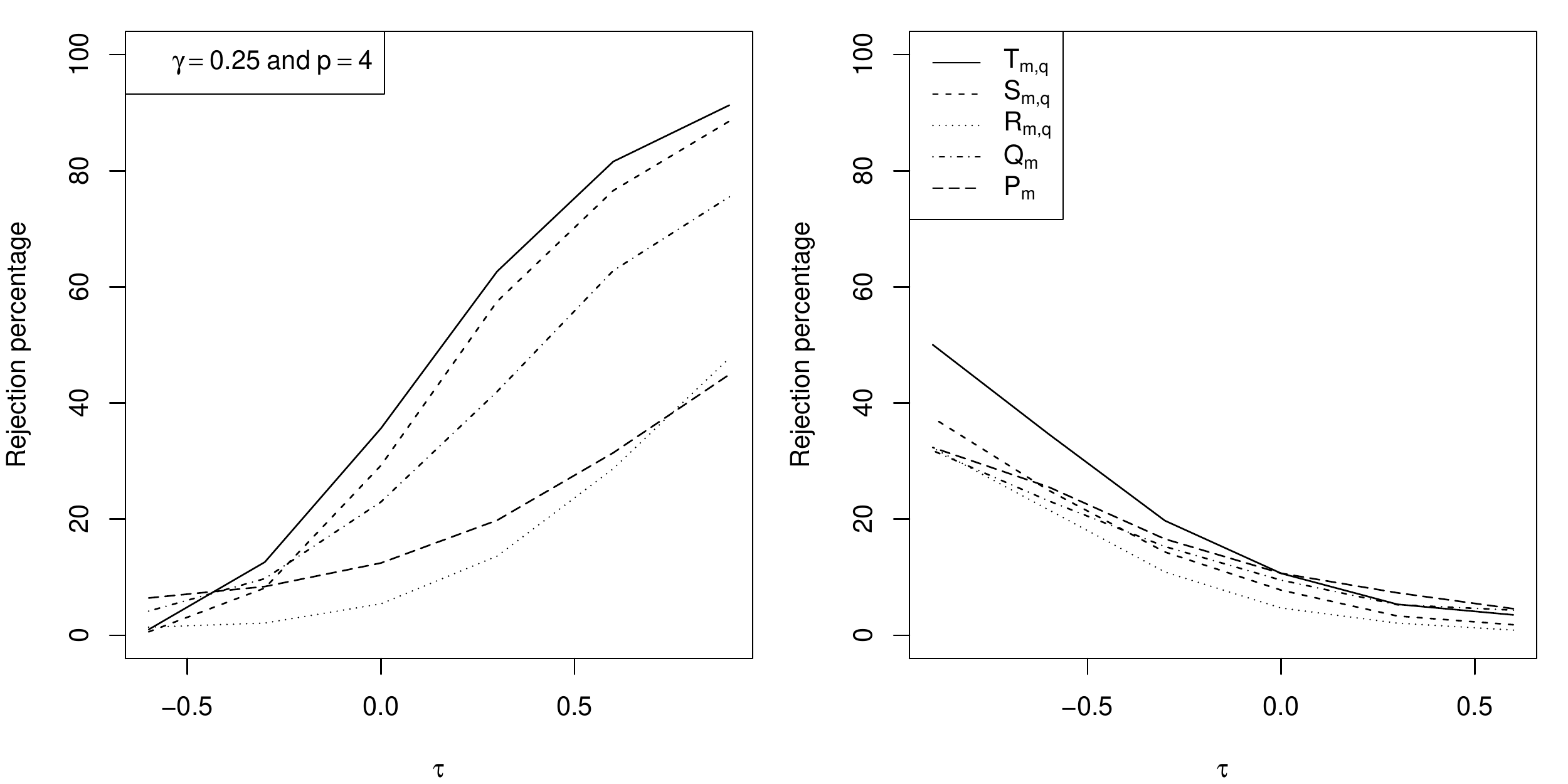}
  \caption{\label{fig:copula} Left: estimated rejection probabilities of $H_0^\text{ind}$ in~\eqref{eq:H0:ind} under $H_1$ in~\eqref{eq:H1} with $m=50$, $k^\star = 75$, $n = 100$, $F$ the bivariate normal copula with a Kendall's tau of -0.6 and $G$ the bivariate normal copula with a Kendall's tau of $\tau$. Right: estimated rejection probabilities of $H_0^\text{ind}$ in~\eqref{eq:H0:ind} under $H_1$ in~\eqref{eq:H1} with $m=50$, $k^\star = 75$, $n = 100$, $F$ the bivariate normal copula with a Kendall's tau of 0.6 and $G$ the bivariate normal copula with a Kendall's tau of $\tau$. The value of $\gamma$ in~\eqref{eq:q} is 0.25. The number of steps in the threshold functions is $p=4$.}
\end{center}
\end{figure}

\subsubsection{Change in the copula parameter} To grasp further the finite-sample behavior of the procedures in the case of bivariate independent observations, we simulated a change in the parameter of a normal copula. The left (resp.\ right) graph in Figure~\ref{fig:copula} displays the estimated rejection probabilities of $H_0^\text{ind}$ in~\eqref{eq:H0:ind} for the procedures based on $R_{m,q}$, $S_{m,q}$, $T_{m,q}$, $P_m$ and $Q_m$ with $\gamma = 0.25$ and $p=4$ under $H_1$ in~\eqref{eq:H1} with $m=50$, $k^\star = 75$, $n = 100$, $F$ the bivariate normal copula with a Kendall's tau of -0.6 (resp. 0.6) and $G$ the bivariate normal copula with a Kendall's tau of $\tau \in \{-0.6, 0.3, 0, 0.3, 0.6, 0.9\}$ (resp. $\tau \in \{-0.9, -0.6, -0.3, 0, 0.3, 0.6\}$). As one can notice, the procedure based on $T_{m,q}$ (resp.\ $P_m$) is always among the most (resp.\ least) powerful ones. %The lack of symmetry between the power curves in the left and right graphs might be due to the fact that empirical d.f.s estimate probabilities of lower-left orthants.
Graphs for other values of $\gamma$ and $p$ are not qualitatively different. As for all previous experiments, we observed that the smaller~$p$, the more powerful the procedures. For this experiment, the parameter $\gamma$ appeared to have a rather small impact on the rejection percentages of $T_{m,q}$.

\section{Data examples}

\begin{figure}[t!]
\begin{center}
  \includegraphics*[width=1\linewidth]{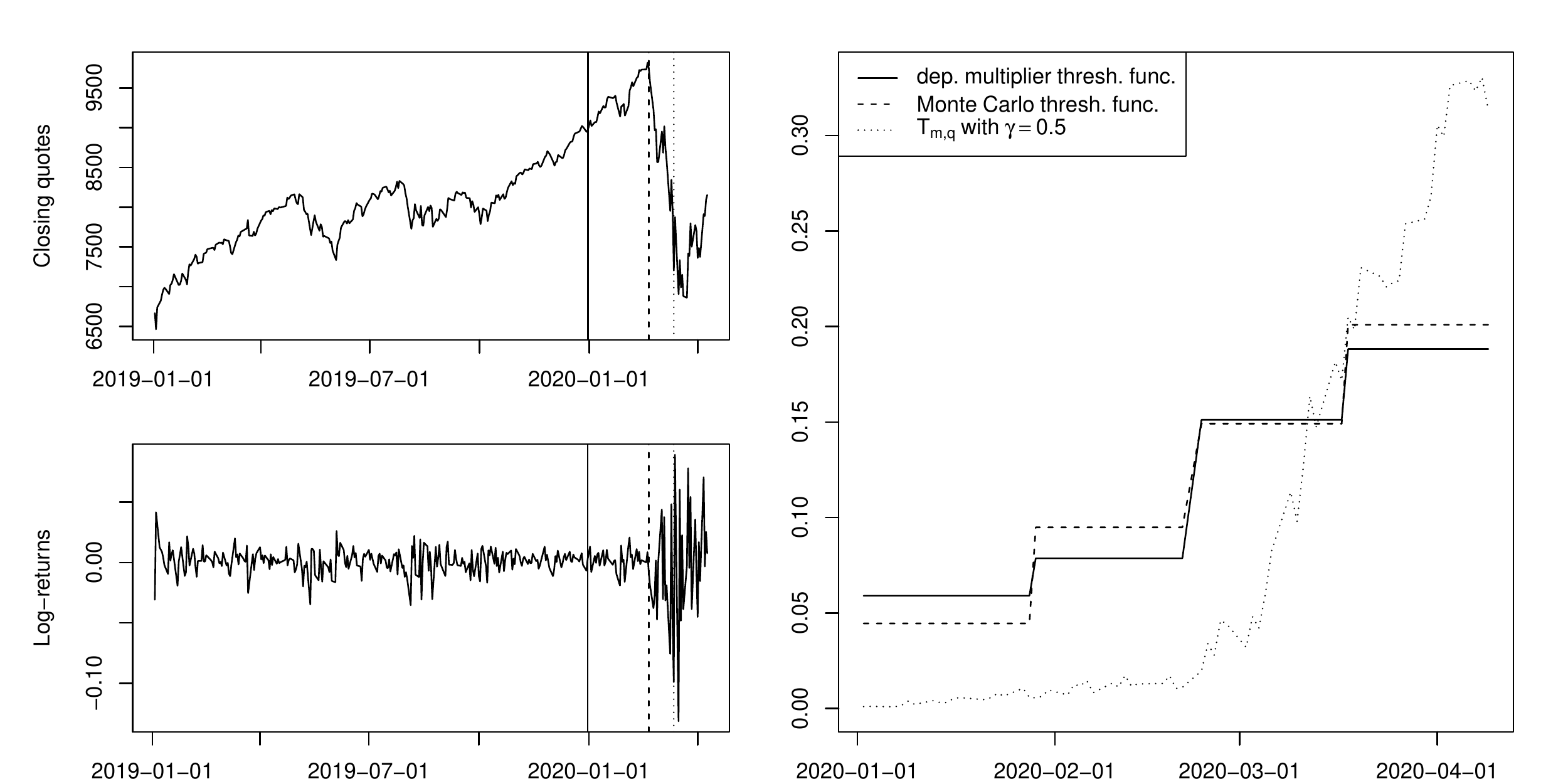} 
  \caption{\label{fig:det-thresh} Left: closing quotes and corresponding daily log-returns of the NASDAQ composite index for the period 2019-01-02 -- 2020-04-11. The solid vertical line represents the beginning of the monitoring. The dotted vertical line represent the date (2020-03-12) at which the detector function based on $T_{m,q}$ with $\gamma = 0.5$ first exceeded the two threshold functions on the right. The dashed vertical line corresponds to the estimated date of change (2020-02-21). Right: the dotted line represents the detector function based on $T_{m,q}$ with $\gamma = 0.5$. The solid (resp. dashed) line represents the threshold function with $p=4$ steps obtained using the dependent multiplier bootstrap (resp.\ Monte Carlo estimation).}
\end{center}
\end{figure}

To illustrate the use of the proposed sequential change-point detection tests, we considered two fictitious scenarios, the first (resp.\ second) of a univariate (resp.\ bivariate) nature based on closing quotes of the NASDAQ composite index (resp.\ Microsoft and Intel stocks) for the period 2019-01-02 -- 2020-04-11. The latter were obtained using the \texttt{get.hist.quote()} function of the \texttt{tseries} \textsf{R} package \citep{tseries}. In both scenarios, it was assumed that, on the last day of 2019, one wished to monitor the (univariate or bivariate) daily log-returns for a change in contemporary distribution possibly using the stretch of $m=251$ (univariate of bivariate) log-returns of 2019 as learning sample. The latter decision was confirmed after the tests of \cite{BucFerKoj19} implemented in the functions \texttt{stDistAutocop()} and \texttt{cpDist()} of the \texttt{npcp} \textsf{R} package \citep{npcp} provided no evidence against stationarity for the two candidate learning samples. Notice that, unsurprisingly, the rank-based test of serial independence of \cite{GenRem04} implemented in the function \texttt{serialIndepTest()} of the \texttt{copula} \textsf{R} package \citep{copula} provided weak evidence against the serial independence of the squared component time series.

The closing quotes and corresponding daily log-returns of the NASDAQ composite index are represented in the left panel of Figure~\ref{fig:det-thresh}. The solid vertical lines mark the beginning of the monitoring. The dotted line in the right panel represents the detector function based on $T_{m,q}$ with $\gamma = 0.5$. The latter was chosen given its overall good performance in our Monte Carlo experiments, both in terms of empirical level, power and mean detection delay. In the right panel, the solid line represents the threshold function with $p=4$ steps estimated using the dependent multiplier bootstrap with $B=10^5$, while the dashed line represents the threshold function with $p=4$ steps estimated using Monte Carlo with $M=10^5$. Note that the latter did not at all use the learning sample as it is computed under the assumption that the observations are serially independent. The relative proximity of the two threshold functions could be explained by the fact that, although present, serial dependence in the learning sample is probably very weak. The detector function exceeded the two threshold functions at the same date (2020-03-12), which is marked by the dotted vertical line in the left panel of Figure~\ref{fig:det-thresh} and corresponds to the 49th daily log-return of 2020. Given the definition of $T_{m,q}$ in~\eqref{eq:Tmq} and having that of $S_{m,q}$ in~\eqref{eq:Smq} in mind, a possible estimate of a point of change for an exceedance at position $k=m+49 = 300$ is given by
\begin{equation}
  \label{eq:point:change}
\mathrm{argmax}_{m \leq j \leq k-1} \frac{1}{k} \sum_{i=1}^k \left[ \frac{j (k-j)}{m^{3/2} q(j/m,k/m)}  \{ F_{1:j}(\bm X_i) - F_{j+1:k}(\bm X_i) \} \right]^2 + 1,
\end{equation}
which returned 286 and corresponds to the date 2020-02-21. The latter is marked by a dashed vertical line in the left panel of Figure~\ref{fig:det-thresh} and corresponds to the beginning of the sharp decrease of the NASDAQ composite index as a consequence of the Covid-19 pandemic.

\begin{figure}[t!]
\begin{center}
  \includegraphics*[width=1\linewidth]{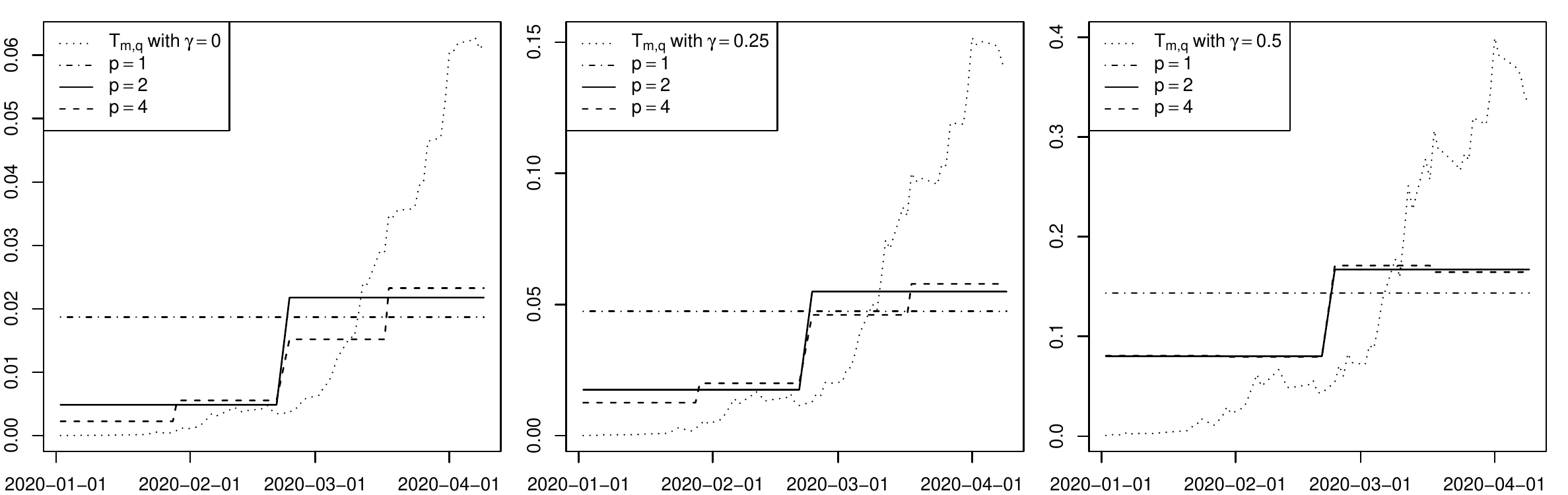} 
  \caption{\label{fig:biv-det-thresh} Monitoring of the bivariate daily log-returns of the Microsoft and Intel stocks for the period 2020-01-01 -- 2020-04-11 using the 2019 bivariate log-returns as learning sample, and the procedure based on $T_{m,q}$ with $\gamma \in \{0, 0.25, 0.5\}$ and a threshold function with $p \in \{1,2,4\}$ steps estimated using the dependent multiplier bootstrap. The estimated dates of change do not depend on $p$ and are the 2020-01-25 for $\gamma = 0$, the 2020-02-15 for $\gamma = 0.25$ and the 2020-02-20 for $\gamma = 0.5$.}
\end{center}
\end{figure}

Figure~\ref{fig:biv-det-thresh} describes the monitoring of the bivariate daily log-returns of the Microsoft and Intel stocks using the procedure based on $T_{m,q}$ with $\gamma \in \{0, 0.25, 0.5\}$ and a threshold function with $p \in \{1,2,4\}$ steps estimated using the dependent multiplier bootstrap. All the dates of exceedance are between the 2020-03-12 ($\gamma=0$ and $p=2$) and the 2020-03-09 ($\gamma = 0.5$ and $p=1$). The estimated dates of change turn out not to depend on $p$ and are the 2020-01-25 for $\gamma = 0$, the 2020-02-15 for $\gamma = 0.25$ and the 2020-02-20 for $\gamma = 0.5$. This effect of $\gamma$ on~\eqref{eq:point:change} was to be expected: larger values of $\gamma$ give more weight to potential break points close to $k$. 

\section{Concluding remarks}

In the context of closed-end sequential change-point detection, it can be argued \citep[see][]{AnaKos18} that it is desirable that the underlying threshold function is such that the probability of false alarm remains approximately constant over the monitoring period. In this work, the asymptotic validity of the bootstrap-based estimation of such a threshold function was established for generic detectors. The latter was applied to sequential change-point tests involving detectors based on differences of empirical d.f.s that can be either simulated or resampled using a dependent multiplier bootstrap depending on whether univariate independent or multivariate serially dependent observations are monitored. The proposed detectors are adaptations of statistics used in a posteriori change-point testing and include a weight function in the spirit of \cite{CsoSzy94b} that can be used to give more importance to recent observations.

Extensive Monte Carlo experiments were used to investigate the finite-sample properties of the resulting sequential change-point tests. Among the proposed detectors, none led to a uniformly better testing procedure. When based on the dependent multiplier bootstrap, the procedure based on $T_{m,q}$ in~\eqref{eq:Tmq} was observed to have the best behavior, overall, in terms of empirical level, power and mean detection delay. In the case of univariate independent observations, when the threshold function can be estimated using Monte Carlo simulation, the number of step $p$ of the threshold function can be chosen as large as the number of monitoring steps. However, in the time series case, when the estimation of the threshold function is based on the dependent multiplier bootstrap, $p$~should not be taken too large because of an error propagation effect.

%With the aim of monitoring multivariate independent observations, we were also able to establish the asymptotic validity of the proposed procedures when the threshold functions are estimated using a smooth bootstrap based on the \emph{empirical beta copula} \citep{KirSegTsu19}.

As already hinted at in Section~\ref{sec:thresh:generic}, a straightforward extension of the generic results on the estimation of the threshold function consists of allowing the conditional probability in~\eqref{eq:thresh:proc2} to change with the underlying monitoring sub-interval (or, equivalently, to have monitoring sub-intervals of different lengths). The choice of the $p$ conditional probabilities (or, equivalently, of the $p$ monitoring sub-intervals) could for instance be carried out according to the user's prior knowledge. In future work, we also plan to investigate the validity of additional bootstraps for monitoring multivariate time series. The current and future theoretical results would further need to be complemented by additional Monte Carlo simulations, involving in particular multivariate experiments. Such finite-sample investigations are however a real computational challenge given the complexity and cost of execution of such change-point detection procedures.

%%%%%%%%%%%%%%%%%%%%%%%%%%%%%%%%%%%%%%%%%%%%%%
%% Multiple Appendixes:                     %%
%%%%%%%%%%%%%%%%%%%%%%%%%%%%%%%%%%%%%%%%%%%%%%
\begin{appendix}
  
\section{Proofs of Propositions~\ref{prop:H0} and~\ref{prop:H1}}
\label{sec:proof:asym:det}

\begin{proof}[Proof of Proposition~\ref{prop:H0}]
  From Theorem~1 of \cite{Buc15}, we have that $\Bb_n \leadsto \Bb_C$ in $\ell^\infty([0,1]^{d+1})$, where $n = \ip{m(T+1)}$.  Let $\psi$ be the map from $\ell^\infty([0,1]^{d+1})$ to $\ell^\infty([0,T+1] \times [0,1]^d)$ defined, for any $f \in \ell^\infty([0,1]^{d+1})$, by
\begin{equation}
  \label{eq:psi}
  \psi(f)(s, \bm u) = \sqrt{T+1} f\{ s/(T+1), \bm u \}, \qquad s \in [0,T+1], \bm u \in [0,1]^d.
\end{equation}
It is straightforward to verify that $\psi$ is a continuous map which immediately implies by the continuous mapping theorem that $\psi(\Bb_n) \leadsto \psi(\Bb_C)$ in $\ell^\infty([0,T+1] \times [0,1]^d)$. Furthermore, it is easy to check that $\psi(\Bb_n) - \Bb_m = o_\Pr(1)$ and that $\psi(\Bb_C)$ is a tight centered Gaussian process such that, for any $s,t \in [0,T+1]$ and $\bm u, \bm v \in [0,1]^d$,
\begin{align*}
  \Cov\{\psi(\Bb_C)(s,\bm u), \psi(\Bb_C)(t,\bm v)\} &= (T+1) \Cov[\Bb_C \{ s/(T+1), \bm u \}, \Bb_C\{t/(T+1),\bm v\}] \\ &= (s \wedge t) \Gamma(\bm u, \bm v) = \Cov\{\Bb_C(s,\bm u), \Bb_C(t,\bm v)\},
\end{align*}
where $\Gamma$ is defined in~\eqref{eq:Gamma}. It follows that $\Bb_m \leadsto \Bb_C$ in $\ell^\infty([0,T+1] \times [0,1]^d)$.

It remains to show the subsequent claims. Under $H_0$,~\eqref{eq:GbmH0} holds and the continuous mapping theorem immediately implies that $\Gb_m \leadsto \Gb_C$ in $\ell^\infty(\Delta \times [0,1]^d)$ and, then, that $\Gb_{m,q} \leadsto \Gb_{C,q}$ in $\ell^\infty(\Delta \times [0,1]^d)$ since $\sup_{s \in [0,T+1]} |\lambda_m(0,s) - s| \leq 1/m$, the function $q$ in~\eqref{eq:q} is continuous on $\Delta$ and $\sup_{(s,t) \in \Delta} \{q(s,t)\}^{-1} < \infty$. The continuous mapping theorem further straightforwardly implies that $\Rb_{m,q} \leadsto \Rb_{C,q}$ in $\ell^\infty([1,T+1])$.

Let $\Jb_{m,q}(s,t) = \int_{[0,1]^d} \{  \Gb_{m,q}(s,t,\bm u) \}^2 \dd C_{1:\ip{mt}}(\bm u)$ for any $1 \leq s \leq t \leq T+1$. In order to prove that $\Sb_{m,q} \leadsto \Sb_{C,q}$ and  $\Tb_{m,q} \leadsto \Tb_{C,q}$ in $\ell^\infty([1,T+1])$, we shall first prove that $\Jb_{m,q} \leadsto \Jb_{C,q}$ in $\ell^\infty(\Delta \cap [1,T+1]^2)$, where $\Jb_{C,q}(s,t) = \int_{[0,1]^d} \{  \Gb_{C,q}(s,t,\bm u) \}^2 \dd C(\bm u)$ for any $1 \leq s \leq t \leq T+1$.

We start by showing that the finite-dimensional distributions of $\Jb_{m,q}$ converge weakly to those of $\Jb_{C,q}$.  Let $p \in \N$, $p > 1$, and $(s_1,t_1),\dots,(s_p,t_p) \in \Delta \cap [1,T+1]^2$ be arbitrary. The result is proven if we show that
\begin{equation}
  \label{eq:fidiJb}
\big( \Jb_{m,q}(s_1, t_1),\dots,\Jb_{m,q}(s_p, t_p) \big) \leadsto \big( \Jb_{C,q}(s_1, t_1),\dots,\Jb_{C,q}(s_p, t_p) \big).
\end{equation}
From the already proven weak convergence of $\Bb_m$ to $\Bb_C$ in $\ell^\infty([1,T+1] \times [0,1]^d)$, we obtain that
$$
\sup_{s \in [1,T+1] \atop \bm u \in [0,1]^d} |C_{1:\ip{ms}}(\bm u) - C(\bm u)| \leq \frac{1}{\sqrt{m}} \times  \sup_{s \in [1,T+1]} \frac{1}{\lambda_m(0,s)} \times \sup_{s \in [1,T+1] \atop \bm u \in [0,1]^d} |\Bb_m(s,\bm u)| \p 0,
$$
since $\sup_{s \in [1,T+1]} \{\lambda_m(0,s)\}^{-1} \leq 1$ and $\sup_{s \in [1,T+1], \bm u \in [0,1]^d} |\Bb_m(s,\bm u)| = O_\Pr(1)$. From the fact that $\Gb_{m,q} \leadsto \Gb_{C,q}$ in $\ell^\infty\{ (\Delta \cap [1,T+1]^2) \times [0,1]^d \}$, we then have that
\begin{multline*}
\big( \Gb_{m,q}(s_1, t_1, \cdot),\dots,\Gb_{m,q}(s_p, t_p, \cdot), C_{1:\ip{mt_1}}, \dots, C_{1:\ip{mt_p}} \big) \\ \leadsto \big( \Gb_{C,q}(s_1, t_1, \cdot),\dots,\Gb_{C,q}(s_p, t_p, \cdot), C, \dots, C \big)
\end{multline*}
in $\{ \ell^\infty([0,1]^d) \}^{2p}$. The latter can be combined with the fact that $C$ is continuous, $\Gb_{C,q}$ has continuous sample paths with probability one, Lemma~1 in \cite{KojSegYan11} and the continuous mapping theorem to obtain~\eqref{eq:fidiJb}.

It remains to show that the process $\Jb_{m,q}$ is asymptotically tight \citep[see, e.g.,][Section 1.5]{vanWel96}. From Section 2.1.2 and Problem 2.1.5 in the same reference, the latter is shown if, for every sequence $\delta_m \downarrow 0$,
\begin{equation}
  \label{eq:asym:equi:cont:Jb}
  \sup_{(s,t), (s',t') \in \Delta \cap [1,T+1]^2 \atop |s - s'|< \delta_m, |t - t'|< \delta_m}  | \Jb_{m,q}(s, t) - \Jb_{m,q}(s',t') | \p 0.
\end{equation}
The supremum on the left-hand side of the previous display is smaller than $I_m + J_m$, where
\begin{align*}
  I_m  =& \sup_{(s,t), (s',t') \in \Delta \cap [1,T+1]^2 \atop |s - s'|< \delta_m, |t - t'|< \delta_m}  \left| \int_{[0,1]^d} \{  \Gb_{m,q}(s,t,\bm u) \}^2 \dd C_{1:\ip{mt}}(\bm u) \right. \\ &- \left. \int_{[0,1]^d} \{  \Gb_{m,q}(s',t',\bm u) \}^2 \dd C_{1:\ip{mt}}(\bm u)  \right|, \\
  J_m =&  \sup_{(s,t), (s',t') \in \Delta \cap [1,T+1]^2 \atop |s - s'|< \delta_m, |t - t'|< \delta_m}  \left| \int_{[0,1]^d} \{  \Gb_{m,q}(s',t',\bm u) \}^2 \dd C_{1:\ip{mt}}(\bm u) \right. \\ &- \left. \int_{[0,1]^d} \{  \Gb_{m,q}(s',t',\bm u) \}^2 \dd C_{1:\ip{mt'}}(\bm u)  \right|.
\end{align*}
Now,
\begin{align*}
  I_m &\leq \sup_{(s,t), (s',t') \in \Delta \cap [1,T+1]^2 \atop |s - s'|< \delta_m, |t - t'|< \delta_m}  \int_{[0,1]^d} \left|  \{ \Gb_{m,q}(s,t,\bm u) \}^2  - \{  \Gb_{m,q}(s',t',\bm u) \}^2 \right| \dd C_{1:\ip{mt}}(\bm u) \\
      &\leq \sup_{(s,t), (s',t') \in \Delta \cap [1,T+1]^2 \atop |s - s'|< \delta_m, |t - t'|< \delta_m}  \sup_{\bm u \in [0,1]^d} \left|  \{ \Gb_{m,q}(s,t,\bm u) \}^2  - \{  \Gb_{m,q}(s',t',\bm u) \}^2 \right| \int_{[0,1]^d} \dd C_{1:\ip{mt}}(\bm u) \\
      &\leq \sup_{(s,t), (s',t') \in \Delta \cap [1,T+1]^2 \atop |s - s'|< \delta_m, |t - t'|< \delta_m}  \sup_{\bm u, \bm u' \in [0,1]^d \atop \| \bm u - \bm u'\| \leq \delta_m} \left| \{ \Gb_{m,q}(s,t,\bm u) \}^2  - \{  \Gb_{m,q}(s',t',\bm u') \}^2 \right| \p 0
\end{align*}
by the asymptotic uniform equicontinuity in probability of $\Gb_{m,q}^2$. Concerning $J_m$, we have
\begin{align*}
  J_m \leq& 2 \sup_{(s,t) \in \Delta \cap [1,T+1]^2 \atop t' \in [1,T+1], 0 \leq t' - t < \delta_m} \left| \frac{1}{\ip{mt}} \sum_{i=1}^{\ip{mt}} \{ \Gb_{m,q}(s,t,\bm U_i) \}^2 - \frac{1}{\ip{mt'}} \sum_{i=1}^{\ip{mt'}} \{ \Gb_{m,q}(s,t,\bm U_i) \}^2 \right| \\
      \leq& 2 \sup_{(s,t) \in \Delta \cap [1,T+1]^2 \atop t' \in [1,T+1], 0 \leq t' - t < \delta_m}  \sum_{i=1}^{\ip{mt}} \{ \Gb_{m,q}(s,t,\bm U_i) \}^2 \left( \frac{1}{\ip{mt}} - \frac{1}{\ip{mt'}} \right)  \\ &+ 2 \sup_{(s,t) \in \Delta \cap [1,T+1]^2 \atop t' \in [1,T+1], 0 \leq t' - t < \delta_m}  \frac{1}{\ip{mt'}} \sum_{i=\ip{mt} + 1}^{\ip{mt'}} \{ \Gb_{m,q}(s,t,\bm U_i) \}^2  \\
  \leq& 4 \sup_{t,t' \in [1,T+1] \atop 0 \leq t' - t < \delta_m} \frac{\ip{mt'} - \ip{mt}}{\ip{mt'}} \times \sup_{(s,t) \in \Delta \cap [1,T+1]^2 \atop \bm u \in [0,1]^d} \{ \Gb_{m,q}(s,t,\bm u) \}^2 \leq \delta_m \times O_\Pr(1) + o_\Pr(1),
\end{align*}
since $\ip{mt'} \geq m$ for all $t' \in [1,T+1]$,
$$
\sup_{t,t' \in [0,T+1]} \left| \frac{\ip{mt'} - \ip{mt}}{m} - (t' - t) \right| \leq 2 \sup_{t \in [0,T+1]} \left| \frac{\ip{mt}}{m} - t \right| \leq \frac{2}{m} \to 0,
$$
and the fact that $\sup_{(s,t) \in \Delta \cap [1,T+1]^2, \bm u \in [0,1]^d} \{ \Gb_{m,q}(s,t,\bm u) \}^2$ converges weakly by the continuous mapping theorem. Hence,~\eqref{eq:asym:equi:cont:Jb} holds and therefore $\Jb_{m,q} \leadsto \Jb_{C,q}$ in $\ell^\infty(\Delta \cap [1,T+1]^2)$. The fact that $\Sb_{m,q} \leadsto \Sb_{C,q}$ and  $\Tb_{m,q} \leadsto \Tb_{C,q}$ in $\ell^\infty([1,T+1])$ is finally and again an immediate consequence of the continuous mapping theorem.

Finally, we have to show that, for any $[t_1,t_2] \subset [1,T+1]$ such that $t_2 > 1$, the distributions of $\sup_{t \in [t_1,t_2]}\Rb_{C,q}(t)$, $\sup_{t \in [t_1,t_2]}\Sb_{C,q}(t)$ and $\sup_{t \in [t_1,t_2]}\Tb_{C,q}(t)$ are absolutely continuous with respect to the Lebesgue measure. To show the latter, we adapt the proof of Proposition~3.3 of \cite{BucFerKoj19} to the current setting. Let $\Cc(\Sc)$ denote the space of all continuous real-valued functions on $\Sc$ equipped with the uniform metric. Since the sample paths of $\Bb_C$ are elements of $\Cc([0,T+1] \times [0,1]^d)$ with probability one and $q$ in~\eqref{eq:q} is continuous, the sample paths of $\Gb_{C,q}$ are elements of $\Cc(\Delta \times [0,1]^d)$ with probability one. Fix $[t_1,t_2] \subset [1,T+1]$ such that $t_2 > 1$ and let $\vartheta_1$, $\vartheta_2$ and $\vartheta_3$ be the maps from $\Cc(\Delta \times [0,1]^d)$ to $\R$ defined, for any $f \in \Cc(\Delta \times [0,1]^d)$, by
\begin{align*}
  \vartheta_1(f) &= \sup_{t \in [t_1,t_2]} \sup_{s \in [1,t]} \sup_{\bm u \in [0,1]^d}  | f(s,t,\bm u) |, \\
  \vartheta_2(f) &= \sup_{t \in [t_1,t_2]} \sup_{s \in [1,t]} \left\{ \int_{[0,1]^d} \{  f(s,t,\bm u) \}^2 \dd C(\bm u) \right\}^{1/2}, \\
  \vartheta_3(f) &= \sup_{t \in [t_1,t_2]} \left\{ \int_1^t  \int_{[0,1]^d} \{  f(s,t,\bm u) \}^2 \dd C(\bm u) \dd s \right\}^{1/2}.
\end{align*}
Then, we have that $\sup_{t \in [t_1,t_2]}\Rb_{C,q}(t) = \vartheta_1(\Gb_{C,q})$, $\sup_{t \in [t_1,t_2]}\Sb_{C,q}(t) = \{ \vartheta_2(\Gb_{C,q}) \}^2$ and $\sup_{t \in [t_1,t_2]}\Tb_{C,q}(t) = \{ \vartheta_3(\Gb_{C,q}) \}^2$ and, to show the desired result, it suffices to prove that the distributions of $\vartheta_1(\Gb_{C,q})$, $\vartheta_2(\Gb_{C,q})$ and $\vartheta_2(\Gb_{C,q})$, denoted respectively by $\Lc\{ \vartheta_1(\Gb_{C,q}) \}$, $\Lc \{ \vartheta_2(\Gb_{C,q}) \}$ and $\Lc \{ \vartheta_2(\Gb_{C,q})\}$, are absolutely continuous with respect to the Lebesgue measure. Since the maps $\vartheta_1$, $\vartheta_2$ and $\vartheta_3$ are continuous and convex, from Theorem 7.1 in \cite{DavLif84}, we obtain that, for any $i \in \{1,2,3\}$, $\Lc\{ \vartheta_i(\Gb_{C,q}) \}$ is concentrated on $[a_i, \infty)$ and absolutely continuous on $(a_i,\infty)$, where
$$
a_i = \inf \{ \vartheta_i(f) : f \text{ belongs to the support of } \Lc(\Gb_{C,q}) \}.
$$
By Lemma~1.2~(e) in \cite{DerFehMatSch03}, we have that, for any $\eps > 0$, $\Pr \{ \vartheta_1(\Gb_{C,q}) \leq \eps \} > 0$. Hence, for any $\eps > 0$, 
$$
\Pr \{ \vartheta_3(\Gb_{C,q}) \leq \eps \} \geq \Pr \{ \vartheta_2(\Gb_{C,q}) \leq \eps \} \geq \Pr \{ \vartheta_1(\Gb_{C,q}) \leq \eps \} > 0.
$$
It follows that, for any $i \in \{1,2,3\}$ and any $\eps > 0$, there exists functions $f$ in the support of $\Lc(\Gb_{C,q})$ such that $\vartheta_i(f) \leq \eps$, which implies that $a_1 = a_2 = a_3 = 0$.

To conclude, it remains to show that $\Lc\{ \vartheta_1(\Gb_{C,q}) \}$, $\Lc \{ \vartheta_2(\Gb_{C,q}) \}$ and $\Lc \{ \vartheta_3(\Gb_{C,q})\}$ have no atom at 0. For any $f \in \Cc(\Delta \times [0,1]^d)$, we have that $\vartheta_2(f) = \vartheta_3(f) = 0$ if and only if $f(s,t,\bm u) = 0$ for all $(s,t) \in \Delta \cap ([1,t_2] \times [t_1,t_2] )$ and all $\bm u$ in the support of the distribution induced by $C$. Let $\bm u^* \in (0,1)^d$ be an arbitrary point in the latter support such that $\Var\{\Bb_C(1, \bm u^*)\} = \Gamma(\bm u^*,\bm u^*) > 0$, where $\Gamma$ is defined in~\eqref{eq:Gamma}, and let $(s^*,t^*) \in [1, t_2] \times [t_1,t_2]$ such that $s^* < t^*$. Then, $\Var\{ \Gb_C(s^*,t^*, \bm u^*) \} = s^* t^* (t^* - s^*) \Gamma(\bm u^*,\bm u^*) > 0$, which implies that $\Var\{ \Gb_{C,q}(s^*,t^*, \bm u^*) \} = s^* t^* (t^* - s^*) \Gamma(\bm u^*,\bm u^*) \{ q(s^*,t^*) \}^{-2} > 0$. The proof is complete since 
$$
\Pr \{ \vartheta_1(\Gb_{C,q})  = 0 \} = \Pr \{ \vartheta_2(\Gb_{C,q})  = 0 \}  = \Pr \{ \vartheta_3(\Gb_{C,q})  = 0 \} \leq \Pr \{ \Gb_{C,q}(s^*,t^*, \bm u^*) = 0 \} = 0.
$$
\end{proof}

\begin{proof}[Proof of Proposition~\ref{prop:H1}]
Let $\Kb_m(s,t, \bm x) = m^{-1/2} \Hb_m(s,t, \bm x) - K_c(s,t,\bm x)$, $(s,t) \in \Delta$, $\bm x \in \R^d$. The first claim is proven if
\begin{equation}
  \label{eq:claim1}
\sup_{(s,t) \in \Delta \atop \bm x \in \R^d} | \Kb_m(s,t,\bm x) | \p 0.
\end{equation}
The supremum on the left-hand side of the previous display is equal to
\begin{equation}
\label{eq:max:sup}
\max \left\{\sup_{0 \leq s \leq t \leq c \atop \bm x \in \R^d} | \Kb_m(s,t,\bm x) |,  \sup_{0 \leq s \leq c\leq t \leq T+1 \atop \bm x \in \R^d} | \Kb_m(s,t,\bm x) |, \sup_{c \leq s \leq t \leq T+1 \atop \bm x \in \R^d} | \Kb_m(s,t,\bm x) | \right\}.
\end{equation}
To prove~\eqref{eq:claim1}, we shall show that each of the three suprema in the previous display converge in probability to zero. Notice first that
\begin{equation}
\label{eq:Kc:parts}
K_c(s,t,\bm x) = \left\{
  \begin{array}{ll}
    0, &\text{if } 0 \leq s \leq t \leq c, \\
    s(t-c) \{F(\bm x) - G(\bm x)\}, &\text{if } 0 \leq s \leq c\leq t \leq T+1, \\
    c(t-s)  \{F(\bm x) - G(\bm x)\}, &\text{if } c \leq s \leq t \leq T+1.
  \end{array}
\right.
\end{equation}
Furthermore, for any $(s,t) \in \Delta$, $\bm x \in \R^d$ and $H \in \{F,G\}$, let
$$
\Fb_{m,H}(s,t,\bm x) = \sqrt{m} \lambda_m(s,t) \{F_{\ip{ms}+1:\ip{mt}}(\bm x) - H(\bm x) \}  =  \Fb_{m,H}^\circ(t,\bm x) - \Fb_{m,H}^\circ(s,\bm x),
$$
where $\Fb_{m,H}^\circ(s,\bm x) = \sqrt{m} \lambda_m(0,s) \{F_{1:\ip{ms}}(\bm x) - H(\bm x) \}$, $\lambda_m$ is defined in~\eqref{eq:lambda} and $F_{1:\ip{ms}}$ is generically defined by~\eqref{eq:Fjk}. By proceeding as in Section~\ref{sec:margin:free:H0}, it can be verified that, under $H_0$ in~\eqref{eq:H0}, $\Fb_{m,F}^\circ(s,\bm x) = \Bb_m \{ s,\bm F(\bm x) \}$ for all $s \in [0,T+1]$ and $\bm x \in \R^d$, where $\Bb_m$ is defined in~\eqref{eq:Bbm} and $\bm F(\bm x) = (F_1(x_1),\dots,F_d(x_d))$. By the continuous mapping theorem, it thus immediately follows that, under $H_0$, $\Fb_{m,F}$ converges weakly in $\ell^\infty(\Delta \times \R^d)$ to a tight limit. Some thought then reveals that, under the conditions of the proposition, $\Fb_{m,F}$ (resp.\ $\Fb_{m,G}$) converges weakly in $\ell^\infty \{ (\Delta \cap [0,c]^2) \times \R^d \}$ (resp.\ in $\ell^\infty\{ (\Delta \cap [c,T+1]^2) \times \R^d \}$) to a tight limit.

From the expression of $K_c$ given in~\eqref{eq:Kc:parts}, for the first supremum in~\eqref{eq:max:sup}, we obtain that
$$
\sup_{0 \leq s \leq t \leq c \atop \bm x \in \R^d} | \Kb_m(s,t,\bm x) | = m^{-1/2} \sup_{0 \leq s \leq t \leq c \atop \bm x \in \R^d} | \Hb_m(s,t,\bm x) | = o(1) \times O_\Pr(1) \p 0,
$$
since $\Hb_m$ converges weakly in $\ell^\infty \{ (\Delta \cap [0,c]^2) \times \R^d \}$ to a tight limit as a consequence of the fact that, for any $0 \leq s \leq t \leq c$ and $\bm x \in \R^d$, $\Hb_m(s,t, \bm x) = \lambda_m(s,t) \Fb_{m,F}(0,s) -\lambda_m(0,s) \Fb_{m,F}(s,t)$ and the continuous mapping theorem.

Regarding the second supremum, for any $0 \leq s \leq c\leq t \leq T+1$ and $\bm x \in \R^d$, we have that
$$
\lambda_m(s,t) F_{\ip{ms}+1:\ip{mt}}(\bm x) = \lambda_m(s,c) F_{\ip{ms}+1:\ip{mc}}(\bm x) + \lambda_m(c,t) F_{\ip{mc}+1:\ip{mt}}(\bm x).
$$
Thus, on one hand,
\begin{multline*}
m^{-1/2} \Hb_m(s,t, \bm x) \\ = \lambda_m(0,s) \{ \lambda_m(s,t) F_{1:\ip{ms}}(\bm x) - \lambda_m(s,c) F_{\ip{ms}+1:\ip{mc}}(\bm x) - \lambda_m(c,t) F_{\ip{mc}+1:\ip{mt}}(\bm x) \}.
\end{multline*}
On the other hand, from~\eqref{eq:Kc:parts} and using the fact that $\sup_{(s,t) \in \Delta} |\lambda_m(s,t) - (t-s)| \leq 2/m$,
$$
K_c(s,t,\bm x) =  \lambda_m(0,s) \{ \lambda_m(s,t)  F(\bm x) - \lambda_m(s,c) F(\bm x) - \lambda_m(c,t) G(\bm x) \} + O(1/m).
$$
By the triangle inequality and using the fact that $\sup_{0 \leq s \leq t \leq T+1} |\lambda_m(s,t)| \leq T+1$, it then follows that
\begin{multline*}
  \sup_{0 \leq s \leq c\leq t \leq T+1 \atop \bm x \in \R^d} | \Kb_m(s,t,\bm x) | \leq m^{-1/2} (T+1) \left[ \sup_{0 \leq s \leq c \atop \bm x \in \R^d} |\Fb_{m,F}(0,s,\bm x)|  \right. \\ \left. + \sup_{0 \leq s \leq c \atop \bm x \in \R^d} |\Fb_{m,F}(s,c,\bm x)| + \sup_{c\leq t \leq T+1 \atop \bm x \in \R^d} | \Fb_{m,G}(c,t,\bm x) |  \right] = o(1) \times O_\Pr(1).
\end{multline*}
Similarly, for the third supremum, for any $c \leq s \leq T+1$ and $\bm x \in \R^d$,
$$
\lambda_m(0,s) F_{1:\ip{ms}}(\bm x) = \lambda_m(0,c) F_{1:\ip{mc}}(\bm x) + \lambda_m(c,s) F_{\ip{mc}+1:\ip{ms}}(\bm x),
$$
and, hence, on one hand, for any $c \leq s \leq t \leq T+1$ and $\bm x \in \R^d$,
\begin{multline*}
m^{-1/2} \Hb_m(s,t, \bm x) \\ = \lambda_m(s,t) \{ \lambda_m(0,c) F_{1:\ip{mc}}(\bm x) + \lambda_m(c,s) F_{\ip{mc}+1:\ip{ms}}(\bm x) - \lambda_m(0,s) F_{\ip{ms}+1:\ip{mt}}(\bm x)\},
\end{multline*}
while, on the other hand,
$$
K_c(s,t,\bm x) =  \lambda_m(s,t) \{ \lambda_m(0,c)  F(\bm x) +  \lambda_m(c,s) G(\bm x) - \lambda_m(0,s) G(\bm x) \} + O(1/m).
$$
Finally, by the triangle inequality,
\begin{multline*}
  \sup_{c \leq s \leq t \leq T+1 \atop \bm x \in \R^d} | \Kb_m(s,t,\bm x) |  \leq m^{-1/2} (T+1) \left[ \sup_{\bm x \in \R^d} | \Fb_{m,F}(0,c,\bm x) | \right. \\ \left. + \sup_{c\leq s \leq T+1 \atop \bm x \in \R^d} | \Fb_{m,G}(c,s,\bm x) | + \sup_{c\leq s \leq t \leq T+1 \atop \bm x \in \R^d}  | \Fb_{m,G}(s,t,\bm x) |  \right] = o(1) \times O_\Pr(1),
\end{multline*}
which completes the proof of~\eqref{eq:claim1}.

Using the fact that $\sup_{s \in [0,T+1]} |\lambda_m(0,s) - s| \leq 1/m$, the function $q$ in~\eqref{eq:q} is continuous on $\Delta$ and $\sup_{(s,t) \in \Delta} \{q(s,t)\}^{-1} < \infty$, we obtain, from the continuous mapping theorem, that $\Hb_{m,q} \p K_{c,q}$ in $\ell^\infty(\Delta \times \R^d)$. From~\eqref{eq:Rmq} and~\eqref{eq:Hbm}, and proceeding as in~\eqref{eq:Rmq:Gmq}, it is easy to verify that $\Rb_{m,q}(t) = \sup_{s \in [1,t]} \sup_{\bm x \in \R^d} |\Hb_m(s,t,\bm x)|$, $t \in [1,T+1]$. Hence, again by the continuous mapping theorem, $m^{-1/2} \Rb_{m,q} \p L_{c,q}$ in $\ell^\infty([1,T+1])$, where $L_{c,q}(t) = \sup_{s \in [1,t]} \sup_{\bm x \in \R^d} |K_{c,q}(s,t,\bm x)|$, and, thus, $m^{-1/2} \sup_{t \in [1,T+1]} \Rb_{m,q}(t) \p  \sup_{t \in [1,T+1]} L_{c,q}(t)$. Since
$$
\sup_{t \in [1,T+1]} L_{c,q}(t) = \sup_{1 \leq s \leq t \leq T+1} \frac{(s \wedge c) \{(t \vee c) - (s \vee c) \}}{q(s,t)} \sup_{\bm x \in \R^d} |F(\bm x) - G(\bm x)| > 0,
$$
we immediately obtain that $\sup_{t \in [1,T+1]} \Rb_{m,q}(t) \p \infty$.

To show the two last remaining claims, we shall first prove that
\begin{equation}
  \label{eq:unifcp}
\sup_{1 \leq s \leq t \leq T+1} \left| m^{-1} \int_{\R^d} \{ \Hb_{m,q}(s,t,\bm x) \}^2 \dd F_{1:\ip{mt}}(\bm x) -  \int_{\R^d} \{ K_{c,q}(s,t,\bm x) \}^2 \dd H_t(\bm x) \right| \p 0,
\end{equation}
where, for any $t \in [1,T+1]$ and $\bm x \in \R^d$,
\begin{equation}
  \label{eq:Ht}
  H_t(\bm x) = \frac{c \wedge t}{t} F(\bm x) +  \frac{t - c \wedge t}{t} G(\bm x).
\end{equation}
By the triangle inequality, the supremum on the left hand-side of~\eqref{eq:unifcp} is smaller than $I_m + J_m$, where
$$
I_m = \sup_{1 \leq s \leq t \leq T+1} \left| m^{-1} \int_{\R^d} \{ \Hb_{m,q}(s,t,\bm x) \}^2 \dd F_{1:\ip{mt}}(\bm x) -  \int_{\R^d} \{ K_{c,q}(s,t,\bm x) \}^2 \dd F_{1:\ip{mt}}(\bm x) \right|
$$
and
$$
J_m = \sup_{1 \leq s \leq t \leq T+1} \left| \int_{\R^d} \{ K_{c,q}(s,t,\bm x) \}^2 \dd F_{1:\ip{mt}}(\bm x) -  \int_{\R^d} \{ K_{c,q}(s,t,\bm x) \}^2 \dd H_t(\bm x) \right|.
$$
On one hand, some thought reveals that
$$
I_m \leq \sup_{1 \leq s \leq t \leq T+1 \atop \bm x \in \R^d} \left| \{ m ^{-1/2} \Hb_{m,q}(s,t,\bm x) \}^2  - \{ K_{c,q}(s,t,\bm x) \}^2 \right| \p 0
$$
as a consequence of the continuous mapping theorem. On the other hand, from~\eqref{eq:Kc}, we have that
\begin{multline*}
J_m \leq \sup_{1 \leq s \leq t \leq T+1}   \frac{[(s \wedge c) \{(t \vee c) - (s \vee c) \}]^2}{\{q(s,t)\}^2}  \\ \times  \sup_{1 \leq t \leq T+1} \left| \int_{\R^d} \{F(\bm x) - G(\bm x)\}^2 \dd F_{1:\ip{mt}}(\bm x) -  \int_{\R^d} \{F(\bm x) - G(\bm x)\}^2 \dd H_t(\bm x) \right|.
\end{multline*}
Using~\eqref{eq:Ht}, the second supremum on the right-hand side of the previous display is smaller than $\max\{K_m,L_m\}$, where
\begin{align*}
  K_m =&  \sup_{1 \leq t \leq c} \left| \int_{\R^d} \{F(\bm x) - G(\bm x)\}^2 \dd F_{1:\ip{mt}}(\bm x) -  \int_{\R^d} \{F(\bm x) - G(\bm x)\}^2 \dd F(\bm x) \right|, \\
  L_m =&  \sup_{c \leq t \leq T+1} \left| \int_{\R^d} \{F(\bm x) - G(\bm x)\}^2 \dd F_{1:\ip{mt}}(\bm x) -  \int_{\R^d} \{F(\bm x) - G(\bm x)\}^2 \dd H_t(\bm x) \right|.
\end{align*}
From the assumptions on the strong mixing coefficients and Theorem~1.2 in \cite{Ber87} \citep[see also][Chapter~3]{Rio17}, the strong law of large numbers implies that, as $k \to \infty$,
$$
M_k = \left| \frac{1}{k} \sum_{i=1}^{k} \{F(\bm Y_i) - G(\bm Y_i)\}^2  -  \int_{\R^d} \{F(\bm x) - G(\bm x)\}^2 \dd F(\bm x) \right| \as 0,
$$
where the arrow~`$\as$' denotes almost sure convergence. Since the previous convergence is equivalent to the fact that $\limsup_{m \to \infty} M_m = 0$ almost surely, we obtain that $K_m = \sup_{m \leq k \leq \ip{mc}} M_k \leq \sup_{m \leq k} M_k \as 0$. Using the fact that, for any $c \leq t \leq T+1$ and $\bm x \in \R^d$,
$$
F_{1:\ip{mt}}(\bm x) = \frac{\ip{mc}}{\ip{mt}} F_{1:\ip{mc}}(\bm x) + \frac{\ip{mt} - \ip{mc}}{\ip{mt}} F_{\ip{mc}+1:\ip{mt}}(\bm x)
$$
and that $\sup_{c \leq t \leq T+1} | \ip{mc} / \ip{mt} - c/t | = O(1/m)$, we obtain that
\begin{multline*}
  L_m \leq \left| \int_{\R^d} \{F(\bm x) - G(\bm x)\}^2 \dd F_{1:\ip{mc}}(\bm x) -  \int_{\R^d} \{F(\bm x) - G(\bm x)\}^2 \dd F(\bm x) \right| \\
         + \sup_{c \leq t \leq T+1} \frac{\ip{mt} - \ip{mc}}{\ip{mt}} \left| \int_{\R^d} \{F(\bm x) - G(\bm x)\}^2 \dd F_{\ip{mc}+1:\ip{mt}}(\bm x) \right. \\ \left. -  \int_{\R^d} \{F(\bm x) - G(\bm x)\}^2 \dd G(\bm x) \right| + O(1/m).
\end{multline*}
The first term on the right-hand side of the previous display is equal to $M_{\ip{mc}}$ and thus converges to zero almost surely. The second term can be written as
\begin{equation}
\label{eq:term:sum}
  \sup_{\ip{mc}+1 \leq k \leq \ip{m(T+1)}}  \left| \frac{1}{k} \sum_{i=\ip{mc} + 1}^{k} \{F(\bm Z_i) - G(\bm Z_i)\}^2  -  \frac{k - \ip{mc}}{k} \int_{\R^d} \{F(\bm x) - G(\bm x)\}^2 \dd G(\bm x) \right|.
\end{equation}
Letting
$$
N_k = \left| \frac{1}{k} \sum_{i=1}^{k} \{F(\bm Z_i) - G(\bm Z_i)\}^2  -  \int_{\R^d} \{F(\bm x) - G(\bm x)\}^2 \dd G(\bm x) \right|
$$
and decomposing the sum in~\eqref{eq:term:sum}, by the triangle inequality, \eqref{eq:term:sum} is smaller than
$$
\sup_{\ip{mc}+1 \leq k \leq \ip{m(T+1)}} N_k + N_{\ip{mc}} \leq \sup_{\ip{mc}+1 \leq k} N_k + N_{\ip{mc}} \as 0,
$$
since $N_k \as 0$ as $k \to \infty$.

Hence, $J_m \as 0$. It follows that~\eqref{eq:unifcp} is proven and, from the continuous mapping theorem, we immediately obtain that $m^{-1} \Sb_{m,q} \p M_{c,q}$ and $m^{-1} \Tb_{m,q} \p N_{c,q}$ in $\ell^\infty([1,T+1])$, where
$$
M_{c,q}(t) = \sup_{s \in [1,t]} \int_{\R^d} \{ K_{c,q}(s,t,\bm x) \}^2 \dd H_t(\bm x) \, \text{and} \, N_{c,q}(t) = \int_1^t \int_{\R^d} \{ K_{c,q}(s,t,\bm x) \}^2 \dd H_t(\bm x) \dd s,
$$
and then that
\begin{align*}
  m^{-1} &\sup_{t \in [1,T+1]} \Sb_{m,q}(t) \p  \sup_{t \in [1,T+1]} M_{c,q}(t), \\
  m^{-1} &\sup_{t \in [1,T+1]} \Tb_{m,q}(t) \p  \sup_{t \in [1,T+1]} N_{c,q}(t).
\end{align*}
Let $\Lambda_F = \int_{\R^d}  \{ F(\bm x) - G(\bm x) \}^2 \dd F(\bm x)$ and $\Lambda_G= \int_{\R^d}  \{ F(\bm x) - G(\bm x) \}^2 \dd G(\bm x)$. Since $F$ and $G$ are continuous and $F \neq G$, $\Lambda_F > 0$ and $\Lambda_G > 0$. As a consequence, for all $t \in [1,T+1]$,
$$
\Lambda_t = \frac{c \wedge t}{t} \Lambda_F + \frac{t - c \wedge t}{t} \Lambda_G > 0.
$$
Furthermore, since $q(s,t) > 0$ for all $(s,t) \in \Delta$ and since, for all $c < t \leq T+1$,
$$
\sup_{1 \leq s \leq t} (s \wedge c) \{(t \vee c) - (s \vee c) \}  =  c (t - c ) > 0,
$$
we have, from the continuity of $q$, $t \mapsto \Lambda_t$ and $(s,t) \mapsto (s \wedge c) \{(t \vee c) - (s \vee c) \}$, that
$$
\sup_{t \in [1,T+1]} M_{c,q}(t) \geq \sup_{t \in [c,T+1]} \Lambda_t \sup_{s \in [1,t]} \left[ \frac{(s \wedge c) \{t - (s \vee c) \}}{q(s,t)} \right]^2  > 0,
$$
and
$$
\sup_{t \in [1,T+1]} N_{c,q}(t) \geq \sup_{t \in [c,T+1]} \Lambda_t \int_1^t \left[ \frac{(s \wedge c) \{t - (s \vee c) \}}{q(s,t)} \right]^2 \dd s > 0.
$$
\end{proof}

%%%%%%%%%%%%%%%%%%%%%%%%%%%%%%%%%%%%%%%%%%%%%%%%%%%%%%%%%%%%%%%%%%%%%%%%%%

\section{Auxiliary lemmas for the proof of Theorem~\ref{thm:thresh:boot}}
\label{sec:aux:thresh:generic}

This section, which is, to a large extent, notationally independent of the rest of the paper, provides the proofs of two lemmas, possibly of independent interest, necessary for showing Theorem~\ref{thm:thresh:boot}.

Let $\bm \Xc_n$ denote available data. No assumptions are made on $\bm \Xc_n$ apart from measurability. To fix ideas, one can think of $\bm \Xc_n$ as a sequence of $n$ multivariate serially dependent random vectors. Let $\bm S_n = \bm S_n(\bm \Xc_n)$ be a $\R^p$-valued statistic such that $\bm S_n = (S_{1,n},\dots,S_{p,n}) \leadsto \bm S = (S_1,\dots,S_p)$ as $n \to \infty$, where the random vector $\bm S$ is assumed to have a continuous d.f. We additionally suppose that we have available \emph{bootstrap replicates} $\bm S_{n}^{[i]} = \bm S_{n}^{[i]}(\bm \Xc_n, \bm \Wc_n^{[i]})$ of $\bm S_n$, where the $\bm \Wc_n^{[i]}$, $i \in \N$, are $n$-dimensional independent and identically distributed random vectors representing the additional sources of randomness involved in the underlying bootstrap mechanism. We shall further assume that, as $n \to \infty$,
\begin{equation}
\label{eq:jwconv}
\big(\bm S_{n},\bm S_{n}^{[1]}, \bm S_{n}^{[2]} \big) \\ \leadsto \big ( \bm S, \bm S^{[1]}, \bm S^{[2]} \big),
\end{equation}
in $(\R^p)^3$, where $\bm S^{[1]}$ and $\bm S^{[2]}$ are independent copies of $\bm S$. Note that, from Lemma~2.2 of \cite{BucKoj19}, \eqref{eq:jwconv} is equivalent to the usual conditional bootstrap consistency statement, that is,
$$
\sup_{\bm x \in \R^p} | \Pr(\bm S_n^{[1]} \leq \bm x \mid \bm \Xc_n) - \Pr(\bm S_n \leq \bm x)| \p 0, \qquad \text{as } n \to \infty.
$$

Before stating and proving the two lemmas, we introduce some additional notation and list useful results. For any $q\in\{1,\dots,p\}$, $\{j_1,\dots,j_q\}\subset\{1,\dots,p\}$, $x_{j_1},\dots, x_{j_q} \in \R$ and $B \in \N$, let
\begin{align*}
   F_{\bm S, \{j_1,\dots,j_q\}}(x_{j_1},\dots,x_{j_q}) &= \Pr(S_{j_1} \leq x_{j_1}, \dots, S_{j_q} \leq x_{j_q}), \\
  F_{\bm S_n, \{j_1,\dots,j_q\}}(x_{j_1},\dots, x_{j_q}) &= \Pr(S_{j_1,n} \leq x_{j_1}, \dots, S_{j_q,n} \leq x_{j_q}), \\
F_{\bm S_n, \{j_1,\dots,j_q\}}^B(x_{j_1},\dots,x_{j_q}) &= \frac{1}{B} \sum_{i=1}^B \1(S_{j_1,n}^{[i]} \leq x_{j_1}, \dots, S_{j_q,n}^{[i]} \leq x_{j_q}).
\end{align*}
Since $\bm S_{n} \leadsto \bm S$ in $\R^p$ as $n \to \infty$ and $\bm S$ has a continuous d.f., we have from Lemma~2.11 of \cite{van98} that, for any $q\in\{1,\dots,p\}$ and $\{j_1,\dots,j_q\}\subset\{1,\dots,p\}$, as $n \to \infty$,
\begin{equation}
  \label{eq:unif}
  \sup_{(x_{j_1},\dots,x_{j_q}) \in \R^q} |F_{\bm S_n, \{j_1,\dots,j_q\}}(x_{j_1},\dots,x_{j_q}) - F_{\bm S, \{j_1,\dots,j_q\}}(x_{j_1},\dots,x_{j_q})| \to 0.
\end{equation}
Proceeding as in the proof of the aforementioned lemma, it can actually also be shown that, for any $q\in\{1,\dots,p\}$ and $\{j_1,\dots,j_q\}\subset\{1,\dots,p\}$, as $n \to \infty$,
\begin{equation}
  \label{eq:unif:sign}
  \sup_{(x_{j_1},\dots,x_{j_q}) \in \R^q} |\Pr(S_{j_1,n} < x_{j_1}, \dots, S_{j_q,n} < x_{j_q}) - \Pr(S_{j_1} < x_{j_1}, \dots, S_{j_q} < x_{j_q})| \to 0.
\end{equation}
Combining~\eqref{eq:jwconv} with Assertion~(f) of Lemma 2.2 in \cite{BucKoj19} and~\eqref{eq:unif}, we further obtain that, for any $q\in\{1,\dots,p\}$ and $\{j_1,\dots,j_q\}\subset\{1,\dots,p\}$, as $n,B \to \infty$,
\begin{equation}
  \label{eq:unif:B}
  \sup_{(x_{j_1},\dots,x_{j_q}) \in \R^q} |F_{\bm S_n, \{j_1,\dots,j_q\}}^B(x_{j_1},\dots,x_{j_q}) - F_{\bm S, \{j_1,\dots,j_q\}}(x_{j_1},\dots,x_{j_q})| \p 0.
\end{equation}
Let $\xi \in (0,1)$ be arbitrary. The following notation will also be used in the lemmas. Let
$$
F_{\bm S, 1}(x) = F_{\bm S, \{1\}}(x) = \Pr(S_1 \leq x), \, x \in \R, \qquad g_1 =  F_{\bm S, 1}^{-1}(1 - \xi),
$$
and, recursively, for $j$ successively equal to $2,\dots,p$,
\begin{equation}
  \label{eq:Fj:gj}
F_{\bm S, j}(x) =  \frac{F_{\bm S, \{1,\dots,j\}}(g_1,\dots,g_{j-1},x)}{F_{\bm S, \{1,\dots,j-1\}}(g_1,\dots,g_{j-1})}, \, x \in\R, \qquad g_j =  F_{\bm S, j}^{-1}(1 - \xi).
\end{equation}
Similarly, let
$$
F_{\bm S_n, 1}^B(x) = F_{\bm S_n, \{1\}}^B(x), \, x \in \R, \qquad g_{1,n}^B =  F_{\bm S_n, 1}^{B,-1}(1 - \xi),
$$
and, recursively, for $j$ successively equal to $2,\dots,p$,
\begin{equation}
  \label{eq:FjB:gjB}
F_{\bm S_n, j}^B(x) = \frac{F_{\bm S_n, \{1,\dots,j\}}^B(g_{1,n}^B,\dots,g_{j-1,n}^B,x)}{F_{\bm S_n, \{1,\dots,j-1\}}^B(g_{1,n}^B,\dots,g_{j-1,n}^B)}, \, x \in\R, \qquad g_{j,n}^B =  F_{\bm S_n, j}^{B,-1}(1 - \xi).
\end{equation}

The following two lemmas are instrumental for proving Theorem~\ref{thm:thresh:boot}.

\begin{lem}
\label{lem:threshold}
As $n,B \to \infty$,
\begin{gather}
  \label{eq:rec11}
  \sup_{x \in \R} |F_{\bm S_n,1}^B(x) - F_{\bm S,1}(x)| \p 0, \\
  \label{eq:rec12}
  \Pr ( S_{1,n} \leq g_{1,n}^B ) \to 1 - \xi, \\
  \label{eq:rec13}
  F_{\bm S,1}( g_{1,n}^B ) \to 1 - \xi,
\end{gather}
and, for any $j \in \{2,\dots,p\}$,
\begin{gather}
  \label{eq:recj1}
  \sup_{x \in \R} |F_{\bm S_n,j}^B(x) - F_{\bm S,j}(x)| \p 0, \\
  \label{eq:recj2}
  \Pr ( S_{j,n} \leq g_{j,n}^B \mid S_{1,n}  \leq g_{1,n}^B, \dots,S_{j-1,n} \leq g_{j-1,n}^B ) \to 1 - \xi, \\
  \label{eq:recj3}
  F_{\bm S,j}( g_{j,n}^B ) \to 1 - \xi.
\end{gather}
The statements~\eqref{eq:rec12} and~\eqref{eq:recj2} with `$\leq$' replaced by `$<$' hold as well.

\end{lem}

\begin{lem}
  \label{lem:thresh:inf}
Let $\bm T_n = \bm T_n(\bm \Xc_n) = (T_{1,n}, \dots, T_{p,n})$ be a $\R^p$-valued statistic such that $\max_{1 \leq j \leq p} T_{j,n} \p \infty$ as $n \to \infty$. Then
$$
\Pr(T_{1,n} \leq g_{1,n}^B,\dots,T_{p,n} \leq g_{p,n}^B) \to 0 \qquad \text{ as } n,B \to \infty.
$$
\end{lem}

\begin{proof}[Proof of Lemma~\ref{lem:threshold}]
The claims in~\eqref{eq:rec11} and~\eqref{eq:rec12} are immediate consequences of~\eqref{eq:unif:B} and Lemma~4.2 in \cite{BucKoj19}, respectively. The claim~\eqref{eq:rec13} follows from the fact that $F_{\bm S,1}(g_1) = 1-\xi$, the triangle inequality,~\eqref{eq:unif} and~\eqref{eq:rec12}.

  To prove~\eqref{eq:recj1},~\eqref{eq:recj2} and~\eqref{eq:recj3}, we proceed by induction on $j$.

\emph{Proof of~\eqref{eq:recj1},~\eqref{eq:recj2} and~\eqref{eq:recj3} for j=2:}
By the triangle inequality, $|F_{\bm S_n,1}^B(g_{1,n}^B) - F_{\bm S,1}(g_1)|$ is smaller than
$$
|F_{\bm S_n,1}^B(g_{1,n}^B) - F_{\bm S,1}(g_{1,n}^B)| + |F_{\bm S,1}(g_{1,n}^B) - F_{\bm S_n,1}(g_{1,n}^B)|+ |F_{\bm S_n,1}(g_{1,n}^B) - F_{\bm S,1}(g_1)|.
$$
The first term converges in probability to zero by~\eqref{eq:unif:B} as $n,B \to \infty$.  The second term converges to zero by~\eqref{eq:unif} as $n,B \to \infty$. The third term converges to zero as a consequence of~\eqref{eq:rec12} since $F_{\bm S,1}(g_1) = 1 - \xi$. Hence, $|F_{\bm S_n,1}^B(g_{1,n}^B) - F_{\bm S,1}(g_1)| \p 0$ as $n,B \to \infty$. From~\eqref{eq:Fj:gj} and~\eqref{eq:FjB:gjB}, the latter implies that to prove that $\sup_{x \in \R} |F_{\bm S_n,2}^B(x) - F_{\bm S,2}(x)| \p 0$ as $n,B \to \infty$, it suffices to prove that
$$
\sup_{x \in \R} |F_{\bm S_n,\{1,2\}}^B(g_{1,n}^B,x) - F_{\bm S,\{1,2\}}(g_1,x)| \p 0
$$
as $n,B \to \infty$. From the triangle inequality,~\eqref{eq:unif} and~\eqref{eq:unif:B}, the latter will hold if
\begin{equation}
  \label{eq:aim1}
\sup_{x \in \R} |F_{\bm S_n,\{1,2\}}(g_{1,n}^B,x) - F_{\bm S_n,\{1,2\}}(g_1,x)| \to 0
\end{equation}
as $n,B \to \infty$. For any $x \in \R$, we have
\begin{multline}
  \label{eq:already2}
  F_{\bm S_n,\{1,2\}}(g_{1,n}^B,x) - F_{\bm S_n,\{1,2\}}(g_1,x) %=  1 - \Pr ( S_{1,n} > g_{1,n}^B, S_{2,n} \leq x ) - 1 + \Pr(S_{1,n} > g_1, S_{2,n} \leq x ) \\
  = \Pr(S_{1,n} \geq g_1, S_{2,n} \leq x ) -  \Pr(S_{1,n} = g_1, S_{2,n} \leq x ) \\ - \Pr ( S_{1,n} \geq g_{1,n}^B, S_{2,n} \leq x )  + \Pr ( S_{1,n} = g_{1,n}^B, S_{2,n} \leq x ).
\end{multline}
Hence, by the triangle inequality,
\begin{multline}
  \label{eq:already1}
\sup_{x \in \R} |F_{\bm S_n,\{1,2\}}(g_{1,n}^B,x) - F_{\bm S_n,\{1,2\}}(g_1,x)| \\ \leq \sup_{x \in \R} |\Pr ( S_{1,n} \geq g_{1,n}^B, S_{2,n} \leq x ) - \Pr(S_{1,n} \geq g_1, S_{2,n} \leq x )| + 2 \sup_{x \in \R} \Pr(S_{1,n} = x).
\end{multline}
From the triangle inequality,~\eqref{eq:unif} and~\eqref{eq:unif:sign}, we obtain that the last supremum on the right-hand side of the previous display converges to zero a $n \to \infty$ since $\bm S$ has a continuous d.f. Hence, to show~\eqref{eq:aim1}, it remains to show that the first supremum on the right converges to zero in probability as $n,B \to \infty$. From the right-continuity of $F_{\bm S, 1}$ and $F_{\bm S_n,1}^B$, we obtain that, for any $x \in \R$,
\begin{align*}
  |\Pr ( S_{1,n} &\geq g_{1,n}^B, S_{2,n} \leq x ) - \Pr(S_{1,n} \geq g_1, S_{2,n} \leq x )| \\
                 &= | \Pr \{ S_{1,n} \geq F_{\bm S_n,1}^{B,-1}(1 -\xi), S_{2,n} \leq x \} - \Pr \{ S_{1,n} \geq F_{\bm S, 1}^{-1}(1 - \xi), S_{2,n} \leq x \} | \\
                 &= | \Pr \{ F_{\bm S_n,1}^B(S_{1,n}) \geq 1 -\xi, S_{2,n} \leq x \} - \Pr \{ F_{\bm S, 1}(S_{1,n}) \geq 1 - \xi, S_{2,n} \leq x \} |.
\end{align*}
Furthermore, using the fact that, for any $a, b, y \in \R$ and $\eps > 0$,
\begin{equation}
  \label{eq:ineq:ind}
|\1(y \leq a) - \1(y \leq b)| \leq \1(|y - a| \leq \eps) + \1(|a-b| > \eps),
\end{equation}
we obtain that, for any $\eps > 0$,
\begin{equation}
  \label{eq:already3}
  \begin{split}
    | \Pr &\{ F_{\bm S_n,1}^B(S_{1,n}) \geq 1 -\xi, S_{2,n} \leq x \} - \Pr \{ F_{\bm S, 1}(S_{1,n}) \geq 1 - \xi, S_{2,n} \leq x \} | \\
    &= | \Ex [ \1 \{ F_{\bm S_n,1}^B(S_{1,n}) \geq 1 -\xi, S_{2,n} \leq x \} - \1 \{ F_{\bm S, 1}(S_{1,n}) \geq 1 - \xi, S_{2,n} \leq x \} ] | \\
    &\leq \Ex | \1 \{ F_{\bm S_n,1}^B(S_{1,n}) \geq 1 -\xi, S_{2,n} \leq x \} - \1 \{ F_{\bm S, 1}(S_{1,n}) \geq 1 - \xi, S_{2,n} \leq x \} | \\
    &\leq \Ex | \1 \{ F_{\bm S_n,1}^B(S_{1,n}) \geq 1 -\xi \} - \1 \{ F_{\bm S, 1}(S_{1,n}) \geq 1 - \xi \} | \\
    &\leq \Pr\{ |F_{\bm S, 1}(S_{1,n}) - 1 + \xi| \leq \eps \} + \Pr\{ |F_{\bm S_n,1}^B(S_{1,n}) - F_{\bm S, 1}(S_{1,n}) | > \eps \}.
    \end{split}
\end{equation}
By the continuous mapping theorem and the Portmanteau theorem, the first probability on the right converges to $\Pr\{  | F_{S_1}(S_1) - 1 + \xi | \leq \eps\}$ as $n \to \infty$ and can be made arbitrarily small by decreasing~$\eps$. The second probability converges to zero as $n,B \to \infty$ by~\eqref{eq:rec11} for any~$\eps > 0$. Hence,~\eqref{eq:aim1} holds and so does~\eqref{eq:recj1} for $j=2$.

Let us now prove~\eqref{eq:recj2} for $j=2$. Since, from~\eqref{eq:rec12}, $\Pr ( S_{1,n} \leq g_{1,n}^B ) \to P(S_1 \leq g_1) = 1 - \xi$, as $n,B \to \infty$, it remains to show that
$$
\Pr(S_{1,n}  \leq g_{1,n}^B, S_{2,n}  \leq g_{2,n}^B) \to \Pr(S_1  \leq g_1, S_2 \leq g_2)
$$
as $n,B \to \infty$. From the triangle inequality and~\eqref{eq:unif}, if suffices to prove that, as $n,B \to \infty$,
$$
| F_{\bm S_n,\{1,2\}}(g_{1,n}^B, g_{2,n}^B) - F_{\bm S_n,\{1,2\}}(g_1, g_2) | \to 0.
$$
The term on the left-hand side of the previous display is smaller than
$$
|F_{\bm S_n,\{1,2\}}(g_{1,n}^B, g_{2,n}^B) - F_{\bm S_n,\{1,2\}}(g_1, g_{2,n}^B)| + |F_{\bm S_n,\{1,2\}}(g_1, g_{2,n}^B) - F_{\bm S_n,\{1,2\}}(g_1, g_2)|.
$$
The first difference between absolute values is smaller than~\eqref{eq:already1}, which was already shown to converge to zero  as $n,B \to \infty$. Proceeding as in~\eqref{eq:already2}, to show that the second difference between absolute values converges to zero as $n,B \to \infty$, it suffices to prove that
\begin{multline*}
|\Pr ( S_{1,n} \leq g_1, S_{2,n} \geq g_{2,n}^B ) - \Pr(S_{1,n} \leq g_1, S_{2,n} \geq g_2 )| \\ = | \Pr \{ S_{1,n} \leq g_1, S_{2,n} \geq F_{\bm S_n,2}^{B,-1}(1 - \xi) \} - \Pr \{ S_{1,n} \leq g_1, S_{2,n} \geq F_{\bm S,2}^{-1}(1 - \xi) \} | \to 0
\end{multline*}
as $n,B \to \infty$. Using again~\eqref{eq:ineq:ind} and proceeding as in~\eqref{eq:already3}, the probability on the left can be shown to be smaller than
$$
\Pr\{ |F_{\bm S, 2}(S_{2,n}) - 1 + \xi| \leq \eps \} + \Pr\{ |F_{\bm S_n,2}^B(S_{2,n}) - F_{\bm S, 2}(S_{2,n}) | > \eps \}
$$
for any $\eps > 0$. Using the continuity of $F_{\bm S, 2}$ and the Portmanteau theorem, the first probability converges as $n \to \infty$ to a probability that can be made arbitrarily small by decreasing~$\eps$. The second probability converges to zero as $n,B \to \infty$ for any $\eps > 0$ by~\eqref{eq:recj1} for $j=2$, which was proven previously.

It remains to prove~\eqref{eq:recj3} for $j=2$. From the triangle inequality,~\eqref{eq:unif} and~\eqref{eq:aim1}, we have that, as $n,B \to \infty$,
$$
\sup_{x \in \R} |F_{\bm S_n,\{1,2\}}(g_{1,n}^B,x) - F_{\bm S,\{1,2\}}(g_1,x)| \to 0.
$$
Starting from~\eqref{eq:recj2} for $j=2$, the desired result follows from the convergence in the previous display and~\eqref{eq:rec12}.

\emph{Proof of~\eqref{eq:recj1},~\eqref{eq:recj2} and~\eqref{eq:recj3} for all $j \in \{3,\dots,p\}$:} As mentioned previously, we proceed by induction. Let $k \in \{3,\dots,p\}$, assume that~\eqref{eq:recj1},~\eqref{eq:recj2} and~\eqref{eq:recj3} hold for all $j \in \{2,\dots,k-1\}$ and let us show that they also hold for $k$.

For any $x \in \R$, let $F_{\bm S_n, 1}(x) = F_{\bm S_n, \{1\}}(x)$, and, for $j$ successively equal to $2,\dots,p$, let
\begin{equation*}
%  \label{eq:Fjn}
F_{\bm S_n, j}(x) = \frac{F_{\bm S_n, \{1,\dots,j\}}(g_{1,n}^B,\dots,g_{j-1,n}^B,x)}{F_{\bm S_n, \{1,\dots,j-1\}}(g_{1,n}^B,\dots,g_{j-1,n}^B)}.
\end{equation*}
Then, from the induction hypothesis for~\eqref{eq:recj2}, as $n,B \to \infty$,
\begin{multline}
  \label{eq:decomp:F}
F_{\bm S_n,\{1,\dots,k-1\}}(g_{1,n}^B,\dots,g_{k-1}^B) = F_{\bm S_n,k-1}(g_{k-1,n}^B) F_{\bm S_n,k-2}(g_{k-2,n}^B) \cdots F_{\bm S_n,1}(g_{1,n}^B)  \\ \to
 F_{\bm S,k-1}(g_{k-1}) F_{\bm S,k-2}(g_{k-2}) \cdots F_{\bm S,1}(g_1) = (1 - \xi)^{k-1} = F_{\bm S,\{1,\dots,k-1\}}(g_1,\dots,g_{k-1}).
\end{multline}
Thus, from~\eqref{eq:unif:B}, as $n,B \to \infty$, $F_{\bm S_n,\{1,\dots,k-1\}}^B(g_{1,n}^B,\dots,g_{k-1}^B) \p F_{\bm S,\{1,\dots,k-1\}}(g_1,\dots,g_{k-1})$. Hence, to show~\eqref{eq:recj1} for $j=k$, it suffices to prove that, as $n,B \to \infty$,
$$
\sup_{x \in \R} |F_{\bm S_n,\{1,\dots,k\}}^B(g_{1,n}^B,\dots,g_{k-1,n}^B,x) - F_{\bm S,\{1,\dots,k\}}(g_1,\dots,g_{k-1},x)| \p 0.
$$
From the triangle inequality,~\eqref{eq:unif} and~\eqref{eq:unif:B}, the latter will hold if, as $n,B \to \infty$,
\begin{equation}
  \label{eq:aim2}
\sup_{x \in \R} |F_{\bm S_n,\{1,\dots,k\}}(g_{1,n}^B,\dots,g_{k-1,n}^B,x) - F_{\bm S_n,\{1,\dots,k\}}(g_1,\dots,g_{k-1},x)| \to 0.
\end{equation}
For any $x \in \R$, $|F_{\bm S_n,\{1,\dots,k\}}(g_{1,n}^B,\dots,g_{k-1,n}^B,x) - F_{\bm S_n,\{1,\dots,k\}}(g_1,\dots,g_{k-1},x)|$ is smaller than the following sum of $k-1$ terms:
\begin{align*}
  &|F_{\bm S_n,\{1,\dots,k\}}(g_{1,n}^B,g_{2,n}^B,\dots,g_{k-1,n}^B,x) - F_{\bm S_n,\{1,\dots,k\}}(g_1,g_{2,n}^B,\dots,g_{k-1,n}^B,x)| \\
  +& |F_{\bm S_n,\{1,\dots,k\}}(g_1,g_{2,n}^B,\dots,g_{k-1,n}^B,x) - F_{\bm S_n,\{1,\dots,k\}}(g_1,g_2,g_{3,n}^B,\dots,g_{k-1,n}^B,x)| \\
  &\vdots \\
  +& |F_{\bm S_n,\{1,\dots,k\}}(g_1,\dots,g_{k-2},g_{k-1,n}^B,x) - F_{\bm S_n,\{1,\dots,k\}}(g_1,\dots,g_{k-1},x)|.
\end{align*}
Proceeding as in~\eqref{eq:already2},~\eqref{eq:already1} and~\eqref{eq:already3}, the $j$th term, $j \in \{1,\dots,k-1\}$, can be shown to converge to zero as $n,B \to \infty$ as a consequence of the fact that
$$
\Pr\{ |F_{\bm S, j}(S_{j,n}) - 1 + \xi| \leq \eps \} + \Pr\{ |F_{\bm S_n,j}^B(S_{j,n}) - F_{\bm S, j}(S_{j,n}) | > \eps \}
$$
converges to zero as $n,B \to \infty$ followed by $\eps \downarrow 0$ using the continuity of $F_{\bm S, j}$ and the induction hypothesis. Hence,~\eqref{eq:aim2} holds.

Let us now show~\eqref{eq:recj2} for $j=k$. From~\eqref{eq:decomp:F}, it suffices to prove that, as $n,B \to \infty$,
$$
| F_{\bm S_n,\{1,\dots,k\}}(g_{1,n}^B, \dots, g_{k,n}^B ) - F_{\bm S,\{1,\dots,k\}}(g_1, \dots,g_k )| \to 0.
$$
From the triangle inequality and~\eqref{eq:unif}, the latter will hold if, as $n,B \to \infty$,
$$
| F_{\bm S_n,\{1,\dots,k\}}(g_{1,n}^B, \dots, g_{k,n}^B ) - F_{\bm S_n,\{1,\dots,k\}}(g_1, \dots,g_k )| \to 0.
$$
The difference between absolute values on the left of the previous display is smaller than
\begin{multline*}
| F_{\bm S_n,\{1,\dots,k\}}(g_{1,n}^B, \dots, g_{k-1,n}^B, g_{k,n}^B ) - F_{\bm S_n,\{1,\dots,k\}}(g_1, \dots,g_{k-1},g_{k,n}^B)| \\ + | F_{\bm S_n,\{1,\dots,k\}}(g_1, \dots,g_{k-1},g_{k,n}^B) - F_{\bm S_n,\{1,\dots,k\}}(g_1, \dots, g_{k-1}, g_k )|.
\end{multline*}
The first term converges to zero in probability as $n,B \to \infty$ as a consequence of~\eqref{eq:aim2}. Proceeding as in~\eqref{eq:already2},~\eqref{eq:already1} and~\eqref{eq:already3}, the second term can be shown to converge to zero as $n,B \to \infty$ as a consequence of the fact that
$$
\Pr\{ |F_{\bm S, k}(S_{k,n}) - 1 + \xi| \leq \eps \} + \Pr\{ |F_{\bm S_n,k}^B(S_{k,n}) - F_{\bm S, k}(S_{k,n}) | > \eps \}
$$
converges to zero as $n,B \to \infty$ followed by $\eps \downarrow 0$ using the continuity of $F_{\bm S, k}$ and the already proven claim~\eqref{eq:recj1} for $j=k$.

It finally remains to show~\eqref{eq:recj3} for $j=k$. From the triangle inequality,~\eqref{eq:unif} and~\eqref{eq:aim2}, we have that, as $n,B \to \infty$,
$$
\sup_{x \in \R} |F_{\bm S_n,\{1,\dots,k\}}(g_{1,n}^B,\dots,g_{k-1,n}^B,x) - F_{\bm S,\{1,\dots,k\}}(g_1,\dots,g_{k-1},x)| \to 0.
$$
Starting from~\eqref{eq:recj2} for $j=k$, the desired result follows by combining the latter convergence with~\eqref{eq:decomp:F}.

The induction is thus complete. The fact that the statements~\eqref{eq:rec11} and~\eqref{eq:recj1} with `$\leq$' replaced by `$<$' hold as well is a consequence of~\eqref{eq:unif} and~\eqref{eq:unif:sign}.
\end{proof}

\begin{proof}[Proof of Lemma~\ref{lem:thresh:inf}]
Let $M_n = \max_{1 \leq j \leq p} T_{j,n}$. Then,
\begin{align*}
  \Pr(T_{1,n} \leq g_{1,n}^B,\dots,T_{p,n} \leq g_{p,n}^B) &\leq \Pr \left(M_n \leq \max_{1 \leq j \leq p} g_{j,n}^B \right) \\
                                                           &\leq \Pr (M_n \leq g_{1,n}^B ) + \dots + \Pr( M_n \leq g_{p,n}^B ).
\end{align*}
The proof is complete if we show that, for any $j \in \{1,\dots,p\}$, $\Pr ( M_n \leq g_{j,n}^B ) \to 0$ as $n,B \to \infty$. Let $j \in \{1,\dots,p\}$ and recall the definition of the continuous d.f.\ $F_{\bm S,j}$ defined in~\eqref{eq:Fj:gj}. Then,
$$
\Pr( M_n \leq g_{j,n}^B ) \leq \Pr \{ F_{\bm S,j}(M_n) \leq F_{\bm S,j} (g_{j,n}^B ) \} = \Pr \{ F_{\bm S,j}(M_n) - F_{\bm S,j} (g_{j,n}^B ) \leq 0\}.
$$
Since $M_n \p \infty$ as $n \to \infty$, we have that $F_{\bm S,j}(M_n) \p 1$ as $n \to \infty$. Furthermore, from Lemma~\ref{lem:threshold}, $F_{\bm S,j}(g_{j,n}^B)  \to 1 - \xi$ as $n,B \to \infty$. Hence, $F_{\bm S,j}(M_n) - F_{\bm S,j}(g_{j,n}^B) \p \xi > 0$  as $n,B \to \infty$, or, equivalently, for any $\eps > 0$,
$$
\Pr \{ | F_{\bm S,j}(M_n) - F_{\bm S,j}(g_{j,n}^B) - \xi | > \eps) \to 0
$$
as $n,B \to \infty$. Hence, for any $0 < \eps < \xi$,
$$
\Pr ( M_n \leq g_{j,n}^B ) \leq \Pr \{ F_{\bm S,j}(M_n) - F_{\bm S,j} (g_{j,n}^B ) \leq 0\} \leq \Pr \{ | F_{\bm S,j}(M_n) - F_{\bm S,j}(g_{j,n}^B) - \xi | > \eps) \to 0
$$
as $n,B \to \infty$, which completes the proof.
\end{proof}

%%%%%%%%%%%%%%%%%%%%%%%%%%%%%%%%%%%%%%%%%%%%%%%%%%%%%%%%%%%%%%%%%%%%%%%%%%%%%%%%%%%%%%

\section{Proof of Theorem~\ref{thm:thresh:boot}}
\label{sec:proof:thm:thres:boot}

\begin{proof}%[\bf Proof of Theorem~\ref{thm:thresh:boot}]
Let $\bm S_m = \big( \sup_{t \in I_1} \Db_m(t), \dots, \sup_{t \in I_p} \Db_m(t) \big)$ and
$$
\bm S_m^{[b]} = \left( \sup_{t \in I_1} \Db_m^{[b]}(t), \dots, \sup_{t \in I_p} \Db_m^{[b]}(t) \right), \qquad b \in \N,
$$
be the corresponding bootstrap replicates of $\bm S_m$. From~\eqref{eq:boot:val} and the continuous mapping theorem, we immediately obtain that, under $H_0$, $(\bm S_m, \bm S_m^{[1]}, \bm S_m^{[2]}) \leadsto (\bm S, \bm S^{[1]}, \bm S^{[2]})$ in $(\R^p)^3$, where
$$
\bm S = \left( \sup_{t \in I_1} \Db_F(t), \dots, \sup_{t \in I_p} \Db_F(t) \right)
$$
and $\bm S^{[1]}$ and $\bm S^{[2]}$ are independent copies of $\bm S$. Since $\bm S$ is assumed to have a continuous d.f., the assumptions of Appendix~\ref{sec:aux:thresh:generic} are satisfied, and, therefore, Lemma~\ref{lem:threshold} with $n = m$ and $\xi = 1 - (1 - \alpha)^{1/p}$ implies that~\eqref{eq:thresh:val:1} and~\eqref{eq:thresh:val:i} hold. The fact that, under $H_0$, $\Pr(\Db_m \leq \tau_m^B) \to 1 - \alpha$ as $m,B \to \infty$ follows from a decomposition similar to the one used in~\eqref{eq:decomp:cond}.

To prove the last claim, it suffices to show that, when $\sup_{t \in [1,T+1]} \Db_m(t) \p \infty$, $\Pr ( \Db_m \leq \tau_m^B ) \to 0$ as $m, B \to \infty$. Let $T_{i,m} = \sup_{t \in I_i} \Db_m(t)$, $i \in \{1,\dots,p\}$, which implies that $\sup_{t \in [1,T+1]} \Db_m(t) = \max_{1 \leq i \leq p} T_{i,m} \p \infty$. Lemma~\ref{lem:thresh:inf} with $n = m$ then implies that, as $m,B \to \infty$,
$$
\Pr(T_{1,m} \leq g_{1,m}^B,\dots,T_{p,m} \leq g_{p,m}^B) = \Pr ( \Db_m \leq \tau_m^B ) \to 0,
$$
which completes the proof.
\end{proof}

%%%%%%%%%%%%%%%%%%%%%%%%%%%%%%%%%%%%%%%%%%%%%%%%%%%%%%%%%%%%%%%%%%%%%%%%%%%%%%%%%%%%%%

\section{Proof of Proposition~\ref{prop:mult}}
\label{sec:proof:thresh:estim}

\begin{proof}%[\bf Proof of Proposition~\ref{prop:mult}]
Recall that $n = \ip{m(T+1)}$ and, for any $b \in \N, s \in [0,1]$ and $\bm u \in [0,1]^d$, let
\begin{equation*}
  \label{eq:tilde:Bbmb}
  \tilde \Bb_n^{[b]}(s, \bm u) = \frac{1}{\sqrt{n}} \sum_{i=1}^{\ip{ns}} \xi_{i,m}^{[b]} \{ \1(\bm U_i \leq \bm u) - C_{1:n}(\bm u) \}.
\end{equation*}
From Corollary~2.2 of~\cite{BucKoj16}, we have that
$$
(\Bb_n, \tilde \Bb_n^{[1]}, \tilde \Bb_n^{[2]}) \leadsto (\Bb_C, \Bb_C^{[1]}, \Bb_C^{[2]})
$$
in  $\{ \ell^\infty([0,1]^{d+1}) \}^3$. Proceeding as at the beginning of the proof of Proposition~\ref{prop:H0}, from the continuous mapping theorem and with the map $\psi$ defined in~\eqref{eq:psi}, we obtain that
$$
\big( \psi(\Bb_n), \psi(\tilde \Bb_n^{[1]}), \psi(\tilde \Bb_n^{[2]}) \big) \leadsto (\Bb_C, \Bb_C^{[1]}, \Bb_C^{[2]})
$$
in $\{ \ell^\infty([0,T+1] \times [0,1]^d) \}^3$. Recall the definition of $\check \Bb_m^{[b]}$ in \eqref{eq:check:Bbmb}. Since $\psi(\Bb_n) - \Bb_m = o_\Pr(1)$ and $\psi(\tilde \Bb_n^{[b]}) - \check \Bb_m^{[b]} = o_\Pr(1)$, $b \in \{1,2\}$, it follows that
\begin{equation}
  \label{eq:inter:jwc}
( \Bb_m, \check \Bb_m^{[1]}, \check \Bb_m^{[2]}) \leadsto (\Bb_C, \Bb_C^{[1]}, \Bb_C^{[2]})
\end{equation}
in $\{ \ell^\infty([0,T+1] \times [0,1]^d) \}^3$. Let $\phi$ be the continuous map from $\ell^\infty([0,T+1] \times [0,1]^d)$ to $\ell^\infty([0,T+1] \times [0,1]^d)$ defined, for any $f \in \ell^\infty([0,T+1] \times [0,1]^d)$, by
$$
\phi(f)(s,\bm u) = \sqrt{T+1} f \{ s/(T+1), \bm u \}, \qquad s \in [0,T+1], \bm u \in [0,1]^d.
$$
Furthermore, for any $b \in \N$, $s \in [0,T+1]$ and $\bm u \in [0,1]^d$, let
$$
\bar \Bb_m^{[b]}(s, \bm u) = \frac{1}{\sqrt{m'}} \sum_{i=1}^{\ip{m' s}}  \xi_{i,m}^{[b]} \{ \1(\bm U_i \leq \bm u) - C_{1:n}(\bm u) \}.
$$
Then, starting from~\eqref{eq:inter:jwc} and using the continuous mapping theorem with the map $\phi$ as well as the fact that $\phi(\check \Bb_m^{[b]}) - \bar \Bb_m^{[b]} = o_\Pr(1)$ and $\phi(\Bb_C^{[b]})$ and $\Bb_C^{[b]}$ have the same distribution, $b \in \{1,2\}$, we obtain that
$$
( \Bb_m, \bar \Bb_m^{[1]}, \bar \Bb_m^{[2]}) \leadsto (\Bb_C, \Bb_C^{[1]}, \Bb_C^{[2]})
$$
in $\{ \ell^\infty([0,T+1] \times [0,1]^d) \}^3$.

Fix $b \in \{1,2\}$. To prove the first claim, it remains to show that
\begin{equation}
  \label{eq:goal}
  I_m = \sup_{s \in [0,T+1]} \sup_{\bm u \in [0,1]^d} |\hat \Bb_m^{[b]}(s, \bm u) - \bar \Bb_m^{[b]}(s, \bm u) | \p 0,
\end{equation}
where $\hat \Bb_m^{[b]}$ is defined in~\eqref{eq:hat:Bbmb}. Let $J_m = \sup_{\bm u \in [0,1]^d} \sqrt{m'} |C_{1:m}(\bm u) - C_{1:n}(\bm u) |$ and
$$
K_m = \sup_{s \in [0,T+1]} \frac{1}{m'} | \sum_{i=1}^{\ip{m' s}}  \xi_{i,m}^{[b]} |,
$$
and notice that $I_m \leq J_m \times K_m$. On one hand,
\begin{align*}
  J_m &\leq \sup_{\bm u \in [0,1]^d} \sqrt{m'} |C_{1:m}(\bm u) - C(\bm u)| + \sup_{\bm u \in [0,1]^d} \sqrt{m'} |C_{1:n}(\bm u) - C(\bm u)| \\
  &= O(1) \times \sup_{s \in [0,T+1]} \sup_{\bm u \in [0,1]^d} |\Bb_m(s, \bm u)| = O_\Pr(1).
\end{align*}
On the other hand, we have
$$
K_m = O(1) \times m^{-1} \max_{1 \leq k \leq m}  \left| \sum_{i=1}^{k}  \xi_{i,m}^{[b]} \right|.
$$
Let $Z_{k,m} = \sum_{i=1}^k \xi_{i,m}^{[b]}$, $k \in \{1,\dots,m\}$, and let $\nu > 1$. Then,
\begin{align*}
\Ex \left\{ \left( m^{-1}\max_{1 \leq k \leq m} | Z_{k,m} | \right)^\nu \right\} &= m^{-\nu} \Ex \left( \max_{1 \leq k \leq m} | Z_{k,m} |^\nu \right) \\ &\leq m^{-\nu} \Ex \left( \sum_{k = 1}^m | Z_{k,m} |^\nu \right) \leq m^{1-\nu} \max_{1 \leq k \leq m} \Ex \left( | Z_{k,m} |^\nu \right).
\end{align*}
Proceeding as in the proof of Lemma~D.2 in the supplementary material of \cite{BucKoj16}, for $\nu=4$, using the fact that the sequence $(\xi_{i,m}^{[b]})_{i \in \Z}$ is stationary and $\ell_m$-dependent,
\begin{align*}
  \Ex \left( | Z_{k,m} |^4 \right) = \Ex \left( Z_{k,m}^4 \right) &=  \sum_{i_1,i_2,i_3,i_4=1}^k \Ex \left(\xi_{i_1,m}^{[b]} \xi_{i_2,m}^{[b]} \xi_{i_3,m}^{[b]}  \xi_{i_4,m}^{[b]} \right) \\
                                   &\leq 4! \sum_{1 \leq i_1 \leq i_2 \leq i_3 \leq i_4 \leq m} \left| \Ex \left( \xi_{i_1,m}^{[b]} \xi_{i_2,m}^{[b]} \xi_{i_3,m}^{[b]}  \xi_{i_4,m}^{[b]} \right) \right| \\
                                  &\leq 4! \, m \sum_{1 \leq i_2 \leq i_3 \leq i_4 \leq m} \left| \Ex \left( \xi_{1,m}^{[b]} \xi_{i_2,m}^{[b]} \xi_{i_3,m}^{[b]}  \xi_{i4,m}^{[b]} \right) \right| \\
                                  &\leq 4! \, m \sum_{i_2=1}^{\ell_m} \sum_{i_3=i_2}^m \sum_{i_4=i_3}^{\ell_m} \left| \Ex \left( \xi_{1,m}^{[b]} \xi_{i_2,m}^{[b]} \xi_{i_3,m}^{[b]}  \xi_{i_4,m}^{[b]} \right) \right|.
\end{align*}
From Cauchy-Swcharz's inequality, there holds $| \Ex ( \xi_{1,m}^{[b]} \xi_{i_2,m}^{[b]} \xi_{i_3,m}^{[b]}  \xi_{i_4,m}^{[b]} ) | \leq \Ex \{ ( \xi_{1,m}^{[b]} )^4 \}$, and therefore $\max_{1 \leq k \leq m} \Ex \left( | Z_{k,m} |^4 \right) = O(m^2 \ell_m^2)$ since $\sup_{m \geq 1} \Ex \{ ( \xi_{1,m}^{[b]} )^4 \} < \infty$. It follows that
$$
\Ex \left\{ \left( m^{-1}\max_{1 \leq k \leq m} | Z_{k,m} | \right)^4 \right\} = O(m^{-2\eps}) \to 0,
$$
which implies that $K_m \p 0$ and therefore that~\eqref{eq:goal} holds. The proof of the first claim is thus complete. The remaining claims are essentially consequences of the continuous mapping theorem and can be shown by proceeding as in the proof of Proposition~\ref{prop:H0}.
\end{proof}

\end{appendix}

%%%%%%%%%%%%%%%%%%%%%%%%%%%%%%%%%%%%%%%%%%%%%%
%% Support information (funding), if any,   %%
%% should be provided in the                %%
%% Acknowledgments section.                %%
%%%%%%%%%%%%%%%%%%%%%%%%%%%%%%%%%%%%%%%%%%%%%%

\section*{Acknowledgments}

The authors would like to thank an anonymous Referee for his/her constructive comments on an earlier version of this manuscript.

%%  MR numbers will be added by VTeX.                      %%
%%                                                         %%
%%  Use \cite{...} to cite references in text.             %%
%%                                                         %%
%%%%%%%%%%%%%%%%%%%%%%%%%%%%%%%%%%%%%%%%%%%%%%%%%%%%%%%%%%%%%

%% if your bibliography is in bibtex format, uncomment commands:
\bibliographystyle{imsart-nameyear}
\bibliography{biblio}

\end{document}

%% file: H0sim.tex
% latex table generated in R 4.0.3 by xtable 1.8-4 package
% Mon Oct 19 17:24:44 2020
\begin{table}[t!]
\centering
\caption{Percentages of rejection of $H_0$ in~\eqref{eq:H0} for the sequential change-point detection procedures when based on the Monte Carlo estimation method with $M=10^5$ as described in Sections~\ref{sec:MC:estim} and~\ref{sec:thresh:ind:univ} for $p \in \{1,2,4,10,50\}$, $m \in \{50,100\}$ and $n=2m$. The rejection percentages are computed from $10^4$ samples of size $n=2m$ generated from the standard uniform distribution.} 
\label{tab:H0:sim}
\begingroup\footnotesize
\begin{tabular}{rrccccccccccc}
  \hline
  \multicolumn{2}{c}{} & \multicolumn{3}{c}{$T_{m,q}$} & \multicolumn{3}{c}{$S_{m,q}$} & \multicolumn{3}{c}{$R_{m,q}$} & \multicolumn{2}{c}{} \\ \cmidrule(lr){3-5} \cmidrule(lr){6-8} \cmidrule(lr){9-11}  $p$ & $m$ & $\gamma$=0 & $\gamma$=0.25 & $\gamma$=0.5 & $\gamma$=0 & $\gamma$=0.25 & $\gamma$=0.5 & $\gamma$=0 & $\gamma$=0.25 & $\gamma$=0.5 & $Q_m$ & $P_m$  \\ \hline
1 & 50 & 5.2 & 5.2 & 5.1 & 4.9 & 5.0 & 4.9 & 4.7 & 4.9 & 4.7 & 5.2 & 5.2 \\ 
   & 100 & 4.9 & 4.9 & 4.6 & 4.8 & 4.8 & 5.0 & 4.9 & 4.8 & 4.7 & 5.1 & 5.0 \\ 
  2 & 50 & 4.9 & 5.1 & 5.0 & 4.8 & 5.2 & 5.1 & 4.9 & 4.9 & 4.9 & 5.1 & 4.4 \\ 
   & 100 & 4.9 & 4.8 & 4.9 & 4.9 & 4.7 & 4.9 & 4.9 & 4.8 & 5.1 & 5.0 & 4.6 \\ 
  4 & 50 & 4.9 & 4.9 & 5.1 & 4.6 & 4.9 & 5.3 & 4.6 & 4.9 & 5.0 & 5.1 & 5.2 \\ 
   & 100 & 5.0 & 5.0 & 5.0 & 4.8 & 4.9 & 5.3 & 4.9 & 4.9 & 4.9 & 5.0 & 4.9 \\ 
  10 & 50 & 5.2 & 5.1 & 5.0 & 4.9 & 4.9 & 5.2 & 4.6 & 4.7 & 5.0 & 5.0 & 4.8 \\ 
   & 100 & 5.0 & 5.1 & 5.1 & 4.9 & 4.9 & 5.1 & 5.0 & 5.1 & 5.0 & 4.9 & 4.6 \\ 
  50 & 50 & 5.0 & 5.1 & 5.1 & 4.8 & 4.9 & 5.1 & 4.3 & 4.6 & 4.9 & 4.9 & 4.5 \\ 
   & 100 & 5.0 & 4.9 & 5.1 & 4.8 & 4.9 & 5.0 & 5.0 & 4.9 & 4.8 & 4.9 & 4.7 \\ 
   \hline
\end{tabular}
\endgroup
\end{table}

%% file: H0multARTmq.tex
% latex table generated in R 3.6.2 by xtable 1.8-4 package
% Fri Apr 17 14:35:17 2020
\begin{table}[t!]
\centering
\caption{Percentages of rejection of $H_0$ in~\eqref{eq:H0} for the procedure based on $T_{m,q}$ with $\gamma \in \{0,0.25,0.5\}$ and $p \in \{1,2,4\}$ when based on the dependent multliplier bootstrap with $B=2000$. The rejection percentages are computed from 1000 samples of size $\ip{m(T+1)}$ generated from AR(1) models with autoregressive parameter $\beta \in \{0, 0.1, 0.3, 0.5 \}$.} 
\label{tab:H0:multARTmq}
\begingroup\footnotesize
\begin{tabular}{rrrrrrrrrrrrrrr}
  \hline
  \multicolumn{3}{c}{} & \multicolumn{3}{c}{$T=0.5$ ($n = 1.5m$)} & \multicolumn{3}{c}{$T=1$ ($n = 2m$)} & \multicolumn{3}{c}{$T=2$ ($n = 3m$)} & \multicolumn{3}{c}{$T=3$ ($n = 4m$)} \\ \cmidrule(lr){4-6} \cmidrule(lr){7-9} \cmidrule(lr){10-12} \cmidrule(lr){13-15}  $\beta$ & $\gamma$ & $m$ & $p$=1 & $p$=2 & $p$=4 & $p$=1 & $p$=2 & $p$=4 & $p$=1 & $p$=2 & $p$=4 & $p$=1 & $p$=2 & $p$=4 \\ \hline
0.0 & 0.00 & 50 & 4.4 & 6.0 & 5.4 & 6.6 & 5.9 & 5.8 & 3.6 & 3.7 & 2.9 & 5.2 & 4.9 & 4.4 \\ 
   &  & 100 & 5.6 & 5.2 & 4.9 & 7.3 & 7.6 & 6.4 & 6.2 & 6.1 & 6.1 & 6.2 & 5.8 & 7.2 \\ 
   &  & 200 & 6.1 & 5.2 & 5.5 & 5.1 & 5.1 & 6.1 & 5.1 & 5.3 & 6.2 & 6.9 & 6.2 & 5.2 \\ 
   &  & 400 & 5.5 & 5.9 & 5.3 & 7.5 & 6.9 & 7.1 &  &  &  &  &  &  \\ 
   & 0.25 & 50 & 4.1 & 5.0 & 4.6 & 5.5 & 4.9 & 4.9 & 3.4 & 4.1 & 3.0 & 6.1 & 5.1 & 4.2 \\ 
   &  & 100 & 5.0 & 4.6 & 4.3 & 6.5 & 6.2 & 6.1 & 6.2 & 5.6 & 5.8 & 6.2 & 5.8 & 6.7 \\ 
   &  & 200 & 5.5 & 5.5 & 5.0 & 5.0 & 4.3 & 5.4 & 5.0 & 5.3 & 5.9 & 6.7 & 5.9 & 5.3 \\ 
   &  & 400 & 4.8 & 5.9 & 4.9 & 7.3 & 6.7 & 6.2 &  &  &  &  &  &  \\ 
   & 0.50 & 50 & 2.8 & 3.1 & 2.7 & 3.6 & 2.5 & 2.9 & 2.9 & 2.5 & 3.0 & 5.1 & 3.2 & 3.3 \\ 
   &  & 100 & 3.7 & 2.7 & 3.2 & 4.6 & 4.3 & 4.6 & 5.3 & 4.2 & 4.3 & 5.1 & 4.9 & 4.9 \\ 
   &  & 200 & 3.9 & 3.0 & 3.1 & 3.6 & 2.6 & 3.4 & 4.2 & 4.7 & 4.6 & 5.9 & 4.8 & 3.6 \\ 
   &  & 400 & 3.7 & 3.6 & 3.2 & 6.2 & 5.7 & 4.6 &  &  &  &  &  &  \\ 
  0.1 & 0.00 & 50 & 4.3 & 4.8 & 6.7 & 7.3 & 8.1 & 9.2 & 4.3 & 4.9 & 3.8 & 4.6 & 5.6 & 4.3 \\ 
   &  & 100 & 5.4 & 6.2 & 7.4 & 6.5 & 7.8 & 7.6 & 6.2 & 5.8 & 6.1 & 7.9 & 6.1 & 7.5 \\ 
   &  & 200 & 6.5 & 7.0 & 7.5 & 6.2 & 6.9 & 7.9 & 6.3 & 6.7 & 6.4 & 7.4 & 8.4 & 8.6 \\ 
   &  & 400 & 6.1 & 5.7 & 7.1 & 5.3 & 5.7 & 6.4 &  &  &  &  &  &  \\ 
   & 0.25 & 50 & 3.6 & 4.0 & 4.7 & 7.2 & 7.1 & 7.2 & 4.4 & 5.4 & 4.3 & 5.0 & 5.9 & 4.4 \\ 
   &  & 100 & 5.4 & 5.6 & 6.1 & 5.5 & 6.9 & 6.5 & 6.0 & 5.2 & 4.9 & 7.7 & 5.8 & 7.0 \\ 
   &  & 200 & 5.9 & 6.0 & 6.5 & 5.7 & 7.2 & 7.2 & 6.4 & 6.9 & 6.6 & 7.5 & 8.2 & 8.6 \\ 
   &  & 400 & 5.8 & 6.0 & 6.6 & 5.6 & 5.0 & 6.3 &  &  &  &  &  &  \\ 
   & 0.50 & 50 & 2.7 & 2.5 & 2.9 & 5.0 & 5.0 & 4.9 & 4.0 & 4.4 & 4.5 & 4.1 & 4.6 & 3.5 \\ 
   &  & 100 & 3.6 & 3.4 & 3.8 & 4.9 & 4.9 & 4.6 & 4.9 & 3.9 & 3.5 & 6.6 & 4.7 & 5.5 \\ 
   &  & 200 & 3.7 & 3.7 & 3.9 & 5.3 & 5.2 & 5.1 & 5.8 & 5.8 & 5.8 & 7.0 & 6.9 & 6.5 \\ 
   &  & 400 & 4.8 & 4.7 & 4.2 & 4.5 & 4.0 & 4.7 &  &  &  &  &  &  \\ 
  0.3 & 0.00 & 50 & 9.3 & 9.3 & 10.0 & 10.8 & 11.7 & 12.9 & 6.4 & 6.5 & 6.5 & 5.9 & 6.9 & 6.2 \\ 
   &  & 100 & 8.4 & 10.1 & 10.2 & 8.3 & 9.5 & 11.1 & 7.4 & 7.6 & 8.3 & 7.7 & 8.8 & 10.2 \\ 
   &  & 200 & 7.4 & 7.7 & 8.5 & 7.0 & 8.0 & 10.1 & 7.2 & 6.8 & 8.9 & 7.1 & 8.6 & 8.9 \\ 
   &  & 400 & 5.7 & 7.7 & 8.4 & 7.0 & 8.1 & 8.5 &  &  &  &  &  &  \\ 
   & 0.25 & 50 & 8.9 & 8.4 & 9.0 & 10.2 & 11.0 & 12.0 & 7.1 & 7.2 & 7.1 & 6.3 & 7.1 & 7.8 \\ 
   &  & 100 & 8.3 & 9.3 & 9.8 & 8.3 & 9.5 & 10.3 & 7.5 & 7.7 & 8.1 & 7.8 & 9.2 & 10.3 \\ 
   &  & 200 & 7.1 & 6.7 & 7.6 & 7.3 & 8.0 & 10.2 & 7.3 & 6.6 & 9.0 & 7.2 & 8.8 & 9.6 \\ 
   &  & 400 & 5.9 & 6.8 & 7.6 & 7.1 & 7.9 & 7.8 &  &  &  &  &  &  \\ 
   & 0.50 & 50 & 5.7 & 6.0 & 6.5 & 9.3 & 9.1 & 10.2 & 7.4 & 6.5 & 6.6 & 6.5 & 7.6 & 8.0 \\ 
   &  & 100 & 6.1 & 6.4 & 8.4 & 7.7 & 6.2 & 7.7 & 6.7 & 7.5 & 7.2 & 6.6 & 7.8 & 8.1 \\ 
   &  & 200 & 5.0 & 4.9 & 5.1 & 6.3 & 6.4 & 8.1 & 6.4 & 5.3 & 6.9 & 6.0 & 6.2 & 7.9 \\ 
   &  & 400 & 4.2 & 5.3 & 4.6 & 5.9 & 5.9 & 5.5 &  &  &  &  &  &  \\ 
  0.5 & 0.00 & 50 & 9.6 & 11.7 & 13.5 & 12.8 & 13.2 & 17.3 & 9.8 & 9.8 & 9.1 & 8.5 & 11.7 & 11.7 \\ 
   &  & 100 & 9.6 & 12.2 & 13.4 & 11.5 & 14.0 & 14.7 & 9.7 & 9.8 & 13.4 & 11.4 & 12.3 & 14.0 \\ 
   &  & 200 & 8.6 & 10.0 & 12.4 & 7.5 & 8.7 & 10.9 & 8.2 & 10.5 & 11.1 & 8.4 & 9.1 & 10.5 \\ 
   &  & 400 & 7.5 & 8.9 & 11.3 & 5.7 & 7.5 & 9.4 &  &  &  &  &  &  \\ 
   & 0.25 & 50 & 9.6 & 11.5 & 12.4 & 12.5 & 13.3 & 16.6 & 10.3 & 10.8 & 10.7 & 9.6 & 12.7 & 13.1 \\ 
   &  & 100 & 9.2 & 11.0 & 12.2 & 11.3 & 13.4 & 14.0 & 10.2 & 10.2 & 13.4 & 11.5 & 12.3 & 14.3 \\ 
   &  & 200 & 7.9 & 9.7 & 11.1 & 6.8 & 7.8 & 10.1 & 8.7 & 10.4 & 11.5 & 8.4 & 9.4 & 10.8 \\ 
   &  & 400 & 6.7 & 8.0 & 10.6 & 5.6 & 7.7 & 8.8 &  &  &  &  &  &  \\ 
   & 0.50 & 50 & 7.7 & 9.6 & 10.2 & 10.9 & 10.9 & 15.1 & 9.7 & 10.3 & 10.5 & 9.7 & 12.3 & 13.5 \\ 
   &  & 100 & 6.8 & 7.6 & 9.8 & 9.5 & 9.8 & 11.4 & 8.2 & 9.7 & 12.4 & 10.2 & 11.1 & 12.0 \\ 
   &  & 200 & 5.8 & 6.8 & 8.7 & 6.1 & 5.8 & 7.1 & 7.5 & 8.5 & 9.7 & 7.6 & 8.6 & 9.7 \\ 
   &  & 400 & 5.2 & 6.2 & 7.4 & 4.8 & 5.3 & 6.5 &  &  &  &  &  &  \\ 
   \hline
\end{tabular}
\endgroup
\end{table}

%% file: H0multGARCHTmq.tex
% latex table generated in R 4.0.3 by xtable 1.8-4 package
% Tue Oct 20 15:59:22 2020
\begin{table}[t!]
\centering
\caption{Percentages of rejection of $H_0$ in~\eqref{eq:H0} for the procedure based on $T_{m,q}$ with $\gamma \in \{0,0.25,0.5\}$ and $p \in \{1,2,4\}$ when based on the dependent multliplier bootstrap with $B=2000$. The rejection percentages are computed from 1000 samples of size $\ip{m(T+1)}$ generated from a GARCH(1, 1) model with parameters $\omega = 0.012$, $\beta = 0.919$ and $\alpha = 0.072$ to mimick SP500 daily logreturns following \cite{JonPooRoc07}.} 
\label{tab:H0:multGARCHTmq}
\begingroup\footnotesize
\begin{tabular}{rrrrrrrrrrrrrr}
  \hline
  \multicolumn{2}{c}{} & \multicolumn{3}{c}{$T=0.5$ ($n = 1.5m$)} & \multicolumn{3}{c}{$T=1$ ($n = 2m$)} & \multicolumn{3}{c}{$T=2$ ($n = 3m$)} & \multicolumn{3}{c}{$T=3$ ($n = 4m$)} \\ \cmidrule(lr){3-5} \cmidrule(lr){6-8} \cmidrule(lr){9-11} \cmidrule(lr){12-14}  $\gamma$ & $m$ & $p$=1 & $p$=2 & $p$=4 & $p$=1 & $p$=2 & $p$=4 & $p$=1 & $p$=2 & $p$=4 & $p$=1 & $p$=2 & $p$=4 \\ \hline
0.00 & 50 & 5.4 & 5.4 & 4.9 & 9.1 & 7.6 & 7.6 & 4.8 & 4.7 & 4.0 & 5.4 & 6.3 & 5.3 \\ 
   & 100 & 6.6 & 7.5 & 7.6 & 7.9 & 8.6 & 7.2 & 9.6 & 8.4 & 9.2 & 10.5 & 10.3 & 10.0 \\ 
   & 200 & 7.4 & 7.2 & 8.6 & 8.3 & 8.4 & 8.5 & 6.6 & 7.7 & 6.1 & 8.7 & 8.0 & 8.2 \\ 
   & 400 & 6.7 & 7.5 & 8.4 & 7.0 & 7.0 & 8.4 &  &  &  &  &  &  \\ 
  0.25 & 50 & 4.7 & 4.2 & 4.3 & 7.5 & 5.9 & 5.6 & 4.8 & 4.5 & 4.2 & 5.8 & 6.4 & 5.9 \\ 
   & 100 & 5.9 & 6.3 & 6.3 & 7.2 & 6.9 & 6.7 & 9.6 & 8.8 & 8.4 & 10.2 & 9.8 & 9.2 \\ 
   & 200 & 7.0 & 6.9 & 7.8 & 7.2 & 7.3 & 8.2 & 6.4 & 6.9 & 5.9 & 8.3 & 7.6 & 8.4 \\ 
   & 400 & 6.2 & 7.9 & 7.9 & 6.9 & 6.7 & 8.4 &  &  &  &  &  &  \\ 
  0.50 & 50 & 2.4 & 2.5 & 2.6 & 4.4 & 3.9 & 4.2 & 4.0 & 3.9 & 3.5 & 5.0 & 4.7 & 4.6 \\ 
   & 100 & 4.1 & 4.2 & 4.5 & 5.0 & 4.1 & 3.9 & 6.6 & 5.9 & 6.0 & 7.9 & 7.5 & 7.0 \\ 
   & 200 & 4.9 & 5.0 & 5.2 & 5.7 & 4.8 & 5.7 & 5.8 & 5.1 & 4.9 & 6.8 & 6.6 & 7.1 \\ 
   & 400 & 5.1 & 5.4 & 5.1 & 6.3 & 5.9 & 5.9 &  &  &  &  &  &  \\ 
   \hline
\end{tabular}
\endgroup
\end{table}

%% file: H0multcop.tex
% latex table generated in R 3.6.2 by xtable 1.8-4 package
% Mon Apr 20 20:19:38 2020
\begin{table}[t!]
\centering
\caption{Percentages of rejection of $H_0$ in~\eqref{eq:H0} for the procedure based on $T_{m,q}$ with $\gamma \in \{0,0.25,0.5\}$ and $p \in \{1,2,4\}$ when based on the dependent multliplier bootstrap with $B=2000$. The rejection percentages are computed from 1000 samples of size $2m$ generated from a normal copula with a Kendall's tau of $\tau \in \{-0.6, -0.3, 0, 0.3, 0.6 \}$.} 
\label{tab:H0:multcop}
\begingroup\footnotesize
\begin{tabular}{rrrrrrrrrrr}
  \hline
  \multicolumn{2}{c}{} & \multicolumn{3}{c}{$\gamma = 0$} & \multicolumn{3}{c}{$\gamma=0.25$} & \multicolumn{3}{c}{$\gamma=0.5$}  \\ \cmidrule(lr){3-5} \cmidrule(lr){6-8} \cmidrule(lr){9-11} $\tau$ & $m$ & $p$=1 & $p$=2 & $p$=4 & $p$=1 & $p$=2 & $p$=4 & $p$=1 & $p$=2 & $p$=4  \\ \hline
-0.6 & 50 & 3.1 & 1.9 & 1.5 & 2.4 & 1.1 & 1.1 & 0.7 & 0.3 & 0.8 \\ 
   & 100 & 3.5 & 2.9 & 2.2 & 2.5 & 2.1 & 1.7 & 1.2 & 1.0 & 0.9 \\ 
   & 200 & 4.0 & 3.2 & 3.1 & 3.0 & 2.4 & 2.1 & 0.9 & 1.0 & 1.0 \\ 
  -0.3 & 50 & 4.5 & 2.9 & 2.7 & 3.4 & 2.5 & 2.5 & 1.8 & 1.5 & 1.2 \\ 
   & 100 & 4.0 & 4.0 & 3.8 & 3.1 & 3.0 & 2.9 & 1.5 & 0.9 & 1.4 \\ 
   & 200 & 4.8 & 4.1 & 4.0 & 4.1 & 3.3 & 2.9 & 1.6 & 1.6 & 1.6 \\ 
  0.0 & 50 & 5.1 & 4.3 & 4.7 & 4.2 & 3.9 & 3.6 & 2.7 & 2.0 & 2.3 \\ 
   & 100 & 7.0 & 5.8 & 4.7 & 5.7 & 4.7 & 3.7 & 2.6 & 2.2 & 2.0 \\ 
   & 200 & 6.0 & 6.2 & 5.8 & 5.3 & 5.0 & 5.2 & 3.3 & 2.4 & 2.3 \\ 
  0.3 & 50 & 6.1 & 4.5 & 4.5 & 4.7 & 3.8 & 3.3 & 2.5 & 2.4 & 2.6 \\ 
   & 100 & 5.0 & 5.5 & 4.9 & 4.4 & 4.8 & 3.9 & 2.3 & 2.5 & 2.5 \\ 
   & 200 & 6.2 & 6.3 & 7.8 & 6.0 & 4.8 & 6.3 & 3.4 & 2.6 & 3.8 \\ 
  0.6 & 50 & 6.7 & 5.5 & 5.8 & 5.8 & 3.9 & 4.2 & 3.2 & 2.6 & 2.5 \\ 
   & 100 & 7.7 & 7.6 & 7.8 & 7.0 & 6.6 & 7.0 & 4.3 & 4.0 & 4.8 \\ 
   & 200 & 5.4 & 6.2 & 7.2 & 5.1 & 5.6 & 6.1 & 3.4 & 3.5 & 3.2 \\ 
   \hline
\end{tabular}
\endgroup
\end{table}